\documentclass{article}
\usepackage[top=3cm, bottom=3cm, left=2cm, right=2cm]{geometry}
\usepackage{graphicx} 
\usepackage{tabularx}
\usepackage{amsmath, amssymb, amsthm}
\usepackage{hyperref}
\usepackage{booktabs}
\usepackage{pifont}
\usepackage{xcolor}
\usepackage{xspace}
\usepackage{threeparttable}
\usepackage{tabularx}
\usepackage{verbatim}
\usepackage{cite}
\usepackage{tikz}
\usepackage{cleveref}
\usepackage{authblk}
\usepackage{bm}

\newcommand{\etal}{\textit{et al.\@} }
\newcommand{\etalx}{\textit{et al.\@}\xspace}
\newcommand{\ts}[1]{\ensuremath{_{\text{#1}}}}
\newcommand{\V}{{\color{green}\ding{51}}}
\newcommand{\X}{{\color{red}\ding{55}}}

\newcommand{\Vd}{\V$^\dagger$}
\newcommand{\Vs}{\,\,\,\V$^*$}
\newcommand{\Xdd}{\X$^\ddagger$}

\newcommand{\Vm}{\V$^\mathsection$}

\newcommand{\na}{{\color{blue}n.a.}}
\newcommand{\BB}{($\boldsymbol{\ast}$)}
\newcommand{\ind}{\mathrel{\perp\!\!\!\perp}}
\newcommand{\notind}{\mathrel{\not\!\perp\!\!\!\perp}}
\newcommand{\Xs}{X_1,\ldots,X_n}
\newcommand{\Xsm}{X_1,\ldots,X_{n-1}}
\newcommand{\XsY}{X_1,\ldots,X_n;Y}
\newcommand{\Ic}{\ensuremath{I_\cap}\xspace}
\newcommand{\Icap}{\ensuremath{I_\cap(X_1,X_2;Y)}}
\newcommand{\Icaps}{\ensuremath{I_\cap(\Xs;Y)}}
\newcommand{\Icapsm}{\ensuremath{I_\cap(X_1,\ldots,X_{n-1};Y)}}
\newcommand{\Ip}{\ensuremath{I_\partial}\xspace}

\newcommand{\cxxy}{\ensuremath{(X_1;X_2;Y)}}
\newcommand{\argct}[3]{\{#1\}\{#2\}\{#3\}}

\newcommand{\SR}{\textbf{(SR)}\xspace} %
\newcommand{\Sz}{\textbf{(S\ts{0})}\xspace} %
\newcommand{\Mz}{\textbf{(M\ts{0})}\xspace} %
\newcommand{\GP}{\textbf{(GP)}\xspace} %
\newcommand{\LPz}{\textbf{(LP\ts{0})}\xspace} %
\newcommand{\ID}{\textbf{(ID)}\xspace} %
\newcommand{\IID}{\textbf{(IID)}\xspace} %
\newcommand{\TM}{\textbf{(TM)}\xspace} %
\newcommand{\TC}{\textbf{(TC)}\xspace} %
\newcommand{\So}{\textbf{(S\ts{1})}\xspace} %
\newcommand{\Mo}{\textbf{(M\ts{1})}\xspace} %
\newcommand{\LPo}{\textbf{(LP\ts{1})}\xspace} %
\newcommand{\BP}{\textbf{(BP)}\xspace} %
\newcommand{\TE}{\textbf{(TE)}\xspace} %
\newcommand{\AST}{\textbf{($\boldsymbol{\ast}$)}\xspace} %
\newcommand{\AD}{\textbf{(AD)}\xspace} %
\newcommand{\CO}{\textbf{(CO)}\xspace} %
\newcommand{\LB}{\textbf{(LB)}\xspace} %
\newcommand{\EI}{\textbf{(EI)}\xspace} %
\newcommand{\SE}{\textbf{(SE)}\xspace} %
\newcommand{\LP}{\textbf{(LP)}\xspace}
\newcommand{\DI}{\textbf{(DI)}\xspace}
\newcommand{\DIp}{\textbf{(DI)}\small$^\|$\xspace}
\newcommand{\COp}{\textbf{(CO)}\small$^\|$\xspace}

\newcommand{\Imin}{\ensuremath{I_\cap^{\text{min}}}\xspace}
\newcommand{\Ired}{\ensuremath{I^{\text{red}}_\cap}\xspace}
\newcommand{\Iwedge}{\ensuremath{I_\cap^{\wedge}}\xspace}
\newcommand{\IBROJA}{\ensuremath{I_\cap^{\text{BROJA}}}\xspace}
\newcommand{\IMMI}{\ensuremath{I_\cap^{\text{MMI}}}\xspace}
\newcommand{\IPM}{\ensuremath{I_\cap^{\text{PM}}}\xspace}
\newcommand{\IRR}{\ensuremath{I_\cap^{\text{RR}}}\xspace}
\newcommand{\ICCS}{\ensuremath{I_\cap^{\text{CCS}}}\xspace}
\newcommand{\IDEP}{\ensuremath{I_\cap^{\text{DEP}}}\xspace}
\newcommand{\ICT}{\ensuremath{I_\cap^{\text{CT}}}\xspace}
\newcommand{\IMES}{\ensuremath{I_\cap^{\text{MES}}}\xspace}
\newcommand{\IIG}{\ensuremath{I_\cap^{\text{IG}}}\xspace}
\newcommand{\Ialpha}{\ensuremath{I_\cap^{\alpha}}\xspace}
\newcommand{\Iprec}{\ensuremath{I_\cap^{\prec}}\xspace}
\newcommand{\Ido}{\ensuremath{I_\cap^{\text{do}}}\xspace}
\newcommand{\Idelta}{\ensuremath{I_\cap^{\delta}}\xspace}
\newcommand{\IRAV}{\ensuremath{I_\cap^{\text{RAV}}}\xspace}
\newcommand{\IRDR}{\ensuremath{I_\cap^{\text{RDR}}}\xspace}
\newcommand{\ISX}{\ensuremath{I_\cap^{\text{SX}}}\xspace}

\newcommand{\IRRs}{\ensuremath{I_\cap^{\text{RR}}(X_1,X_2;Y)}\xspace}

\newcommand{\ICTs}{\ensuremath{I_\cap^{\text{CT}}(X_1,X_2;Y)}\xspace}
\newcommand{\IMESs}{\ensuremath{I_\cap^{\text{MES}}(X_1,X_2;Y)}\xspace}
\newcommand{\IIGs}{\ensuremath{I_\cap^{\text{IG}}(X_1,X_2;Y)}\xspace}

\newcommand{\Idos}{\ensuremath{I_\cap^{\text{do}}(X_1,X_2;Y)}\xspace}
\newcommand{\Ideltas}{\ensuremath{I_\cap^{\delta}(X_1,X_2;Y)}\xspace}
\newcommand{\IRAVs}{\ensuremath{I_\cap^{\text{RAV}}(X_1,X_2;Y)}\xspace}

\newcommand\figsubref[2]{\hyperref[#1]{\ref*{#1}#2}}

\newtheorem{theorem}{Theorem}
\newtheorem{lemma}{Lemma}
\newtheorem{proposition}{Proposition}
\newtheorem{corollary}{Corollary}

\newtheorem{remark}{Remark}

\newenvironment{DIFnomarkup}{}

\usepackage{caption}
\usepackage{subcaption}
\title{The mathematical landscape of partial information decomposition: \\ A comprehensive review of properties and measures}

\begin{document}

\author[1]{Alberto Liardi\thanks{Corresponding author: \href{mailto:a.liardi@imperial.ac.uk}{a.liardi@imperial.ac.uk}}}
\affil[1]{Department of Computing, Imperial College London, UK}

\author[1,2,3]{Keenan J. A. Down}
\affil[2]{Department of Psychology, University of Cambridge, Cambridge, UK}
\affil[3]{Department of Psychology, Queen Mary University, London, UK}

\author[4,1]{George Blackburne}
\affil[4]{Department of Experimental Psychology, University College London, London, UK}

\author[5]{Matteo Neri}
\affil[5]{Institut de Neurosciences de la Timone, Aix Marseille Université, Marseille, France}

\author[1,6]{\mbox{Pedro A. M. Mediano}}
\affil[6]{Division of Psychology and Language Sciences, University College London, UK}

\date{\vspace{-1.2cm}}

\maketitle

\begin{abstract}
    
Partial Information Decomposition (PID) has become one of the most prominent information-theoretic frameworks for describing the structure and quality of information in complex systems. Despite its widespread utility, there exists no unique solution constraining precisely how a PID should be constructed, leading to a multiverse of different formalisms with different mathematical commitments. In this work, we provide a comprehensive overview of the mathematical landscape of PID. By integrating existing PID measures into a common language, we systematically examine all major approaches to the PID framework that have emerged so far, determining for each measure whether or not each known property holds. In addition, we derive a web of all known theorems mapping the relationships and incompatibilities between these properties, before also revealing some novel interdependency results. 
In doing so, we chart a brief history of the framework, promote a unified perspective for its discussions, and offer a path towards both theoretical refinement and informed empirical applications for the future of this powerful method.
\end{abstract}

\section{Introduction}

Classical information theory, unified by Claude Shannon in 1948~\cite{shannon1948mathematical}, is one of the great mathematical triumphs of the 20th century, bringing together notions from communications theory, thermodynamics, and statistical mechanics into one all-encompassing theoretical discipline~\cite{nyquist1924certain, hartley1928transmission, clausius1865ueber, boltzmann1872weitere, boltzmann1970weitere, jaynes1957information}. A key reason for its influence across disciplines and its broad applicability is that information theory provides a precise and remarkably general formalism for describing and capturing all kinds of relationships between random variables. 
Be they linear, nonlinear, or otherwise difficult to describe, mutual information, for example, is always equal to zero if and only if the variables are entirely independent~\cite{cover1999elements}. 
That is to say, while correlation can detect linear relationships between variables (or at the very least, relatively monotonic relationships), mutual information and other information-based approaches can detect all possible classes of interaction phenomena.

Despite this, surprising gaps emerge in classical information theory~\cite{james2017multivariate, williams2010nonnegative}, especially in its treatment of higher-order dependencies in systems of many interacting components. Understanding these interdependencies is fundamental to studying the behaviour of complex systems, where precise mechanisms are often difficult to reckon with~\cite{rosas2022disentangling, luppi2020synergisticCore, down2026synergistic, neri2025taxonomy, combrisson2025higher, battiston2020networks, robiglio2025synergistic}. 
Fully describing these interactions using Shannon's theory alone is challenging, since the information-sharing relationships among more than two variables can take qualitatively distinct forms that are not captured by mutual information or other classical measures of interdependence.
Suppose, for example, we wish to send an encrypted letter, accompanied by a second letter that explains how to decrypt it. Neither letter alone provides any information about the intended message, each leaving the recipient uncertain of the intended message. It is only when both letters are combined that the intended information is conveyed: knowledge of both systems is required to describe the content, with the two letters acting \textit{synergistically} to convey the message. Suppose now that we write the same unencrypted message in both letters. In this case, the message would be transmitted \textit{redundantly}, offering the recipient two copies of the intended data. The classical theory, somewhat surprisingly, struggles to differentiate between these two effects: determining whether multiple variables contain repeated or complementary information about a target variable cannot be directly quantified using classical information theory alone.

Partial Information Decomposition (PID) is a method in contemporary information theory which addresses this particular weakness of the classical theory: the inability to separate redundant and synergistic contributions. Introduced in 2010 by Williams and Beer \cite{williams2010nonnegative}, PID aims to disentangle these components by extending the inclusion–exclusion view of information beyond the quantities described by Shannon. By defining an intersection measure $I_{\cap}(X_1, \ldots, X_n;Y)$ that captures how much information \textit{the sources} $X_1,\ldots, X_n$ share redundantly about \textit{a target} $Y$, one can derive elementary \textit{atoms} that partition information into qualitatively distinct parts. In doing so, atoms corresponding to redundant, synergistic, and unique contributions to the mutual information emerge naturally from the proposed lattice structure.

Despite the beauty and seminal nature of Williams' and Beer's PID framework, the intersection measure they proposed, $\Imin$, was immediately met with additional scrutiny~\cite{harder2013bivariate, bertschinger2013shared}. Most notably, Harder \etal pointed out that for two independent bits $X_1$ and $X_2$, with target $Y = X_1X_2$ (a system called the two-bit-copy, TBC), $\Imin$ assigns to the sources $X_1$ and $X_2$ one bit of redundant information about $Y$, even though they are entirely independent \cite{harder2013bivariate}. Bertschinger \etal also noted that $\Imin$ has the peculiar property that adding resolution to the target variable (by joining it with another variable, for example), can sometimes reduce the amount of redundancy. If \textit{more} can be known overall about the target variable, it is certainly surprising that the sources somehow share \textit{less} information about the target \cite{bertschinger2013shared}.

For these reasons, while the partial information lattice of Williams and Beer was (rightly) praised, their proposed intersection information measure $\Imin$ was not, leaving a complex vacuum: an algebraic theory without an accompanying measure. As a result, a suite of new properties or ``axioms'' that a redundancy function should satisfy was proposed shortly after PID's inception. This, in turn, resulted in the introduction of many different PID measures, each designed to cover properties that the original measure $\Imin$ could not achieve. Since the genesis of the partial information decomposition method, at least 19 different approaches have been proposed, and, at the time of this review, no general consensus on the best approach has emerged~\cite{williams2010nonnegative, harder2013bivariate, griffith2014intersection, bertschinger2014quantifying, griffith2014quantifying, griffith2015quantifying, barrett2015exploration, goodwell2017temporal, ince2017measuring, james2018unique, kay2018exact, james2018uniquekey, finn2018pointwise, niu2019measure, kay2024partial, sigtermans2020partial, sigtermans2020towards, makkeh2021introducing, schick2021partial, ehrlich2024partial, kolchinsky2022novel, venkatesh2022partial, mages2023decomposing, lyu2024explicit}.
 
Although this process enriched the conceptual landscape, it also led to a proliferation of partially overlapping and competing properties, each reflecting different conceptions of what redundant information is. To complicate the matter, several of these properties were subsequently shown to be mutually incompatible, implying that no single redundancy measure can satisfy them all \cite{rauh2014reconsidering}. Thus, the field now comprises a diverse collection of measures grounded in different, and often irreconcilable, axiomatic choices, contributing to ambiguity in both the theoretical interpretation and empirical application of PID.  In this work, we endeavour to alleviate some of these outstanding shortcomings by providing a unified and comprehensive perspective on the state of the field.

\subsection{Main contributions}
In this work, we provide a unified treatment of Partial Information Decomposition (PID), generating the first comprehensive formal resource for its largely decentralised research programme. 

Firstly, we present a brief mathematical introduction to those who are unfamiliar with the PID methodology, offering some motivation for the problem itself, and exploring why so much research has been done in this surprisingly challenging direction (Sec.~\ref{sec:guide}). 
In doing so, we provide an overview of all major properties stated in the literature, along with a brief discussion of their significance and the motivations behind their original conception (Sec.~\ref{sec:properties} and Table~\ref{tab:property_summary}).

Similarly, we also present a brief historical overview of each of the PID measures (Sec.~\ref{sec:brief_history}), with a longer version reported in the appendix, describing the rationale underpinning their introductions. 
To make comparisons between measures and properties more straightforward, we provide mathematical definitions of each PID measure and axiom in a standardised language in the appendix.

We then present our first major result: a systematic verification of whether each proposed PID property is satisfied by each existing PID measure (Table~\ref{tab:PID_properties}). For all results not already known in the literature, we provide either a proof or a counterexample for each combination (available in the appendix).
To yield further insights into the latent relationships between PID measures, we hierarchically categorise PID measures based on these axiomatic profiles (Sec.~\ref{sec:results_measures} and Fig.~\ref{fig:similarity}), providing insights into the major philosophical branches describing \textit{redundant} information.

Furthermore, we then collate all known theorems relating PID properties, encompassing both implications and no-go results (Sec.~\ref{sec:imps_nogos}). In doing so, we present new and old theorems, introduce a hypergraph-based representation of these relationships (Fig.~\ref{fig:nogo}), and integrate them into an open-source implementation based on an automatic theorem prover, enabling automated verification of property compatibility.

Finally, we discuss the implications of this work on the Partial Information Decomposition method as a whole (Sec.~\ref{sec:discussion}). There, we outline useful practices for empirical implementation of the PID framework, briefly describe alternative approaches to the PID method, and discuss potential gaps in the literature and open avenues for future research.

In short, our contributions are as follows:
\begin{itemize}
    \item We present an overview of all major PID properties stated in the literature, along with a brief description of their significance and motivation (Sec.~\ref{sec:guide} and Table~\ref{tab:property_summary}).
    \item We provide a systematic verification of whether or not each known property holds for each and every measure (Table~\ref{tab:PID_properties}).
    \item We collate all known theorems relating PID properties into a common language, including both implications and no-go theorems, provide novel theorems (Sec.~\ref{sec:properties}), and link them together in an automatic theorem prover~\cite{de2008z3}.
    \item We perform a descriptive analysis of the measures and properties, resulting in a hierarchy of clustered measures and a hypergraph description of the theorems (Sec.~\ref{sec:results}).
\end{itemize}

\section{A comprehensive guide to PID properties and their relationships} \label{sec:guide}

\subsection{A partially gentle introduction to Partial Information Decomposition}
Given a set of $n$ sources $\Xs$ and a target $Y$, partial information decomposition aims to decompose the information that a collection of $n$ sources $\Xs$ provides about a target $Y$ into multiple qualitatively different parts: redundancy, unique information, and synergy. These denote the information provided about the target by both sources separately, only one source, and the two sources taken together, respectively. In the case of more than two sources, combinations of these three can also appear.

To achieve the decomposition, Williams and Beer build on the inclusion-exclusion principle (IEP) and construct a partially ordered set (poset) of redundancies, i.e.\ information quantities that represent the information common to a set of sources. The corresponding algebraic object is known as the \textit{redundancy lattice}. Building the redundancy lattice involves considering the set of possible collections of sets of sources $\mathcal{A}$ where no source is a superset of any other:
\begin{equation}
    \mathcal{A}(\Xs) \equiv \{ \alpha\in\mathcal{P}(\mathcal{P}(\Xs)) \,|\, \forall A_i,A_j\in\alpha, A_i\not\subseteq A_j\} \,,
\end{equation}
where $\mathcal{P}(\mathcal{X})$ represents the set of all non-empty subsets of $\mathcal{X}$. This set can then be equipped with the partial order $\preceq$, where $\alpha \preceq \beta$ for $\alpha, \beta \in \mathcal{A}$ if, for every source $B \in \beta$ (consisting of possibly multiple variables), there is at least one source $A \in \alpha$ such that $A \subseteq B$. Put simply, the lattice creates a hierarchy of information based on availability, or the level of knowledge required to access it: $\alpha$ is ``below'' $\beta$ if the information in $\alpha$ can be fully recovered from any source present in $\beta$.

For each collection $\alpha\in\mathcal{A}$, we can then define a redundancy function $\Ic$. Consider, for example, a system of three variables $X_1, X_2, X_3$ providing information about a target $Y$. Then $\Ic(X_1,X_2X_3;Y)$ corresponds to the information that can be known about $Y$ from either $X_1$ or the joint variable $X_2X_3$.

The redundancy lattice provides insights into how redundant information is distributed across the sources, with higher elements in the lattice providing at least as much redundant information as lower nodes. Following this setup, one can implicitly define the so-called PID atoms $I_\partial$, which signify how much information is gained at each level:
\begin{align}\label{eq:PID_intersect_def}
    \Ic(\alpha;Y) = \sum_{\beta\preceq\alpha} I_\partial(\beta;Y) \,,
\end{align}
where $\alpha$ is a generic collection of sets of sources, and the value of $I_\partial(\alpha, Y)$ is the amount of information corresponding to the \textit{atom} represented by $\alpha$. Using this expression, we can isolate the atoms themselves by applying M{\"o}bius inversion to $\Ic$~\cite{rota1964foundations, jansma2025fast}, obtaining each atom recursively as:
\begin{equation} \label{eq:PID_partial_def}
    I_\partial(\alpha;Y) = \Ic(\alpha;Y) - \sum_{\beta\prec\alpha} I_\partial(\beta;Y) \,.
\end{equation}
Throughout this work, we will, wherever possible, explicitly write the collection of sets of sources that we refer to. For example, the redundancy function of the collection $\alpha=\{\{X_1\},\{X_2,X_3\}\}$ will be written as $\Ic(X_1,X_2X_3;Y)$ (with corresponding atom $I_\partial(X_1,X_2X_3;Y)$), so as to minimise ambiguity.

Unsurprisingly, the number of possible collections of sources $\alpha$ can become extremely large even for a limited number of sources, growing with the Dedekind numbers \cite{gutknecht2021bits}. Offering some further constraints on the problem, Williams and Beer introduce several axioms on the redundancy measure $I_\cap$. The first of these, (Weak) Symmetry \Sz, follows the intuitive idea that the ordering of sources in $\alpha$ should not affect the value of redundancy. Additionally, they introduce (Weak) Monotonicity \Mz, which states that adding sources to the redundancy function should not increase \Ic, keeping it unchanged if the added source is more informative than another already present. For reasons that will be apparent later, we refer to this last part of the property as Subset Equality \SE~\cite{ince2017measuring}. 

Finally, since this is a decomposition of mutual information, a connection between these quantities and classical information-theoretic measures needs to be established. Hence, Williams and Beer introduced the Self-Redundancy \SR property, imposing that the redundancy function of one source and the target is equal to the mutual information between the two.  Thus, not only are \Sz, \Mz, and \SR intuitive, but they also ground the redundancy lattice in a well-behaved framework.

Although some proposed PID measures challenge \Mz, the remaining \Sz, \SR, and \SE are core properties of all PID definitions. In fact, as we will see, almost all of the proposed measures possess all three.

An example worth directly referencing is the bivariate case. Given two sources $X_1,X_2$ and a target $Y$, the collection of sets of possible sources is given by $\mathcal{A}(X_1,X_2)=\{\{\{1\},\{2\}\},\{1\},\{2\},\{1,2\}\}$, for a total of four possible source combinations. Correspondingly, there are four nodes in the redundancy lattice, and we can relate the intersection measure $\Ic$ to the individual atoms $I_\partial$ using Eq.~\ref{eq:PID_intersect_def} as
\begin{align} \label{eq:PID2_1}
    \Icap & = I_\partial(X_1,X_2;Y) \\
    I(X_1;Y) & = I_\partial(X_1,X_2;Y) + I_\partial(X_1;Y) \\
    I(X_2;Y) & = I_\partial(X_1,X_2;Y) + I_\partial(X_2;Y) \\
    I(X_1,X_2;Y) & = I_\partial(X_1,X_2;Y)+I_\partial(X_1;Y)+I_\partial(X_2;Y)+I_\partial(X_1X_2;Y) \,, \label{eq:PID2_4}
\end{align}
where we employed \Sz, \Mz, and \SR to significantly decrease the number of non-trivial redundancy terms.
In the notation above, $I_\partial(X_1,X_2;Y)$ represents redundancy, $I_\partial(X_i;Y)$ indicates information unique to $X_i$, and $I_\partial(X_1X_2;Y)$ is synergy.

As previously mentioned, performing partial information decomposition on a system with $n\ge3$ sources introduces additional complexity, as different combinations of redundant, unique, and synergistic effects become possible. For example, $I_\partial(X_1,X_2X_3;Y)$ represents the redundant information between source $\{X_1\}$ and $\{X_2,X_3\}$, with the latter term including the synergistic information $X_2$ and $X_3$ provide about $Y$. For more details on the PID lattice and the notation used in this work, we refer to Appendix \ref{appendix:mathematical_bg}.

The conceptual clarity and mathematical simplicity of this framework has contributed to its widespread popularity, with applications spanning from artificial~\cite{beer2015information, wibral2017partial, ince2017measuring, tax2017partial, proca2024synergistic} and biological neural networks~\cite{luppi2020synergisticCore, varley2023partial, varley2023multivariate, luppi2024information, gatica2021high, gatica2022high, combrisson2025higher, coronel2025integrated, santoro2025beyond}, to gene regulatory systems~\cite{chan2017gene, chen2018evaluating}, cellular automata~\cite{finn2018quantifying, rosas2018information}, and labour market dynamics \cite{rajpal2025synergistic}. 
However, to uniquely solve the linear system of equations given by Eq.~\eqref{eq:PID_partial_def}, an additional definition of intersection information $\Ic$ is required, as it is absent from the classical theory. Since specifying a redundancy function which has sufficiently many desired properties has been the primary route forward in this area, we now explore these additional constraints. In the next section, we outline the various properties that have been proposed over the years for the measure $I_\cap$, giving a brief history of their introduction, as well as implications and incompatibilities. 

\subsection{A summary of PID properties} \label{sec:properties}

Throughout the development of the partial information decomposition framework, a wide range of PID properties have been suggested. These properties, often referred to as \textit{axioms}, specify desiderata that a reasonable notion of redundancy is argued to satisfy. Generally, they are motivated by intuitive principles that one might expect redundancy to obey, supported by operational interpretations, grounded in empirical considerations, or inspired by the mathematical properties of classical information measures.
Unfortunately, it was demonstrated early in the literature that some of these properties are, in fact, mutually inconsistent, so that no single redundancy measure can simultaneously satisfy all of them. As a consequence, prolific discussion about which of these properties should be favoured or abandoned has exploded, with a lack of consensus persisting to the present day.

The first properties to be introduced, Self-Redundancy \SR, (Weak) Symmetry \Sz, and (Weak) Monotonicity \Mz, were introduced by Williams and Beer for the construction of an interpretable redundancy lattice~\cite{williams2010nonnegative}. In the same work, the authors also showed that, taken together, these properties guarantee a non-negative redundancy function -- a property referred to as Global Positivity, \GP. 
In addition, the authors showed that their proposed PID measure \Imin satisfied a property called Local Positivity, \LP, whereby all atoms obtained from the decomposition are themselves non-negative. Since then, \LP has long been a heavily desired property of PID definitions, as non-negative quantities enable a clear interpretation of the resulting PID atoms, e.g.\ in terms of information transmission, as is possible for standard information-theoretic measures.

More recently, PID measures that reject \LP in favour of a pointwise approach have been proposed~\cite{finn2018pointwise, makkeh2021introducing}. These suggest that negative atoms might have a natural interpretation as providing \textit{misinformation} about the target, similar to how negative values of pointwise mutual information naturally appear in the classical theory. Alternatively, such negative atoms might indicate the presence of the so-called ``mechanistic redundancy''~\cite{ince2017measuring}, i.e.\ redundancy that arises purely due to the source-target interdependence, independently of the source-source correlation~\cite{harder2013bivariate}. Within these approaches, the authors also reject the inequality condition of \Mz while retaining the equality statement, i.e.\ that $\Ic$ should remain unchanged if a source more informative than another is considered, thus introducing the Subset Equality property \SE~\cite{ince2017measuring}, later also referred to as ``deterministic equality''~\cite{kolchinsky2022novel}.

A fundamental property which is often not explicitly stated is Equivalence-class Invariance \EI, first formally suggested by Griffith \etal~\cite{griffith2014intersection}. Under \EI, redundancy functions are required to be invariant under isomorphisms, i.e.\ invertible transformations of the random variables. This invariance is a key principle in statistics, as it ensures that statistical quantities depend solely on the statistical structure of the variables and not on semantic details such as their particular labelling or representation. Importantly, this requirement highlights a fundamental distinction between statistical analysis and the study of underlying mechanisms, as it is known that different mechanisms can sometimes provide isomorphic probability distributions~\cite{pearl2014probabilistic}, therefore being indistinguishable from a purely statistical perspective~\cite{james2017multivariate}. For this reason, if \EI is taken to hold, purely causal considerations cannot serve to guide intuitions for the construction or interpretations of PID definitions~\cite{finn2018pointwise}. In the early works on PID, \EI was majorly adopted implicitly, being a universal and well-established property in statistics. However, its application to specific, seemingly intuitive toy gates quickly contributed to the emergence of profoundly contrasting views of how redundancy should behave.

Along this line, an apparently simple yet profound system that has fundamentally shaped the history of PID is the Two-Bit Copy (TBC) gate. The TBC represents the situation in which the target is a copy of the two sources (Table~\ref{tab:TBC}). 
\begin{DIFnomarkup}
\begin{table}[t]
    \centering
    \begin{tabular}{ccc}
    \toprule
         $X_1$ & $X_2$ & $Y$ \\
    \midrule
         0 & 0 & 00 $\equiv$ A \\
         0 & 1 & 01 $\equiv$ B \\
         1 & 0 & 10 $\equiv$ C \\
         1 & 1 & 11 $\equiv$ D \\
     \bottomrule
    \end{tabular}
    \caption{\textbf{Truth table for the Two-Bit Copy system (TBC).}  The relabelling of the target is ensured by \EI.}
    \label{tab:TBC}
\end{table}
\end{DIFnomarkup}
Importantly, if we accept that information measures should be invariant under bijective transformations \EI, this system can be seen in two ways:
\begin{enumerate}
    \item the target can be viewed as a multivariate random variable $Y=(Y_1,Y_2)$, where $Y_1=X_1$ and $Y_2=X_2$; or
    \item the target can be seen as an entirely univariate variable with four states, with outcomes that can be renamed (e.g.\ to $\{A,B,C,D\}$). 
\end{enumerate}
In the former instance (1), from mechanistic intuitions, one could expect that redundancy is only present if the two sources are strongly correlated. This is more apparent if we reflect on an empirical scenario, where we consider two time series $X_1^t,X_2^t$ with high autocorrelation and low cross-correlation. If the past states $X_1^{t-1},X_2^{t-1}$ are considered as sources and the joint future state ($X_1^t,X_2^t$) as the target, one would expect low redundancy and high unique information, with redundancy increasing with the correlation between the sources. In the limit, one could expect zero redundant information for independent sources. This reasoning has led Harder \etal to first argue against \Imin as it assesses the same \textit{amount} of information but not the same \textit{kind} of information, and then to introduce the Identity property \ID~\cite{harder2013bivariate}.

Formally described, \ID requires that the redundancy in the TBC should be equal to the mutual information between the sources~\cite{harder2013bivariate}, guaranteeing some continuity between the \textit{sharedness} of mutual information and the \textit{sharedness} of redundancy. However, in the second scenario above (2), it is not entirely clear what about the system should indicate sharedness of information between the sources. With this perspective, interpreting or explaining the meaning of the \ID property is not completely trivial. However, if one additionally assumes \EI, then relabelling outcomes should not affect the statistical structure of the system, and hence the redundancy calculation should not depend on whether one chooses perspective (1) or (2). Therefore, the TBC highlight the divergence between the two approaches even for simple systems, challenging the coexistence of \EI with \ID. Moreover, it can be argued that \ID prevents the characterisation of ``synergistic entropy'', the additional uncertainty arising over and above the uncertainty guaranteed if they were independent~\cite{ince2017measuring}.

To partially alleviate these issues, the Independent Identity property \IID has been introduced as a weaker variant of \ID, prescribing that independent sources should not provide redundancy in the TBC~\cite{ince2017measuring}. This property maintains the intuitions developed in the example above, while avoiding the stronger implications of \ID for correlated sources. Since its appearance, \IID has been embraced as a key desideratum by numerous authors in the field~\cite{ince2017measuring, kolchinsky2022novel, lyu2024explicit}. However, even if one accepts \IID and not \ID, it has been argued that a non-vanishing redundancy should still be expected even with uncorrelated sources in the TBC since they can ``happen to provide the same [probability mass] exclusions with respect to the [two-event] target distribution''~\cite{finn2018pointwise}. Following this reasoning, \ID and \IID must be rejected, as accepting \EI naturally implies that any conclusion should not be inferred from causal interpretations, and hence that the intuitions derived from (1) are misleading~\cite{finn2018pointwise}.

On the other hand, a complementary approach advocated by Chicharro is to drop \EI in favour of \IID~\cite{chicharro2018identity}, suggesting that the semantic value is meaningful as it is crucial to assess mechanistic redundancy. 
As we see later, this difference in stance is one of the major factors that changes the results obtained by the different PID definitions.
We formalise the inconsistencies between \ID, \EI, and other properties in the following sections, further discussing the interpretational consequences of adopting a specific axiom.

\begin{DIFnomarkup}
\begin{table}[t]
    \centering
    \setlength{\tabcolsep}{-1pt}
    \begin{tabularx}{\textwidth}{c c >{\centering\arraybackslash}X}
        \toprule
         \textbf{Name} & \textbf{Symbol} & \textbf{Description} \\
         \midrule
        Self-redundancy & \SR & \Ic of a single source equals its mutual information. \\
        (Weak) Symmetry & \Sz & \Ic is invariant under any permutation of the sources. \\
        (Weak) Monotonicity & \Mz & Adding a source to \Ic does not increase redundancy\footnotemark. \\
        Global Positivity & \GP & \Ic is non-negative. \\
        (Weak) Local Positivity & \LPz & All $I_\partial$ atoms are non-negative for $n=2$ sources. \\
        Strong Local Positivity & \LPo & All $I_\partial$ atoms are non-negative for $n\ge3$ sources. \\
        Local Positivity & \LP & All $I_\partial$ atoms are non-negative $\forall\, n$ sources. \\
        Independent Identity & \IID & In the TBC, independent sources provide vanishing \Ic\footnotemark[2]. \\
        Identity & \ID & In the TBC, \Ic is equal to the mutual information between sources\footnotemark[2]. \\
        Target Monotonicity & \TM & Adding a target does not decrease the value of \Ic. \\
        Target Chain rule & \TC & The chain rule for \Ic holds for the target variables. \\
        Subset Equality & \SE & Removing a source more informative than another preserves \Ic. \\
        Lower Bound & \LB & A variable less informative than all sources lower-bounds \Ic. \\
        Target Equality & \TE & Adding the target to the sources does not change \Ic. \\
        Strong Monotonicity & \Mo & Adding a source to \Ic does not increase redundancy\footnotemark[\value{footnote}]. \\
        Strong Symmetry& \So & \Ic is invariant under any permutation of sources or target. \\
        Assumption ($\ast$) & \AST & Redundancy and unique information only depend on the marginal and pairwise distributions between sources and target.\\
        Blackwell Property & \BP & Unique information vanishes iff one source Blackwell-precedes another. \\
        Additivity & \AD & \Ic of two independent subsystems equals the sum of their \Ic. \\
        Continuity & \CO & \Ic is continuous under small changes in the probability distribution. \\
        Differentiability & \DI & \Ic is differentiable with respect to the probability distribution. \\
        Equivalence-class Invariance & \EI & \Ic is invariant under the equivalence relation $\sim$ (Eq.~\eqref{eq:def_equivalence}). \\
         \bottomrule
    \end{tabularx}
    \caption{\textbf{Summary of existing PID properties.}}
    \label{tab:property_summary}
\end{table}
\end{DIFnomarkup}

Following a different line of reasoning, several axioms have been proposed in an effort to mirror properties of classical Shannon information measures, such as Target Monotonicity \TM, Target Chain rule \TC~\cite{bertschinger2013shared, finn2018pointwise, makkeh2021introducing}, Strong Symmetry \So~\cite{bertschinger2013shared}, and Strong Monotonicity \Mo~\cite{griffith2014intersection}. It soon became apparent, however, that some of these PID properties are mutually incompatible with other fundamental ones, such as \LP and \ID~\cite{rauh2014reconsidering, matthias2025novel}. Although in this work we restrict our attention to the implications of these properties for redundancy functions, we note that previous studies have analysed what it means for the PID atoms themselves to satisfy the axioms~\cite{banerjee2015synergy, rauh2014reconsidering, rauh2017extractable}. Along this line, Rauh \etal have recently advocated a more focused investigation of additional properties that had previously received little attention: Additivity \AD, Continuity \CO, and Differentiability \DI. These are fundamental mathematical properties satisfied by all classical information measures~\cite{rauh2023continuity}. 

Operational considerations grounded in game theory and decision-making strategies have also contributed to the formulation of PID desiderata. Investigating what it means for an agent to be more informed than another in a decision-theoretic setting led to the introduction of Assumption \AST~\cite{bertschinger2014quantifying}, which suggests that redundancy and unique information should depend on the marginal distributions between sources and targets, as opposed to the full probability distribution of the system.

\footnotetext{Note that these properties also entail an equality condition, see Appendix for the precise definition.}
\footnotetext[2]{Note that these properties also hold for more general distributions than the TBC, see Appendix for more details.}

Since its introduction, \AST has been criticised from several perspectives. Firstly, it has been shown that it can artificially induce correlations between the sources~\cite{james2018uniquekey}. Moreover, the decision problem that originally prompted the introduction of \AST was later deemed ``too restrictive''~\cite{ince2017measuring}: when the operational approach was generalised to a broader game-theoretic setting, explicit violations of Assumption \AST were demonstrated~\cite{ince2017measuring}.

Another central concept in statistical decision theory is the so-called Blackwell order, which provides a partial ordering relation between random variables in terms of their informational usefulness for decision-making problems. Informally, a variable is said to be Blackwell superior to another if it enables at least as good performance in any decision problem. This intuition is formalised by Blackwell’s informativeness theorem~\cite{blackwell1953equivalent}. Drawing inspiration from this framework, Blackwell's theorem was operationalised within PID through the introduction of the Blackwell Property \BP, which states that a source provides zero unique information if and only if it is Blackwell inferior to all other sources~\cite{venkatesh2022partial, kolchinsky2022novel, mages2023decomposing}. While both \AST and \BP arise from operational considerations rooted in decision theory, it is important to emphasise that they are logically independent: neither implies nor is implied by the other, except for specific cases (see Appendix C of Ref.~\cite{venkatesh2022partial} for a detailed discussion).

Finally, a small subset of properties has been proposed primarily on the basis of their intuitive appeal. Target Equality \TE~\cite{griffith2014intersection, kolchinsky2022novel} states that adding the target variable itself to the set of sources should not alter the value of the redundancy \Ic. This property is closely related to Strong Monotonicity \Mo: although \TE was introduced as a separate property by Kolchinsky~\cite{kolchinsky2022novel}, it had already appeared as part of \Mo in the earlier work of Griffith \etal~\cite{griffith2014intersection}. With the aim of constraining the range of admissible redundancy values, the Lower Bound property \LB was introduced~\cite{griffith2014intersection}. \LB asserts that any random variable that is less informative about the target than all sources should provide a lower bound on the redundancy of the system. 

Additional properties have been suggested in the literature. Although we do not consider them in detail in this work, we briefly mention them here. Related to \AD and \CO is also \textit{superaddivity}, which is often satisfied by all measures containing a $\min$. The \textit{Two-event partition} property was proposed by~\cite{finn2018pointwise} as a pointwise version of \AST, for which redundancy and unique information should only depend on the marginal and pairwise pointwise distributions between each source and the target. %

In Ref.~\cite{chicharro2018identity}, Chicharro proposed two additional properties: the \textit{weak} and the \textit{strong} axioms on stochasticity, where a generalisation of \ID is formalised by setting specific constraints on the synergy atom when deterministic relations between source and targets are present. Interestingly, both conditions are stronger requirements than \ID, being both sufficient, but not necessary.

Another property related to the identity axioms is the \textit{combined secret sharing property} suggested by Rauh~\cite{rauh2017secret}. At its core, this property captures the intuition that PID should reflect the behaviour of ideal secret sharing schemes, in which groups of sources either fully know a secret or know nothing about it. Accordingly, redundancy is required to match exactly the amount of information that is jointly accessible to all the considered source groups~\cite{rauh2017secret}. In its pairwise version, this property implies \IID, but not \ID.

More recently, Kolchinsky introduced a property called \textit{order equality}~\cite{kolchinsky2022novel}, which resembles \SE but generalises it to a generic ordering of choice, e.g.\ Blackwell. In the case of a partial information decomposition with $n=2$ sources, the order equality property with Blackwell ordering corresponds to the left implication of \BP.

In the same work, Kolchinsky also discusses the \textit{inclusion-exclusion principle} (IEP). Although not strictly a PID property per se, the IEP is a fundamental assumption within PID that allows one to express synergistic effects in terms of redundancy effects. Kolchinsky argues that there is no universal reason why this property should hold, advocating for the need to define a union information quantity from which synergy can be derived~\cite{kolchinsky2022novel}. In particular, Kolchinsky highlights that rejecting the IEP might remove some of the incompatibilities that we describe in the next section. Although we do not treat IEP as a property in and of itself in this work, future research might benefit strongly from carefully reconsidering its role in the field of information decomposition.

For the reader's benefit, we provide a brief description of each property in Table~\ref{tab:property_summary}, with the full definitions for each property available in the appendix.

\subsection{Implications and no-go theorems}
\label{sec:theorems_intro}

\begin{DIFnomarkup}
\begin{table}[t]
\centering
\renewcommand{\arraystretch}{1.1} %
\begin{tabular}{lcr}
\toprule
\textbf{Name} & \textbf{Implications} & \textbf{Source} \\
\midrule
Lemma~\ref{lemma:LP--GP} 
    & $\LP \implies \GP$ 
    & \cite{williams2010nonnegative} \\

Lemma~\ref{lemma:ID--IID} 
    & $\ID \implies \IID$ 
    & \cite{ince2017measuring} \\

Lemma~\ref{lemma:So--Sz} 
    & $\So \implies \Sz$ 
    & \cite{bertschinger2013shared} \\

Lemma~\ref{lemma:Mz--Mo} 
    & $\Mz \implies \SE$ 
    & \cite{ince2017measuring} \\

Lemma~\ref{lemma:DI--CO} 
    & $\DI \implies \CO$ 
    & \cite{makkeh2021introducing} \\
\midrule
Prop.~\ref{prop:M0-SR--LB} 
    & $\Mz, \SR \implies \LB$ 
    & \cite{griffith2014intersection} \\

Prop.~\ref{prop:LB-SR--GP} 
    & $\SR, \LB \implies \GP$ 
    & \cite{williams2010nonnegative} \\

Prop.~\ref{prop:Mz-SR--GP}
    & $\Mz, \SR \implies \GP$ 
    & \cite{williams2010nonnegative} \\

Prop.~\ref{prop:LP--M0} 
    & $\Sz, \LP \implies \Mz$ 
    & \cite{williams2010nonnegative, gutknecht2025babel,matthias2025novel} \\

Prop.~\ref{prop:TC-GP--TM} 
    & $\GP, \TC \implies \TM$ 
    & \cite{bertschinger2013shared} \\

Prop.~\ref{prop:LP--TE}
    & $\SR, \Sz, \LP \implies \TE$
    & new \\

Prop.~\ref{prop:LP--M1}
    & $\SR, \Sz, \LP \implies \Mo$
    & new \\

\midrule
Theo.~\ref{theo:TC-TE-SR--ID} 
    & $\SR, \Sz, \TC, \TE \implies \ID$ 
    & new \\

Cor.~\ref{cor:TC-LP--ID}
    & $\SR, \Sz, \LP, \TC \implies \ID$
    & \cite{bertschinger2013shared} \\

Theo.~\ref{theo:M0-TE--M1} 
    & $\Mz, \TE \iff \Mo$ 
    & new \\

Theo.~\ref{theo:BP--TE-SE} 
    & $\BP \implies \SE, \TE \quad (n=2)$ 
    & new \\
    
Prop.~\ref{prop:ID-TM--ineq} 
    & $\ID, \TM \implies I_\cap(X_1,X_2;f(X_1,X_2)) \le I(X_1;X_2)$ 
    & \cite{banerjee2015synergy} \\

\midrule
Theo.~\ref{theo:Mz-So--TM} 
    & $\Mz, \So \implies \TM$ 
    & \cite{bertschinger2013shared} \\

Theo.~\ref{theo:S1-SE--TE} 
    & $\SE, \So \implies \TE$ 
    & new \\

Theo.~\ref{theo:SR-SE-So--ID} 
    & $\SR, \SE, \So \implies \ID$ 
    & new \\
\bottomrule
\end{tabular}
\caption{\textbf{Implications, equivalences, and inequalities among PID properties.}}
\label{tab:PID_implications}
\end{table}
\end{DIFnomarkup}

As mentioned above, not all the PID properties thus far proposed are independent of each other. 
In fact, multiple theorems have been established, demonstrating non-trivial relations between the properties, complicating the quest for finding the minimal requirements that a redundancy function should satisfy. Several results show how certain properties can be derived from others (Table~\ref{tab:PID_implications}). For instance, examples of standard implications can be found in the original works by Williams and Beer, who showed that \GP follows from the minimal axioms \Sz, \SR, and \Mz~\cite{williams2010nonnegative}, or in Bertschinger \etal~\cite{bertschinger2013shared}, who remarked that \TC and \GP ensure \TM. 

Further exacerbating the problem, it has been shown that not all the suggested properties can be simultaneously satisfied, raising the question of which axioms to drop and which ones to retain. One of the most influential inconsistency theorems was the one presented by Rauh \etal, who showed that the original Williams and Beer axioms (\Sz, \SR, and \Mz), along with \EI, \LPo, and \ID, are not compatible~\cite{rauh2014reconsidering}. This was the first demonstration that apparently intuitive and widely accepted properties could not be simultaneously reconciled, highlighting the need for a more systematic understanding of the relationships among PID axioms. Since then, many additional no-go theorems have begun to populate the PID literature, primarily regarding properties such as \So, \LP, \ID, and \TC (Table~\ref{tab:PID_incompatibilities}).

A common denominator across the majority of these incompatibility results is the use of the so-called XOR-Source-Copy gate, first introduced by Rauh \etal~\cite{rauh2014reconsidering} and then also used in the later works~\cite{finn2018pointwise, kolchinsky2022novel, lyu2024explicit, matthias2025novel}. This system is composed of three sources of the form $X_1,X_2\sim Bern(1/2)$ and $X_3=X_1 \oplus_2 X_2$, together with a target given by the joint variable $Y=(X_1,X_2,X_3)$. The fact that this gate serves as such a good counterexample for many, otherwise intuitive, properties begs the question: what is peculiar about this system that current measures struggle to understand? A plausible explanation is that it exhibits genuinely higher-order redundancy, i.e.\ redundancy at the triplet level that cannot be reduced to pairwise relations. This is a consequence of $X_1,X_2,X_3$ being pairwise independent, yet any one of them is fully determined by the other two. Such a configuration is clearly not possible in systems with only two sources. From this perspective, PID appears to lack the ability to distinguish between different kinds of redundancy. Indeed, Rauh pointed out that the XOR-Source-Copy is characterised by one bit of \textit{redundant} information which is not \textit{shared}, and as such it cannot be appointed to a specific atom of the PID lattice~\cite{rauh2017secret}. This observation suggests that intuitions based on shared information might be misleading when defining redundancy in multivariate settings. It is also worth noting that this counterexample leverages the relabelling invariance of the outcomes (\EI). Therefore, an alternative solution could be to drop \EI in favour of \IID (Sec.~\ref{sec:MR-ID-EI}).

Finally, we note that several results proposed in the literature were, unfortunately, later shown to be incorrect. These include the claim that \LP and \TM are incompatible, reported in Ref.~\cite{banerjee2015synergy}, and the fact that $\Sz, \Mz, \SR, \EI, \IID$ are incompatible for $n \ge 3$ (Lemma 3 of Ref.~\cite{lyu2024explicit}), as the proposed proof requires the property \LP, which is not stated in the hypotheses.

\subsection{A brief history of PID measures}
\label{sec:brief_history}

In this section, we provide a brief chronological overview of the introduction of the existing PID measures to date. A more comprehensive description is reported in the appendix, where we also include a visual illustration of the historical development of PID (Fig.~\ref{fig:timeline}).

\begin{DIFnomarkup}
\begin{table}[t]
\centering
\renewcommand{\arraystretch}{1.1} %
\begin{tabular}{lcr}
\toprule
\textbf{Name} & \textbf{Incompatible properties} & \textbf{Source} \\
\midrule
Theo.~\ref{theo:noIID-LP1} 
    & $\Sz, \SR, \EI, \LPo, \IID$ 
    & \cite{matthias2025novel} \\

Theo.~\ref{theo:noSR-Sz-Mz-TC}
    & $\SR, \Sz, \Mz, \TC$
    & new,\cite{matthias2025novel}\footnotemark[3] \\

Cor.~\ref{cor:noSR-LP-TC} 
    & $\SR, \Sz, \LP, \TC$ 
    & \cite{matthias2025novel} \\

Cor.~\ref{cor:LPz-ID--noTC} 
    & $\SR, \Sz, \EI, \LP, \IID, \TC$ 
    & \cite{finn2018pointwise} \\

Theo.~\ref{theo:noGP-IID-TM-Mech}
    & $\GP, \IID, \TM,$ mechanistic redundancy 
    & \cite{rauh2017extractable,james2018unique} \\

Theo.~\ref{theo:noSo} 
    & $\Mz, \SR, \LPo, \So$ 
    & \cite{bertschinger2013shared} \\
\bottomrule
\end{tabular}
\caption{\textbf{Incompatibility results among PID properties.} Where applicable, \EI has been explicitly stated, and \ID has been relaxed into \IID.}
\label{tab:PID_incompatibilities}
\end{table}
\end{DIFnomarkup}

The first redundancy measure \Imin was proposed alongside the PID framework in the seminal work by Williams and Beer~\cite{williams2010nonnegative}, making use of the idea that redundant information is present in every realisation of the target. Despite being intuitive, it was then quickly criticised for overestimating redundancy by estimating the same \textit{amount} of redundant information, rather than the same information \textit{content}~\cite{bertschinger2013shared, harder2013bivariate, griffith2014quantifying}. 

This concern -- that quality counts as much as quantity -- motivated Harder \etal to introduce \Ired~\cite{harder2013bivariate}, a redundancy function grounded in information geometry. \Ired overcomes this shortcoming of \Imin by satisfying the Identity property \ID. Despite this, \Ired is majorly limited by its applicability, as it is only valid in the bivariate PID case~\cite{harder2013bivariate}.

Shortly thereafter, Griffith \etal proposed an algebraic approach to defining redundancy, first with \Iwedge, inspired by the Gács-Körner common information~\cite{griffith2014quantifying}, and later with \Ialpha, which follows from intuitive considerations based on Markov chains~\cite{griffith2015quantifying}. Although these measures are well-defined for any number of sources and satisfy highly desirable properties, most notably \TM in the case of \Iwedge, they fail to satisfy \LP, tend to impose overly restrictive bounds on redundancy~\cite{griffith2015quantifying, kolchinsky2022novel}, and, for \Iwedge, are insensitive to statistical correlations not captured by G\'acs-K\"orner common random variables~\cite{james2018unique}.

A major step forward came with \IBROJA, simultaneously introduced by Bertschinger \etal alongside Griffith and Koch~\cite{bertschinger2014quantifying, griffith2014quantifying}. \IBROJA was the first redundancy function backed by an explicit operational approach, defining unique information by comparing the maximal expected reward functions of each source. Despite its interpretational clarity and the fact that it satisfies many of the proposed PID properties, this measure has been mildly criticised for overestimating redundancy by artificially inflating the correlation between the sources~\cite{ince2017measuring, james2018uniquekey, james2018unique}, as well as for relying on a decision-theoretic setup deemed too restrictive~\cite{ince2017measuring}. Additionally, this measure is also limited to the bivariate case.

\footnotetext[3]{This result has been introduced independently here and in Ref.~\cite{matthias2025novel}. We were also made aware of the statement of this theorem via private communication by Joseph Lizier. The proof presented here was developed independently.}

Up until this point, PID measures had focused exclusively on discrete systems. This changed with the introduction of \IMMI by Barrett~\cite{barrett2015exploration}, a redundancy measure specifically designed for Gaussian systems. \IMMI was shown to yield the same decomposition as \Imin, \Ired, and \IBROJA in the specific case of a bivariate Gaussian PID with a univariate target, hence becoming very popular and widely applied. However, \IMMI inherited the same shortcoming of \Imin, being only able to measure the same \textit{amount} of information between sources, without discerning between qualitatively distinct components of the interaction. This limitation motivated the proposal of \IRR by Goodwell and Kumar, a modified version of \IMMI which provides zero redundant information for uncorrelated sources~\cite{goodwell2017temporal}. 

A parallel but conceptually distinct line of work then led to the introduction of pointwise PID measures. Ince proposed the first fully pointwise PID function, \ICCS, defining redundancy via common changes in surprisal~\cite{ince2017measuring}. In doing so, \ICCS departs from several PID properties that had previously been uncontroversially accepted, such as \Mz, while retaining the essential axioms required to construct the redundancy lattice. This perspective was also pursued by Finn and Lizier with \IPM~\cite{finn2018probability, finn2018pointwise}, which distinguished informative from \textit{misinformative} contributions and suggested a double decomposition over specificity and ambiguity lattices. As with \ICCS, this pointwise approach entails abandoning some of the original PID axioms, although they continue to hold at the local level.

Later approaches explored the role of a maximum-entropy principle in defining redundancy. Investigating the subtle differences between entropy maximisation and mutual information minimisation under Assumption \AST, James \etal defined \IMES~\cite{james2018uniquekey}. \IMES was never intended to be a fully fledged PID measure, as it suffers from the major shortcoming of not satisfying \SE.

Around the same time, inspired by works on cybernetics and reconstructability analysis~\cite{zwick2004overview, krippendorff2009ross}, James \etal also introduced \IDEP, a measure of unique information based on the construction of a constraint lattice. Originally proposed only for discrete systems, this measure was then formalised for continuous Gaussian processes~\cite{kay2018exact}. As with \IBROJA, \IDEP only provides a full decomposition in the bivariate case~\cite{james2018unique}, thereby restricting its applicability.

Subsequent definitions of redundant information have continued to explore new avenues of thought. Niu and Quinn proposed \IIG~\cite{niu2019measure}, an information-geometry-based PID inspired by Amari's hierarchical decomposition~\cite{amari2001information}. Originally defined for a bivariate discrete PID with full-support distributions, \IIG was later generalised to the bivariate Gaussian case~\cite{kay2024partial}. Aiming to account for the indirect mediated associations between sources and target, Sigtermans proposed the redundancy measure \ICT based on the framework of Causal Tensors~\cite{sigtermans2020towards, sigtermans2020partial}.

An additional pointwise redundancy measure named \ISX was then suggested by Makkeh \etal, both for discrete~\cite{makkeh2021introducing} and continuous scenarios~\cite{schick2021partial, ehrlich2024partial}. \ISX measures redundancy via the so-called principle of \textit{shared mass exclusion}. Similarly to \IPM, \ISX is based on separate information and misinformation lattices, thereby renouncing some of the original PID axioms at the global level. 

More recently, inspired by the algebraic measures \Iwedge and \Ialpha, Kolchinsky introduced \Iprec, a redundancy measure grounded in the notion of Blackwell order~\cite{kolchinsky2022novel}. This definition admits a clear operational interpretation in decision-theoretic terms, satisfies \IID, and is defined for any number of sources, but can provide negative atomic values.

With a similar intent to define a PID based on Blackwell order, Venkatesh and Schamberg proposed \Idelta for multivariate Gaussian systems~\cite{venkatesh2022partial}, extending the earlier results of Barrett~\cite{barrett2015exploration} by measuring how distant one source is from being Blackwell superior to the other. Although this approach generalises the behaviour of \IMMI, it is only defined for a bivariate PID.

Continuing along the Blackwell ordering paradigm, Mages \etal finally introduced \IRDR, quantifying redundancy via the set-overlap of the reachable decision region of each source~\cite{mages2023decomposing}. \IRDR yields a pointwise decomposition that nevertheless satisfies the original PID axioms for any number of sources, reducing to \Imin and \IBROJA in specific circumstances. 

Finally, the most recent PID measure introduced, \Ido by Lyu \etal~\cite{lyu2024explicit}, defines unique information using an interventional approach inspired by Pearl's \textit{do-calculus}~\cite{pearl1995causal, pearl2009causal, pearl2012calculus}. While satisfying several desirable PID properties, a major shortcoming of \Ido is that it does not satisfy \SE, assigning non-zero unique information even when the sources are identical. Again, this measure is also only defined for the case of two sources.

In passing, we also mention \IRAV, an unpublished proposal that defines redundancy via the maximisation of coinformation over functions of the sources~\cite{james2018dit}. Despite being unpublished, an implementation of \IRAV is available via the \verb|dit| package, and we have hence included it in our analyses.

\section{Results}
\label{sec:results}

In this section, we present our main results. We begin with a systematic overview of the PID properties satisfied by each existing PID measure. We then present a comprehensive collection of theorems and no-go theorems that elucidate the logical relationships among these properties.

\begin{DIFnomarkup}
\begin{center}
\begin{table*}[t]
    \centering
    \setlength{\tabcolsep}{3pt} %
    \renewcommand{\arraystretch}{1.1} %
\begin{tabular}{l*{19}{c}}
\toprule
 &
$\Imin$ & $\IMMI$ & $\IRDR$ &
$\IBROJA$ & $\Idelta$ & $\Ired$ &
$\IIG$ & $\IRR$ & $\IDEP$ & $\ICT$ & $\IRAV$ & 
$\Iwedge$ & $\Ialpha$ & $\Iprec$ & $\IMES$ & $\Ido$ &
$\IPM$ & $\ISX$ & $\ICCS$ \\
\midrule
\textbf{(SR)}       & \V & \V & \V & \V & \V & \V  & \V & \V & \V & \V & \V & \V & \V & \V & \V & \V & \V  & \V  & \V \\
\textbf{(S\ts{0})}  & \V & \V & \V & \V & \V & \V  & \V & \V & \V & \V & \V & \V & \V & \V & \V & \V & \V  & \V  & \V \\
\textbf{(M\ts{0})}  & \V & \V & \V & \V & \V & \V  & \V & \V & \V & \V & \Vd& \V & \V & \V & \X & \X & \Xdd& \Xdd& \X \\
\textbf{(GP)}       & \V & \V & \V & \V & \V & \V  & \V & \V & \V & \V & \V & \V & \V & \V & \V & \V & \Xdd& \Xdd& \X \\
\textbf{(LP\ts{0})} & \V & \V & \V & \V & \V & \V  & \V & \V & \V & \X & \V & \X & \X & \X & \X & \X & \Xdd& \Xdd& \X \\
\textbf{(LP\ts{1})} & \V & \V & \V & \na& \na& \na & \na& \na& \na& \X & \X & \X & \X & \X & \na& \na& \Xdd& \Xdd& \X \\
\textbf{(IID)}      & \X & \X & \X & \V & \V & \V  & \X & \V & \V & \V & \V & \V & \V & \V & \V & \V & \X  & \X  & \V \\
\textbf{(ID)}       & \X & \X & \X & \V & \V & \V  & \X & \X & \V & \V & \V & \X & \X & \X & \V & \V & \X  & \X  & \X \\
\textbf{(TM)}       & \X & \V & \X & \X & \X & \X  & \Vs& \X & \X & \V & \X & \V & \X & \X & \X & \X & \X  & \X  & \X \\
\textbf{(TC)}       & \X & \X & \X & \X & \X & \X  & \X & \X & \X & \X & \X & \X & \X & \X & \X & \X & \V  & \V  & \X \\
\textbf{(SE)}       & \V & \V & \V & \V & \V & \V  & \V & \V & \V & \V & \V & \V & \V & \V & \X & \X & \V  & \V  & \V \\
\textbf{(LB)}       & \V & \V & \V & \V & \V & \V  & \V & \V & \V & \V & \Vd& \V & \V & \V & \X & \X & \X  & \X  & \X \\
\textbf{(TE)}       & \V & \V & \V & \V & \V & \V  & \V & \V & \V & \V & \Vd& \X & \V & \V & \V & \V & \X  & \X  & \X \\
\textbf{(M\ts{1})}  & \V & \V & \V & \V & \V & \V  & \V & \V & \V & \V & \Vd& \X & \V & \V & \X & \X & \X  & \X  & \X \\
\textbf{(S\ts{1})}  & \X & \X & \X & \X & \X & \X  & \X & \X & \X & \X & \X & \X & \X & \X & \X & \X & \X  & \X  & \X \\
\textbf{\BB}        & \V & \V & \V & \V & \V & \V  & \X & \X & \X & \X & \X & \V & \X & \X & \V & \V & \Xdd& \X  & \X \\
\textbf{(BP)}       & \X & \X & \V & \V & \V & \V  & \X & \X & \X & \X & \X & \X & \X & \V & \X & \X & \X  & \X  & \X \\
\textbf{(AD)}       & \X & \X & \X & \V & \X & \X  & \X & \X & \X & \X & \X & \V & \V & \V & \V & \V & \X  & \X  & \X \\
\COp       & \V & \V & \V & \V & \V & \Vm & \V & \V & \V & \V & \X & \X & \X & \X & \V & \V & \V  & \V  & \X \\
\DIp                & \X & \X & \X & \X & \X & \X  & \V & \X & \X & \X & \X & \X & \X & \X & \X & \V & \X  & \V  & \X \\
\textbf{(EI)}       & \V & \V & \V & \V & \V & \V  & \V & \V & \V & \V & \V & \V & \V & \V & \V & \V & \V  & \V  & \V \\
\bottomrule
\end{tabular}
    \caption{\textbf{Comparison of theoretical properties across redundancy measures.} The ordering of the measures is inspired by the hierarchical clustering of Fig.~\figsubref{fig:similarity}{b}. 
    $^*$: this property is not formally proven but is supported by empirical simulations.   
    \dag: these properties only hold for $n=2$ sources.
    $\mathsection$: these properties only hold in the case of full-support distributions. 
    $\small\|$: these properties are intended globally; for many measures, the almost-everywhere version holds.
    \ddag: the pointwise version of these properties holds for both informative and misinformative lattices separately. 
    }
    \label{tab:PID_properties}
\end{table*}
\end{center}
\end{DIFnomarkup}

\subsection{A comprehensive classification of PID measures} \label{sec:results_measures}
As discussed above, the proliferation of PID properties and measures has led to substantial ambiguity regarding which axioms are satisfied by which measures. We address this issue by providing a complete and unified classification (Table~\ref{tab:PID_properties}), explicitly showing for all existing PID measures whether they satisfy or not a given PID axiom. 
While a small subset of these results is already available in the literature, the majority of them have not been established previously. We report detailed proofs of all findings in the Appendix, also indicating where previous results can be found in the literature.

This comprehensive picture provides many different insights at a glance. First, it allows us to gauge which properties are most commonly satisfied in the PID literature, from classic axioms such as \Mz and \GP, to others which are often not explicitly stated, like \TE and \CO (Fig.\figsubref{fig:similarity}{a}). Then, it enables drawing connections among properties, suggesting which ones are usually satisfied together. Additionally, it makes clear what a specific operational interpretation entails in terms of PID properties. Finally, it also helps us understand the theoretical relationships between the various redundancy functions. 

\begin{figure}[t]
    \centering
    \includegraphics[width=\linewidth]{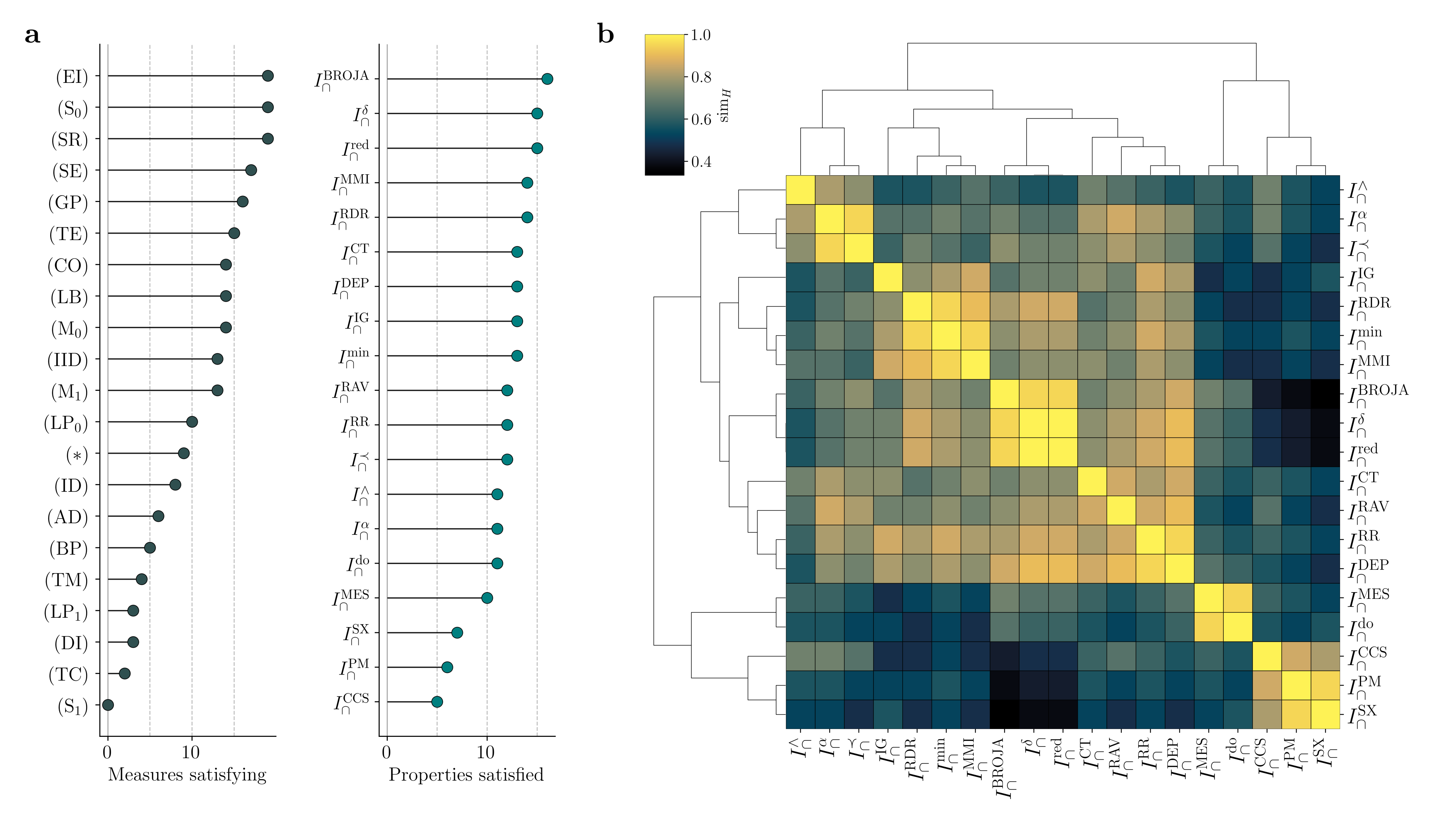}
    \caption{\textbf{Hierarchical organisation of PID measures induced by properties.}
    \textbf{a)} Number of measures satisfying each mathematical property (left), and number of properties satisfied by each measure (right), from Table~\ref{tab:PID_properties}. \textbf{(SR)}, \textbf{(EI)} and \textbf{(S\ts{0})} are satisfied by all measures, whereas \textbf{(S\ts{1})} is currently satisfied by no measures. $\IBROJA$ satisfies the most properties (16), whereas $\ICCS$ satisfies the least (5).
    \textbf{b)} Hierarchical clustering of the PID measures based on mathematical properties. 
    $\ICCS$, $\IPM$ and $\ISX$ clearly emerge as a distant community of measures likely due to their unique lack of \textbf{(GP)}. $\IMES$ and $\Ido$ occupy another cluster in between this and the main community of measures, likely due to the absence of \textbf{(M\ts{0})}, \textbf{(M\ts{1})}, \textbf{(LB)}, which they share with $\ICCS$, $\IPM$ and $\ISX$. The main cluster comprises $\IRDR$, $\Imin$, $\IMMI$, $\Ired$, $\IBROJA$, $\Idelta$, $\IIG$, $\IRR$, $\IRAV$, $\IDEP$, and $\ICT$. A particularly compact subgroup is unsurprisingly formed by $\IRDR$, $\Imin$ and $\IMMI$ due to satisfying \LPo and their lack of \IID. Similarly, $\Ired$, $\IBROJA$, $\Idelta$ form a subcluster as a consequence of satisfying \textbf{(BP)}, \textbf{\BB}, and \ID. Finally, the last community comprises \Iwedge, \Ialpha, and \Iprec, which violate \LPz while retaining \IID.
    }
    \label{fig:similarity}
\end{figure}

Specifically, we make this latter observation quantitative by studying the similarity of the PID measures via a hierarchical clustering based on the properties they satisfy (Fig.~\figsubref{fig:similarity}{b}). Results show that the measures are clustered in groups characterised by a small subset of PID properties. The first big differentiation is between measures that do not satisfy \GP (\ICCS, \IPM, \ISX) and those that do. This bigger group can be further clustered first into \Ido and \IMES, which do not satisfy \Mz nor \LPz, then into the algebraic measures $\Iwedge, \Ialpha, \Iprec$, which satisfy \Mz but not \LPz, and finally into all the other measures which satisfy both \Mz and \LPz. Within this subgroup, we notice additional clusterings due to \Imin, \IMMI, and \IRDR satisfying both \LPo and \AST, \IBROJA, \Idelta, and \Ired satisfying \AST but not \LPo, and the remaining that satisfy neither. Finally, the final cluster is given by the measures that satisfy \ID (\IRAV, \IDEP, \ICT) and those that do not (\IIG, \IRR).
In sum, although we considered twenty axioms in this analysis, only around a quarter of them are sufficient to capture the main mathematical distinctions among the PID measures proposed to date. This highlights that the various properties are not independent, and raises the question of whether there exist yet undiscovered relationships linking the key discriminating properties to the others.

\begin{figure}[ht]
    \centering
    \includegraphics[width=\linewidth]{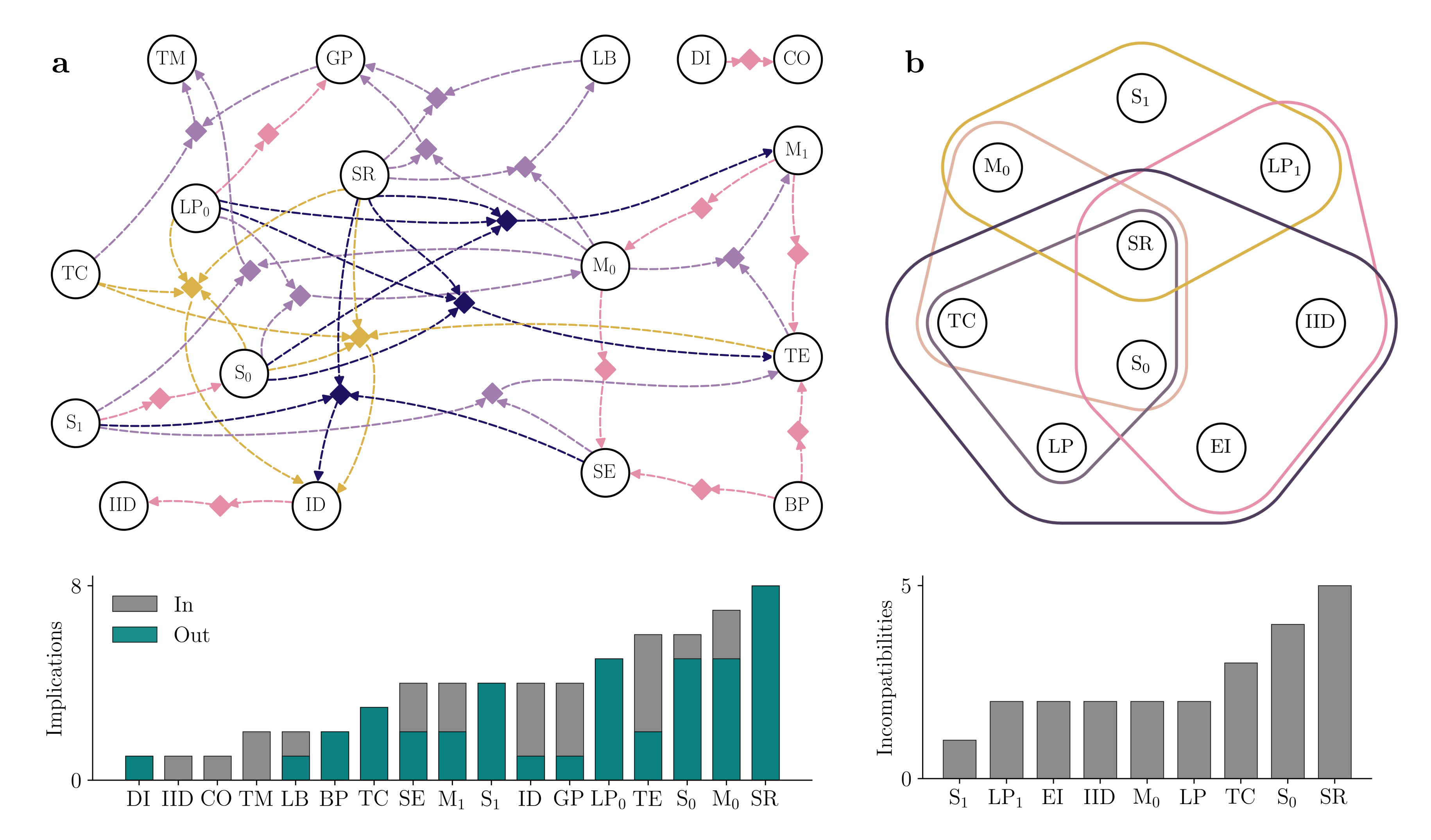}
    \caption{\textbf{Formal relationships among PID properties.} 
    \textbf{a)} Directed hypergraph representing logical implications of PID properties from Table~\ref{tab:PID_implications}, with hyperedges coloured by order of implication (top). In- and out-degree distributions for property implications (bottom) reveal that \textbf{(SR)} implies the most properties (8), followed by \Mz and \Sz (5), whereas \textbf{(TE)} is implied by the most properties (4).
    \textbf{b)} Undirected hypergraph representing logical incompatibilities from Table~\ref{tab:PID_incompatibilities}. Degree distributions for property incompatibilities reveal that despite implying the most properties, \textbf{(SR)} is also incompatible with the most properties (5).
    }
    \label{fig:nogo}
\end{figure}

\subsection{A comprehensive view of theorems linking PID properties}
\label{sec:imps_nogos}

In this section, we present a comprehensive account of the relations among PID properties. We review results previously established in the literature, strengthen their hypotheses and clarify explicitly assumptions that were previously left implicit, and complement them with new findings. In particular, we accompany the presentation of the results with a brief commentary on what their consequences imply, and how these incompatibility issues could be overcome.
The mathematical formulation of the results, as well as their proofs, is reported in the Appendix.

The first demonstration that not all the PID properties can simultaneously coexist was proposed by Rauh \etal~\cite{rauh2014reconsidering}, who showed that \Sz, \Mz, \SR, \EI, \LPo and \ID are incompatible. Although not explicitly stated in the hypotheses, we underline the necessity of \EI for the relabelling of the outcomes within the XOR-Source-Copy gate, as later clarified by Matthias \etal \cite{matthias2025novel}. In that work, the authors further refine Rauh \etalx's result by showing that \Mz is not necessary (Theorem 2 of \cite{matthias2025novel}). %
Hence, in its minimal assumptions, the incompatibility involves \Sz, \SR, \EI, \LPo, and \IID (Theo.~\ref{theo:noIID-LP1}). 

Interestingly, we also note that this same result has been indirectly derived from different premises. Kolchinsky assumed \Sz, \Mz, \SR, \EI, \IID, and showed that the union information ($I_\cup$) can exceed the joint mutual information ($I$), leveraging this result to argue against the inclusion-exclusion principle (IEP) (Lemma 1 of Ref.~\cite{kolchinsky2022novel}). However, while having $I_\cup>I$ can be intuitively unpleasant, it only becomes a formal mathematical contradiction if one also assumes \LP, as $I_\cup>I$ together with IEP indeed implies the presence of negative atoms (see Prop.~\ref{prop:Union-LP}). Hence, if one were to retain IEP at all costs, this argument leads to the rejection of \LP, therefore recovering Theo.~\ref{theo:noIID-LP1}.
Using a different approach based on subsystem consistencies, Lyu \etal proposed that \Sz, \Mz, \SR, \EI, \IID were incompatible for $n\ge3$ sources (Lemma 3 of \cite{lyu2024explicit}). However, although not stated in their hypotheses, their proof implicitly relies on \LP, thus corresponding once again to Theo.~\ref{theo:noIID-LP1}.
Notably, all the incompatibilities discussed so far involve \IID together with \EI: dropping either of them removes the contradiction, a possibility we discuss in Sec.~\ref{sec:MR-ID-EI}. As we show next, no such escape route exists for the Target Chain rule.

A property which has long been highly desired is Target Chain rule \TC, considered by many to be the \textit{holy grail} of the PID axioms. 
Unfortunately, here we present a new finding showing that \TC is incompatible with the fundamental axioms on which the redundancy lattice is built. 
Specifically, we prove that \TC cannot coexist with the original Williams and Beer axioms \SR, \Sz, and \Mz, even for a bivariate PID (Theo.~\ref{theo:noSR-Sz-Mz-TC}; this result has been derived independently, under slightly weaker assumptions, in Prop.~3 of Ref.~\cite{matthias2025novel}). Since \LP implies \Mz (Prop.~\ref{prop:LP--M0}), this immediately yields the incompatibility of \TC with \LP (Cor.~\ref{cor:noSR-LP-TC}; see Theorem 1 of Ref.~\cite{matthias2025novel}), and recovers as a special case the earlier result of Ref.~\cite{finn2018pointwise}, in which it was shown the incompatibility of \TC with \SR, \Sz, \EI, \LPz, and \IID (Theorem~6 of Ref.~\cite{finn2018pointwise}, here Cor.~\ref{cor:LPz-ID--noTC}). Hence, no PID measure can satisfy the Williams and Beer axioms and the Target Chain rule at the same time. This explains why the only measures satisfying \TC (\IPM and \ISX) violate \Mz at the global level, retaining monotonicity only separately on their informative and misinformative lattices.

These results also shed new light on the relationship between \TC and \ID. Here, we further present a novel finding for which combining \TC with the intuitive \TE and the fundamental \SR and \Sz directly leads to the controversial \ID (and hence \IID) (Theo.~\ref{theo:TC-TE-SR--ID}). 
Since \LP implies \TE (Prop.~\ref{prop:LP--TE}), this generalises the existing result that \TC and \LP imply \ID, originally stated in Ref.~\cite{bertschinger2013shared} (Cor.~\ref{cor:TC-LP--ID}). 
However, by Cor.~\ref{cor:noSR-LP-TC}, no measure can satisfy the premises of this implication. More generally, by Theo.~\ref{theo:noSR-Sz-Mz-TC}, any PID function satisfying the hypotheses of Theo.~\ref{theo:TC-TE-SR--ID} must violate \Mz. 
Moreover, by additionally requiring \GP, we further obtain \TM (see Prop.~\ref{prop:TC-GP--TM}), which was thought to be hard to reconcile with \ID~\cite{rauh2014reconsidering}.

Thus, if abandoning \IID does not resolve these issues, how else can one avoid such undesired behaviours?
Taken together, these results suggest that the properties that classical information measures benefit from, such as \TC and \LP, should not be expected to hold for PID measures. This perspective can be more naturally accommodated by redundancy measures that give operational meaning and clear interpretation to negative atoms, such as \ICCS, \IPM, and \ISX, though at the cost of losing \Mz. 
Moreover, another possible point of reflection is that the mechanistic intuition underlying \IID might be fundamentally incompatible with the purely statistical perspective embodied by \EI. Consequently, despite the natural appeal of \EI, relaxing this property could offer a viable route to alleviate these issues. 
We refer to Sec.~\ref{sec:MR-ID-EI} for further discussion on these points.

In closing, we also mention a couple of novel results regarding the PID properties \TE and \BP. 
\TE is justified by the intuitive idea that adding the target to the set of sources should not affect redundancy. Interestingly, combining \TE with the fundamental \Mz provides a necessary and sufficient condition for \Mo, supporting the claim that \Mo is a useful desideratum \cite{griffith2014intersection} (Theo.~\ref{theo:M0-TE--M1}). 
Finally, we note that the decision-making-inspired \BP is a stronger condition of both \TE and \SE, at least for $n=2$ sources (Theo.~\ref{theo:BP--TE-SE}). Hence, if one accepts the operational idea underpinning \BP, then \TE follows directly, as well as its implications discussed above.

We represent all these results and additional existing relations and incompatibilities in a hypergraph (Fig.~\ref{fig:nogo}), also reporting which properties are the most commonly involved.

To further investigate the logical relationships among the properties, we encoded them in the Z3 automatic theorem prover~\cite{de2008z3}, a satisfiability modulo theories program that allows us to consistently check implications and incompatibilities by deciding satisfiability of combinations of axioms. 
With this tool at hand, we first verified that the properties satisfied by each measure (Table~\ref{tab:PID_properties}) are indeed compatible with the known relationships among the PID axioms.

Furthermore, we used Z3 to determine the maximal sets of mutually compatible properties, i.e.\ sets to which no property can be added without introducing an incompatibility.
Assuming the Williams and Beer axioms \SR, \Sz, and \Mz, then \TC is always excluded (Theo.~\ref{theo:noSR-Sz-Mz-TC}). For $n=2$ sources, this is the only property that must be abandoned, yielding a unique maximal set of 19 out of the 20 properties. For an arbitrary number of sources $(n>2)$, Theos.~\ref{theo:noIID-LP1}-\ref{theo:noSo} further restrict the admissible combinations, and three maximal sets emerge: abandoning \LPo leaves 18 properties, whereas retaining it requires dropping \So together with either \EI (17 properties) or both \IID and \ID (16 properties). 
On the other hand, dropping \Mz leads to a unique maximal set of 17 properties, in which \LP (Prop.~\ref{prop:LP--M0}) and \Mo (Theo.~\ref{theo:M0-TE--M1}) need to be abandoned, but \TC becomes admissible. Notably, this is the region of the landscape occupied by \IPM and \ISX, the only measures satisfying \TC.

This analysis carries two important implications. First, it encourages the development of new PID measures capable of satisfying a larger number of axioms, enabling researchers to determine the maximal set of compatible properties given certain desirable features. Second, the fact that no existing PID measure achieves the maximal set suggests the existence of yet undiscovered incompatibility relationships among the properties.

\section{Discussion} \label{sec:discussion}

\subsection{Identity property, Equivalence-class Invariance, and Mechanistic redundancy}\label{sec:MR-ID-EI}
As we have seen above, if one wants to retain the fundamental \Sz, \SR, and the convenient \LP, then \EI is not compatible with \ID or \IID, at least for $n\ge3$ sources (Theo.~\ref{theo:noIID-LP1}). %
The conceptual inconsistency between \EI and \IID already transpired from the conflicting interpretations briefly outlined in Sec.~\ref{sec:properties}. Here we further discuss the implications.

A concept strictly related to \ID, and \EI is that of \textit{mechanistic redundancy}. Introduced by Harder \etal~\cite{harder2013bivariate}, mechanistic redundancy indicates the redundant information that is generated by the mechanism, as opposed to the \textit{source redundancy} which stems from the correlation between the sources. The existence of mechanistic redundancy entails that redundant information can be non-zero even if the sources are independent. However, as seen above, if \EI is adopted, then the prospect of measuring mechanistic redundancy vanishes, since different mechanisms can lead to the same probability distribution, up to an isomorphism~\cite{pearl2014probabilistic}.
Thus, although \EI has not received much attention in the PID literature, we believe it is fundamental to investigate its significance and ponder on which scenarios it should or should not be expected to hold. 
To the best of our knowledge, the only work which explicitly argues in favour of \IID as opposed to \EI is that of Chicharro~\cite{chicharro2018identity}. In this work, the author specifically investigates the property of PID being invariant under isomorphisms of the target, concluding that if mechanistic redundancy is to hold, then redundancy depends on the composition of the target and one should drop the Target to Source Copy (TSC) isomorphism, i.e.\ a specific instance of \EI applied to the target. 
In fact, contrary to the strong axiom, the weak version is compatible with the violation of \EI, since it allows for dependencies on the semantic properties of the outcomes~\cite{chicharro2018identity}. %
Moreover, he also derives that the negativity of atoms follows from the existence of deterministic target-source dependencies and \EI~\cite{chicharro2018identity}. Hence, accepting the semantic value of specific realisations can also resolve the non non-negativity of the decomposition. 

\subsection{Monotonicity and the mirage of the Target Chain rule}
\label{sec:LP_TC}

A fundamental result of this work is that \TC is incompatible with any PID satisfying the original Williams and Beer axioms \SR, \Sz, and \Mz, and hence with any PID satisfying \LP (Theo.~\ref{theo:noSR-Sz-Mz-TC} and Cor.~\ref{cor:noSR-LP-TC}; see also Ref.~\cite{matthias2025novel}). This incompatibility is of a different nature from the ones discussed in Sec.~\ref{sec:MR-ID-EI}. Theo.~\ref{theo:noIID-LP1} requires at least three sources, leverages the relabelling of outcomes allowed by \EI, and involves the contested \IID; each of these ingredients offers a possible escape route. By contrast, Theo.~\ref{theo:noSR-Sz-Mz-TC} holds already for two independent fair bits, involves no relabelling, and uses only the axioms on which the redundancy lattice is built, together with a property that mirrors one of the core properties of classical information-theoretic quantities. Hence, the incompatibility is neither an artefact of the complex algebraic structure of the lattice for $n \ge 3$ sources, nor a consequence of PID-specific desiderata: the chain rule, a defining property of mutual information, cannot be inherited by a monotonic redundancy function.

This happens as \TC requires redundancy about a composite target to be path-independent: the redundancy accumulated about $Y_1$ and then about $Y_2$ given $Y_1$ must equal that accumulated in the reverse order. However, conditioning on a component of the target can radically change the relationship between the sources. In the counterexample of Theo.~\ref{theo:noSR-Sz-Mz-TC}, conditioning on $Z = X_1 \oplus X_2$ turns two independent sources into perfect copies of each other, a well-known ``explaining away'' effect, whereas conditioning on $X_1$ turns one of the sources into a constant. Monotonicity then constrains the redundancy of each conditional system to opposite extremes: full redundancy for copies, and no redundancy when one source carries no information. As a result, the two paths accumulate one bit and zero bits of redundancy, respectively. Mutual information is immune to this phenomenon, since it involves a single source, whose relationship with other sources plays no role. Redundancy, instead, is intrinsically relational: it quantifies what the sources have in common, and this commonality can be created or destroyed by conditioning on the target. In this sense, \TC requires redundancy to behave as an amount of information that can be accumulated piecewise, while \Mz needs it to behave as a shared content, and these two requirements cannot be reconciled.

These findings have direct practical consequences, as one is forced to choose which property to give up. A first option is to abide by the original Williams and Beer axioms (\SR, \Sz, \Mz) and give up \TC. In this case, PID results about a composite target cannot be computed or interpreted piecewise. In many applications the target is naturally multivariate, e.g.\ the joint future state of a dynamical system or several features of a stimulus. In these cases, the redundancy about the whole target is in general not the sum of the redundancy about one component and the conditional redundancy about the rest. Hence, any sequential analysis, e.g.\ conditioning on one stimulus feature before decomposing the information about another, yields results that depend on the order in which the components are considered, and cannot be combined with inferences about the whole target.

A second option is to retain \TC at the price of \Mz, and hence of \LP. Measures satisfying \TC, such as \IPM and \ISX, allow such piecewise decompositions, but can yield negative redundancy values, which must then be interpreted as misinformation~\cite{finn2018pointwise, makkeh2021introducing}. Which of these two options is preferable depends on whether the scientific question concerns how information about a composite target splits across its components, or rather requires an operational interpretation of the atoms in communication-theoretic terms (Sec.~\ref{sec:practical}).

Finally, a third alternative is to abandon the consistency equations, and thus the lattice structure, altogether. An example of this route is the approach of Kolchinsky~\cite{kolchinsky2022novel}, where redundancy and union information are defined independently of each other, without being related through the inclusion-exclusion principle, at the cost of losing a decomposition into atoms that sum to the mutual information. A more drastic alternative would be to let redundancy depend on the order in which the components of the target are listed, i.e.\ to stop treating the joint target as a single random variable.

\subsection{Practical guidance for empirical applications} \label{sec:practical}

With the many approaches available in the PID literature, each satisfying different properties, the endeavour of choosing a PID measure for empirical purposes might seem an overwhelming task. In this section, we aim to alleviate this burden by providing explicit, practical recommendations tailored to the specific application under consideration and to the nature of the sources and target. 

A possible starting point could be to reflect upon what kind of information independent sources are expected to provide. 
For example, PID is often employed in the study of dynamical processes evolving over time, with the sources being the states of two elements of the system in the past, and the target being the joint state of those components in a future timestep, rather than another variable \cite{faes2025predictive, liardi2025null}. In this respect, under a mechanistic perspective, one could expect that independent sources provide unique information to the target and zero redundancy or synergy, thus satisfying \IID. This reasoning alone already helps restrict the space of possible PIDs to choose from. 

Another key point to consider is the relevance of Local Positivity \LP, which is closely tied to the interpretation one wants to assign to the PID atoms~\cite{chicharro2018identity}. If PID atoms are interpreted operationally in the classical communication-theoretic sense, \LP becomes essential, as the atoms represent quantities of information that can be transmitted or exploited, with negative values lacking an operational meaning. Within this view, the informational architecture is necessarily grounded in the statistical structure of the variables alone, without reference to semantic content, labels, or specific realisations, hence in accordance with \EI. 
Belonging to this paradigm is the operational interpretation of PID atoms in terms of secret-key agreement rates, where information is quantified in terms of the ability of parties to establish shared keys under public communication constraints \cite{james2018uniquekey}. Within this view, one can distinguish between the so-called \textit{camel} and \textit{elephant} intuitions. The former considers the sources as contributors that jointly construct the outcome of the target, whereas the latter interprets the sources as partial and potentially overlapping observations of the target. Importantly, these two perspectives induce different directional interpretations of redundancy, and certain apparent shortcomings of PID measures arise only when a decomposition is assessed through one intuition rather than the other \cite{james2018uniquekey, banerjee2018unique}.

Conversely, if information is interpreted as a measure of statistical dependence, negative local contributions correspond to misinformation, i.e.\ a higher uncertainty about the target after observing the sources. 
Along this line, pointwise PID approaches relax \LP deliberately, favouring an interpretation of PID atoms in terms of local dependencies and belief updates, rather than communication capacity~\cite{ince2017measuring, finn2018pointwise, makkeh2021introducing}. 
Moreover, within this perspective, since PID atoms are associated with changes in the target distribution putatively induced by the sources, a semantic dependency based on the specific value of the variables might be contemplated, objecting \EI to retain mechanistic interpretations.
Therefore, whether \LP should be pursued depends on the intended interpretation and operational significance of the decomposition.

Finally, additional considerations relating \AD and \CO can be relevant in empirical settings characterised by noise and multiscale organisation, such as many biological systems \cite{walpole2013multiscale, dada2011multi}. In these contexts, \AD can be desirable when combining information contributions across different scales, e.g.\ by performing PID analyses considering microscopic and macroscopic states, as it allows to disentangle independent contributions \cite{liardi2025simple}.
Likewise, \CO plays an important practical role in the presence of sampling noise and estimation errors, guaranteeing that small perturbations in the underlying distributions lead to small changes in the resulting PID atoms. This robustness is particularly valuable when dealing with high-dimensional data with a limited sample size, where sharp discontinuities in the PID of choice may hinder interpretability or comparability across conditions. 

An additional consideration that further restricts the set of admissible measures is their domain of definition. Some PID measures are defined only for discrete variables, whereas others are formulated for continuous Gaussian settings, and only a small subset extends to generic continuous distributions \cite{ehrlich2024partial}. Moreover, several discrete PID definitions require the underlying distribution to have full support, imposing additional constraints on their applicability.

\subsection{Other approaches to information decomposition}
\label{sec:other_approaches}

In this paper, we centred our analysis on the Williams and Beer redundancy lattice, as it is the original and most widely used formalism for studying synergy and redundancy in recent work \cite{williams2010nonnegative}.
Accordingly, our results concern information decomposition measures based on this lattice.
Nevertheless, we acknowledge the existence of alternative approaches that investigate information-sharing relationships through different means.
In particular, some frameworks decompose mutual information using alternative lattice structures; others focus on decomposing functionals different from mutual information (e.g.\@, entropy-based decompositions); and others characterise synergy and redundancy without performing an explicit decomposition, relying instead on linear combinations of Shannon measures or other statistical quantities.

Approaches that decompose mutual information based on alternative lattice structures or different theoretical frameworks have been proposed for different purposes. 
Some are motivated by the goal of addressing specific limitations of the original PID lattice, while others aim to provide more general formalisms or to bridge disparate theoretical perspectives.
For instance, the constraint lattice originally introduced by Zwick in reconstructability analysis \cite{zwick2004overview} was later employed by James \etal to define \IDEP \cite{james2018unique}, and has since become widely adopted for information decomposition purposes.
In fact, Ay \etal proposed an extended version of the lattice by rooting in cooperative game theory, addressing the incompatibility between \LP and \ID \cite{ay2021information}. 
While this construction resolves the conflict between these two central properties, it achieves so by collapsing several redundant atoms of the PID lattice, thereby reducing the level of detail captured by the decomposition.
Similarly, Rosas \etal leveraged this extended constraint lattice to introduce a synergy-first decomposition, based on an operational notion of synergy motivated by recent developments in the literature of data privacy \cite{rosas2020operational}. 
Due to the superexponential growth of the number of atoms, they further proposed a backbone lattice, which summarises the decomposition by collapsing multiple terms into a one-dimensional representation.
This backbone construction was subsequently employed by Varley, who introduced a scalable synergy-based decomposition in which synergy is quantified in terms of the system’s resilience to the removal of its components \cite{varley2024scalable}.

Other proposals investigate the simultaneous use of multiple lattices, exploiting their complementarities by mapping their relationships \cite{chicharro2017synergy, kolchinsky2022novel}.
Within this line of work, Chicharro \etal proposed a synergy-first decomposition of mutual information, introducing a lattice based on mutual information loss dual to the original redundancy one, where synergy and redundancy contributions are inverted \cite{chicharro2017synergy}.
A similar approach has been explored by Kolchinsky, who posits the idea that the ``inclusion-exclusion principle'', upon which the PID lattice is constructed, should not be expected to hold, thus suggesting a double decomposition of union and intersection information \cite{kolchinsky2022novel}.
Beyond the inversion of synergy and redundancy, Pica \etal proposed analysing all possible PID decompositions obtained by swapping sources and target, then extracting invariant informational components shared across the different combinations \cite{pica2017invariant}.
From a different perspective, Perrone \etal introduced an information-geometric approach to capture the hierarchical organisation of interactions, offering an alternative viewpoint on how synergy can be quantified across different orders of dependencies \cite{perrone2016hierarchical}. 
More recently, an alternative lattice that removes singleton nodes has been introduced by Lyu \etal \cite{lyu2025whole}. 
This reduced lattice focuses on higher-order synergistic and unique contributions, at the cost of not resolving redundancy terms, resulting in a coarser decomposition. 

Other frameworks depart from decomposing mutual information altogether. 
To investigate how information is structured across multiple sources and targets, Mediano \etal developed the $\Phi$ID formalism \cite{mediano2025toward}. $\Phi$ID extends the PID redundancy lattice to a double-redundancy lattice, introducing additional atoms that describe novel modes of informational contributions. 
However, while this framework allows a finer-grained decomposition of information, it also requires the definition of a double-redundancy function.
A different direction consists of entropy-based decompositions such as Partial Entropy Decomposition (PED), which avoid the explicit distinction between sources and targets and instead decompose the overall informational structure of a set of random variables \cite{ince2017partial, varley2023partial}. 
Along similar lines, a more general proposal that embeds both entropy and mutual-information decompositions is the Generalised Information Decomposition (GID) proposed by Varley \cite{varley2024generalized}. GID provides a unifying perspective by performing a decomposition based on the Kullback–Leibler divergence, enabling a comparison of how different lattice-based approaches relate to one another. 

Finally, several approaches characterise synergy and redundancy without relying on lattices at all. While these approaches offer advantages, either in not relying on specific definitions of redundancy or being less computationally expensive, they also do not allow for decomposing information into all the atoms defined by PID.
Quax \etal introduced an operational definition of synergy based on synergistic auxiliary variables, providing a constructive way to isolate purely synergistic information \cite{quax2017quantifying}. 
Olbrich \etal proposed a geometric approach that decomposes synergistic contributions by projecting probability distributions onto exponential families \cite{olbrich2015information}.
Other methods quantify synergy and redundancy using only linear combinations of Shannon entropies or mutual information, thereby avoiding the need to define a redundancy function or a decomposition lattice. 
These approaches have the advantage of requiring no additional assumptions beyond classical information theory. 
Prominent examples include the O-information \cite{rosas2019quantifying }, the redundancy–synergy index \cite{chechik2001group}, and the coinformation \cite{bell2003co}, which identify whether synergy or redundancy is the dominant mode of information sharing, and the Total Correlation hierarchical decomposition \cite{amari2001information}, the neural information decomposition \cite{reing2021discovering}, the connected information \cite{schneidman2003network, schneidman2003synergy, schneidman2006weak} and the symmetric decomposition \cite{rosas2016understanding}, which instead aim to disentangle dependencies of different orders. 
A recent unifying framework formalises this class of measures by defining Shannon invariants as all combinations of informational atoms that can be directly computed as linear combinations of entropy or mutual information \cite{gutknecht2025shannon}.

An additional approach developed to study synergistic and redundant effects is the Logarithmic Decomposition (LD) by Down and Mediano~\cite{down2025logarithmic,down2025algebraic}. LD extends the M\"obius inversion to create a signed-measure space for classical information measures, where variables correspond to measurable sets on which standard set-theoretic operations can be performed. Using LD, the authors showed that the dyadic and triadic systems of James and Crutchfield~\cite{james2017multivariate} -- famously thought to be classically indistinguishable -- can in fact be algebraically distinguished by interrogating classical information measures carefully~\cite{down2025logarithmic}.

\subsection{Limitations and future work}
Despite the systematic and comprehensive nature of our work, we acknowledge some limitations.
While we presented the most complete map of relationships among PID axioms to date, collecting all the properties mentioned in the literature and introducing new theorems, it remains possible that additional implications or inconsistencies exist and have not been identified by our analysis.
Moreover, our analysis is primarily conceptual and theoretical, leaving for future work a systematic empirical evaluation across real-world systems.
In fact, while we focused on comparing formal properties and logical relationships between different measures, we did not assess their practical behaviour or performance on empirical data when addressing concrete scientific questions.
Future studies might further investigate the relationship between the different informational approaches of this study and the notions of causality and emergence, which are central to the study of complex systems.
In addition, although we concentrated on the standard PID lattice, future work could extend this analysis to alternative lattice structures and to decompositions of information-theoretic quantities beyond those considered here.

\section{Conclusion}

In this work, we provided a unified and systematic analysis of Partial Information Decomposition approaches grounded in the Williams and Beer redundancy lattice, organising the existing PID measures and properties and clarifying the relationships among them.
By collecting existing results, identifying previously implicit connections, and introducing new theorems, we offer a coherent map of the conceptual landscape underlying the quantification of redundancy, synergy, and unique information within this widely used formalism.
Beyond serving as a reference for existing PID measures, this analysis highlights the structural constraints that shape possible decompositions and clarifies which desiderata can, or cannot, be simultaneously satisfied.
We hope that this perspective will help guide both the development of new information-decomposition frameworks and the informed application of existing ones, contributing to a clearer understanding of multivariate information sharing in complex systems.

\subsubsection*{Acknowledgments}
We thank Hardik Rajpal, Gustavo Menesse, Fernando Rosas, Davide Orsenigo, Maximilian Kathofer, and all the members of the Imperial College MIND Lab for the useful discussions and precious comments. We also thank Aaron J. Gutknecht, Philip H. Matthias, and Artemy Kolchinsky for suggestions that helped improve this work.
M.N. has received funding from the French government through the ‘France 2030’ investment plan managed by the French National Research Agency (Agence Nationale de la Recherche; reference: ANR-16-CONV000X/ANR-17-EURE-0029) and from the Excellence Initiative of Aix-Marseille University-A*MIDEX (AMX-19-IET-004). M.N. has received funding from the Boehringer Ingelheim Fonds (BIF). G.B. is supported by the Leverhulme Trust.

\subsubsection*{Code availability}
The code used to perform the analyses with the Z3 automatic theorem prover is publicly available in the open-source repository at \url{https://github.com/Imperial-MIND-lab/pid-prover}.

\clearpage
\newpage
\appendix

\setcounter{footnote}{3}

\section{Notation and mathematical background} \label{appendix:mathematical_bg}
In this work, we make use of the following notation. 
We indicate a generic probability distribution for the random variables $\Xs$ as $P(\Xs)$, or equivalently as $P_{X_1\ldots \,X_n}$. On the other hand, we denote the probability associated with a specific outcome ${X_1=x_1,\ldots,X_n=x_n}$ as ${p(x_1,\ldots,x_n)}$, ${p(X_1=x_1,\ldots,X_n=x_n)}$, or ${P(X_1=x_1,\ldots,X_n=x_n)}$. Throughout the work, we use these notations interchangeably.
Unless otherwise specified, within PID the sources will be indicated as $\Xs$, while the (possibly multivariate) target is $Y\equiv (Y_1,\ldots,Y_n)$.
We indicate both the mutual information and the coinformation with the letter $I(\cdot)$, since the latter reduces to the former in the case of two arguments. The mutual information can be both expressed in terms of the probability of each outcome, or in terms of entropies:
\begin{equation}
\begin{aligned}
    I(X;Y) & = H(X) - H(X|Y) = H(X) + H(Y) - H(X,Y) \\
    & = \sum_{xy} p(x,y) \log \frac{p(x,y)}{p(x)p(y)} \,.
\end{aligned}
\end{equation}
For the coinformation, we write it following the convention of interaction information, i.e.\
\begin{equation}
    I(X_1;\ldots;X_{n+1}) = I(X_1;\ldots;X_n) - I(X_1;\ldots;X_n\mid X_{n+1}) \,.
\end{equation}
For three arguments, this reduces to 
\begin{equation}
    I(X_1;X_2;X_3) = I(X_1;X_3)+I(X_2;X_3)-I(X_1,X_2;X_3) \,.
\end{equation}
If we interpret ($X_1,X_2$) as sources and $X_3$ as a target of a PID, Williams and Beer showed that the coinformation corresponds the difference between redundancy and synergy, thus explaining the until then unexplained negative values obtained by the coinformation~\cite{williams2010nonnegative}.
To distinguish between multivariate and multiple arguments, we respectively make use of the comma $,$ and the semicolon $;\,$. For instance, $I(X_1,X_2,X_3;Y)$ is the mutual information between $(X_1,X_2,X_3)$ and $Y$, while $I(X_1,X_2,X_3;Y)$ is the coinformation between $(X_1,X_2)$, $X_3$ and $Y$. 

\subsection{Equivalence relations and partial ordering}

In Ref.~\cite{li2011connection}, an equivalence relation between random variables is proposed. Intuitively, this predicates that two random variables are informationally equivalent if and only if they can be mapped to each other through a one-to-one correspondence. Clearly, this induces an equivalence relation:
\begin{equation} \label{eq:def_equivalence}
    X\sim Y  \iff \exists, f \text{ s.t. } X=f(Y), \text{ and } Y=f^{-1}(X)\,.
\end{equation}

A similar idea was already introduced by Shannon to define an \textit{ordering} between random variables \cite{shannon1953lattice}. 
Within this approach, a random variable $X$ is more informative than $Y$ if there exists a function $f$ such that $Y=f(X)$. We indicate this as $Y\leq_I X$:
\begin{equation}
    Y\leq_I X\iff \exists\, f \text{ s.t. } Y=f(X) \,.
\end{equation}
If $f$ is bijective and can be inverted, then equality holds and we obtain the equivalence relation of Eq.~\eqref{eq:def_equivalence}. 

However, there are many other ways to define an ordering between random variables \cite{gacs1973common, griffith2015quantifying, aumann2016agreeing}.
Another fundamental approach is the one based on Blackwell sufficiency and operationalised by Blackwell's theorem~\cite{blackwell1953equivalent}. 
Blackwell sufficiency is a key concept of statistical decision theory~\cite{blackwell1953equivalent, torgersen1991comparison}, which formalises the idea that one random variable is at least as informative as another for any decision-making problem with respect to a given target, in the sense that any strategy based on the latter can be replicated, possibly after randomisation, using the former. 
This setup leads to Blackwell’s theorem~\cite{blackwell1953equivalent, rauh2017coarse}, which provides the following definition: a random variable $X_1$ is Blackwell sufficient for $X_2$ relative to the target $Y$ if there exists a stochastic channel $k$ such that
\begin{equation} \label{eq:def_blackwell}
    p(x_2|y) = \sum_{x_1} k(x_2|x_1) p(x_1|y) \,.
\end{equation}
Hence, intuitively, since the channel $p(x_2|y)$ can be obtained as a linear combination of $p(x_1|y)$, this operationalises the idea that $X_2$ does not provide any more information than $X_1$ about the target $Y$. We indicate this relation as $X_2\preceq_Y X_1$.
In other contexts, this condition is sometimes phrased by saying that the channel $p(x_1|y)$ is stochastically degraded w.r.t.\ $p(x_2|y)$~\cite{venkatesh2022partial}, or that $X_2$ is a (degraded) garbling of $X_1$ relative to $Y$. 
This operational reasoning has been the foundation of the Blackwell property \BP and various PID definitions~\cite{bertschinger2014quantifying, rauh2017extractable, venkatesh2022partial, kolchinsky2022novel, mages2023decomposing}.

\begin{figure}[t]
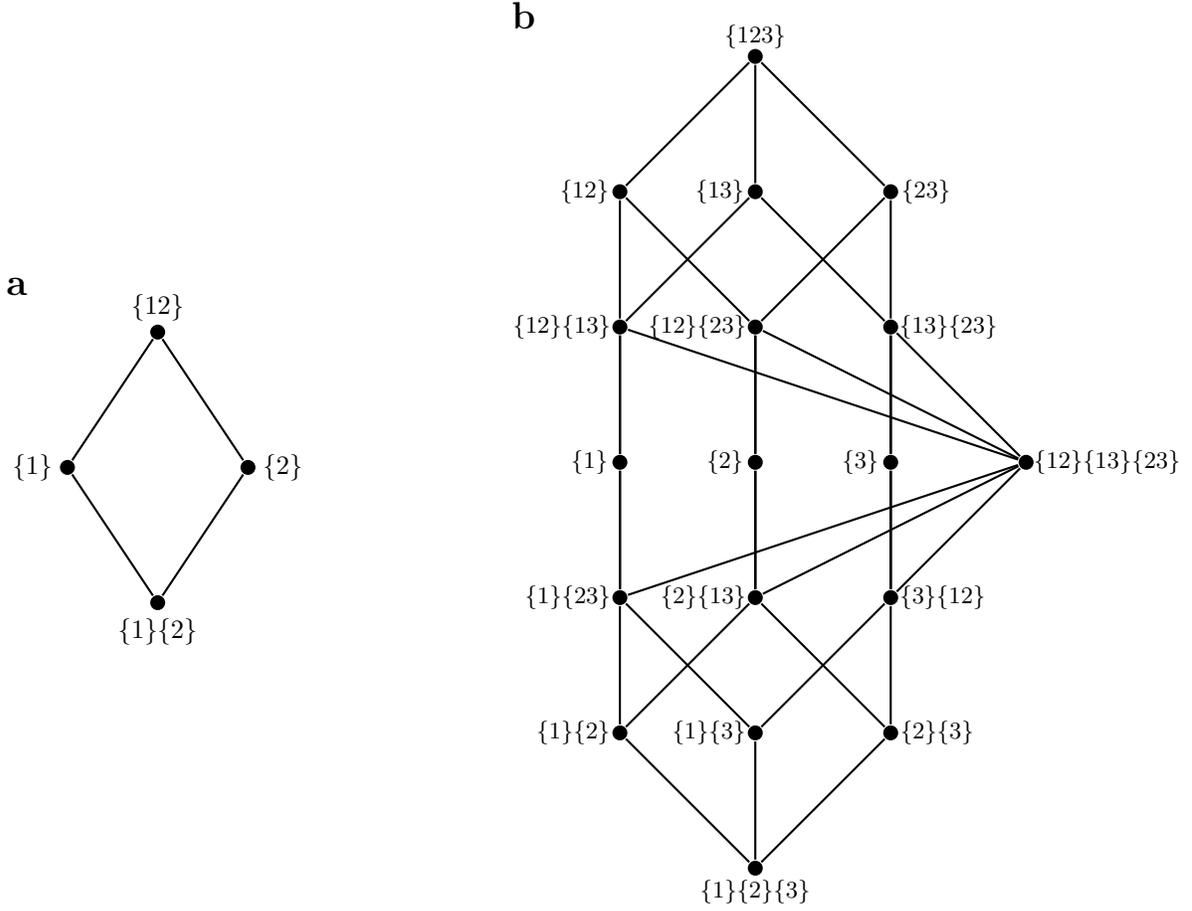

    \centering
    \begin{subfigure}[c]{0.38\textwidth}
        \centering
        \input{PID_N=2.tkz}
        \label{fig:pid-n2}
    \end{subfigure}
    \hfill
    \begin{subfigure}[c]{0.58\textwidth}
        \centering
        \input{PID_N=3.tkz}
        \label{fig:pid-n3}
    \end{subfigure}
    \caption{Partial information decomposition lattices for \textbf{a)} $n=2$ and \textbf{b)} $n=3$ sources.}
    \label{fig:pid-lattices}
\end{figure}

\subsection{Redundancy lattice and PID atoms}
Partial information decomposition (PID) is a framework designed to refine the notion of multivariate information by distinguishing how multiple sources jointly inform a target. Rather than treating the mutual information $I(\Xs;Y)$ as a single quantity, PID separates it into components that reflect information shared across sources, exclusive to individual sources, or that emerges from the joint consideration of several sources. These correspond to the so-called redundancy, unique information, and synergy. In large systems, these contributions lead to a combinatorial structure of informational terms.

Williams and Beer introduced a sophisticated way to organise these components by exploiting an ordering of the redundancy relations inspired by the inclusion–exclusion principle (IEP). 
Given two sets $X_1$ and $X_2$, the IEP states that the sizes of their union and intersection are not independent, being related as
\begin{equation}
    |X_1\cup X_2| = |X_1|+|X_2|-|X_1\cap X_2| \,.
\end{equation}
More generally, for $n$ sets, the IEP reads
\begin{equation}
\left| \bigcup_{i=1}^{n} X_i \right| = \sum_{i=1}^{n} |X_i| - \sum_{1 \le i < j \le n} \left |X_i \cap X_j\right | + \sum_{1 \le i < j < k \le n} |X_i \cap X_j \cap X_k| - \dots + (-1)^{n-1} |X_1 \cap \dots \cap X_n| \,.
\end{equation}
This setup enables a set-theoretic interpretation of information-theoretic quantities \cite{mcgill1954multivariate, yeung1991new} and, within PID, it defines a direct link between redundancy and synergy \cite{williams2010nonnegative}.

Hence, the key idea is to enumerate all meaningful ways in which information can be common to groups of sources, then arrange them within a partially ordered set. This construction gives rise to the so-called redundancy lattice, whose elements correspond to distinct notions of shared information.

Formally, the lattice is constructed from the following collections of sources:
\begin{equation}
\mathcal{A}(\Xs) \equiv { \alpha \in \mathcal{P}(\mathcal{P}(\Xs)) \mid \forall A_i,A_j \in \alpha,; A_i \not\subseteq A_j } \,
\end{equation}
where $\mathcal{P}(\mathcal{X})$ denotes the set of all non-empty subsets of $\mathcal{X}$. Each element $\alpha \in \mathcal{A}$ specifies a collection of source groups that define a candidate redundancy relationship, which we indicate as $\Ic(\alpha;Y)$.
The set $\mathcal{A}$ is also in correspondence with all so-called \textit{antichains} of the set ${X_1,\ldots, X_n}$ \cite{dilworth1990decomposition}. As such, for each antichain there exists a corresponding piece in the Partial Information Decomposition.

Furthermore, this set of antichains can then be endowed with a partial order given by 
\begin{equation}
    \alpha \preceq \beta \iff \forall B \in \beta \,\,\,\exists A\in\alpha \text{ s.t. } A\subseteq B \,,
\end{equation}
with $\alpha,\beta \in \mathcal{A}$. Intuitively, this ordering reflects the idea that redundancy defined over smaller or more specific source groups refines redundancy defined over larger ones.

Throughout this work, we denoted with \Ic the redundancy function, and as $I_\partial$ the PID atoms obtained by the decomposition. Given a collection of sources $\alpha$, the fundamental construction of the redundancy lattice imposes that:
\begin{equation} \label{eq:pid_lattice_defeq}
    \Ic(\alpha;Y) = \sum_{\beta\preceq\alpha} I_\partial(\beta;Y) \,,
\end{equation}
which then leads to the following recursive relation to find the PID atoms $I_\partial$:
\begin{equation}
    I_\partial(\alpha;Y) = \Ic(\alpha;Y) - \sum_{\beta\prec\alpha} I_\partial(\beta;Y) \,.
\end{equation}
To fix ideas, consider the collection $\alpha=\{\{X_1\},\{X_2,X_3\},\{X_4, X_5\}\}$, with a small abuse of notation, we also indicate this as $\alpha=\{\{1\},\{2,3\},\{4 ,5\}\}$. The corresponding redundancy function will be $\Ic(X_1,X_2X_3,X_4 X_5;Y)$, which indicates the shared information among the sets $\{X_1\}$,$\{X_2,X_3\}$, and $\{X_4,X_5\}$, where $\{X_i,X_j\}$ also include the synergistic dependencies between $X_i$ and $X_j$.
Hence, the only purely redundant atom is the one following from the set $\alpha=\{\{X_1\},\ldots,\{X_n\}\}$, whose redundancy function is $\Ic(\Xs;Y)$.

We remark the difference between $\Ic(X_i,X_j;Y)$ and $I_\cap(X_iX_j;Y)$: the first is the redundancy function between the sources $X_i$ and $X_j$, while the second is the redundancy between the composite variable $(X_1X_2)$ and the target, and is equal to the mutual information $I(X_1,X_2;Y)$ (by \SR).

In the simplest case of $n=2$ sources, the PID lattice is composed of only four atoms (Fig.~\figsubref{fig:pid-lattices}{a}), and the underlying equations directly read:
\begin{align}
        \Icap & = I_\partial(X_1,X_2;Y) \\
    I(X_1;Y) & = I_\partial(X_1,X_2;Y) + I_\partial(X_1;Y) \\
    I(X_2;Y) & = I_\partial(X_1,X_2;Y) + I_\partial(X_2;Y) \\
    I(X_1,X_2;Y) & = I_\partial(X_1,X_2;Y)+I_\partial(X_1;Y)+I_\partial(X_2;Y)+I_\partial(X_1X_2;Y) \,.
\end{align}
Although PID is mathematically well-defined for any number of sources, performing the decomposition on systems with multiple sources becomes quickly infeasible, with the number of atoms increasing with the Dedekind number. For instance, already for $n=3$ sources there are 18 atoms (Fig.~\figsubref{fig:pid-lattices}{b}), and 166 for $n=4$. 

As discussed in detail throughout this work, there are many approaches to defining the $\Icaps$, each characterised by distinct properties. In turn, such properties are linked via specific relationships and incompatibilities. 
We provide a visual overview of the historical development of PID in Fig.~\ref{fig:timeline}, indicating when each measure, property, and mutual relationship was introduced. 
\begin{figure}[t]
    \centering
    \includegraphics[width=1.0\linewidth]{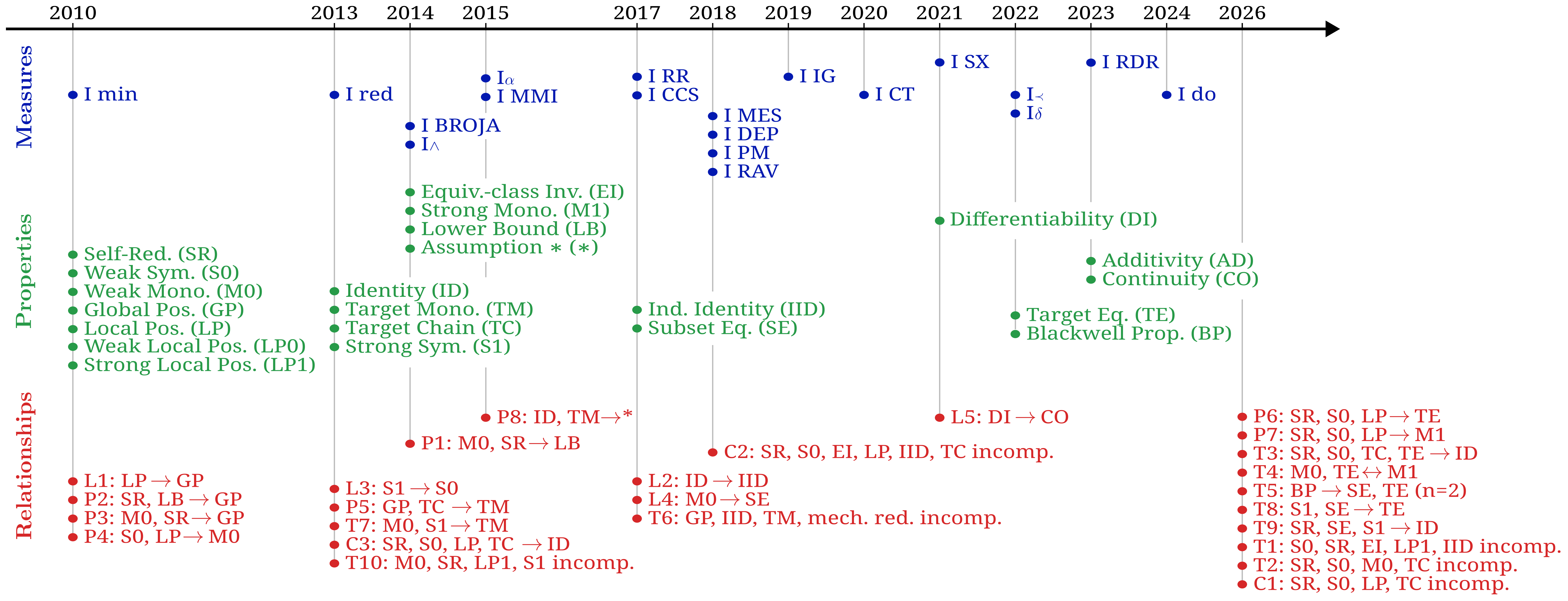}
    \caption{\textbf{Temporal timeline of the introduction of PID measures, properties, and their relationships.} $^*$ refers to the inequality: $I_\cap(X_1,X_2;f(X_1,X_2))\le I_\cap(X_1;X_2)$ (see Prop.~\ref{prop:ID-TM--ineq}).}
    \label{fig:timeline}
\end{figure}
Further information about the measures, properties, and interrelationships can be found in the following sections.

\section{PID properties}
\label{sec:properties_mathy}
\begin{itemize}
    \item \textbf{(SR)}: Self-redundancy~\cite{williams2010nonnegative}. \\
    The redundancy of a single source equals its mutual information:
    \begin{equation} \label{ax:SR}
        I_\cap(X_i;Y) = I(X_i;Y) \,.
    \end{equation}

    \item \textbf{(S\ts{0})}: Weak Symmetry~\cite{williams2010nonnegative}. \\
    Redundancy is invariant under any permutation of the sources:
    \begin{equation}\label{ax:S0}
        \Icaps = I_\cap(\sigma(X_1),\ldots,\sigma(X_n);Y) \,.
    \end{equation}
    for any permutation $\sigma$.
    
    \item \textbf{(M\ts{0})}: Weak Monotonicity~\cite{williams2010nonnegative, griffith2014intersection}. \\
    Adding a source to the redundancy function does not increase redundancy, and removing a source more informative than another preserves redundancy:
    \begin{equation}\label{ax:M0}
        \Icaps \leq \Icapsm \,,
    \end{equation}
    with equality if there exists $Z\in\{\Xsm\}$ such that $Z=f(X_n)$, i.e.\ $Z\leq_I X_n$.

    \item \textbf{(GP)}: Global Positivity~\cite{williams2010nonnegative}. \\
    Redundancy is non-negative:
    \begin{equation}\label{ax:GP}
        \Icaps\ge0 \,.
    \end{equation}

    \item \textbf{(LP\ts{0})}: Weak Local Positivity~\cite{williams2010nonnegative}. \\
    The atoms $I_\partial$ are non-negative for a bivariate PID:
    \begin{equation}\label{ax:LP0}
        I_\partial(\alpha;Y)\ge0\,,
    \end{equation}
    where $\alpha$ indicates any collection of subsets of $n=2$ sources, i.e.\ $\alpha\in\mathcal{A}(X_1,X_2)=\{\{\{1\},\{2\}\},\{1\},\{2\},\{1,2\}\}$.

    \item \textbf{(LP\ts{1})}: Strong Local Positivity~\cite{williams2010nonnegative}. \\
    The atoms $I_\partial$ are non-negative for a PID with $n>2$ sources:
    \begin{equation}\label{ax:LP1}
        I_\partial(\alpha;Y)\ge0 \,,
    \end{equation}
    where $\alpha$ indicates any collection of subsets of $n\ge3$ sources. If both \LPz and \LPo hold, then we denote it as just \LP.

    \item \textbf{(IID)}: Independent Identity~\cite{ince2017measuring}. \\
    Independent sources provide vanishing redundancy about a target that is a copy of them:
    \begin{equation}\label{ax:IID}
        I_\cap(X_1,X_2;X_1X_2)=0 \,,
    \end{equation}
    for $X_1\ind X_2$. If $X_1$ and $X_2$ are binary random variables, this corresponds to the TBC gate.

    \item \textbf{(ID)}: Identity~\cite{harder2013bivariate}. \\
    The redundancy between two sources and a target that is a copy of them is equal to the mutual information between the sources:
    \begin{equation}\label{ax:ID}
        I_\cap(X_1,X_2;X_1X_2)=I(X_1;X_2) \,.
    \end{equation}
    If $X_1$ and $X_2$ are binary random variables, this corresponds to the TBC gate.

    \item \textbf{(TM)}: Target Monotonicity~\cite{bertschinger2013shared}. \\
    Adding a target does not decrease the value of redundancy:
    \begin{equation}\label{ax:TM}
        I_\cap(\Xs;Y_1\ldots Y_n)\ge I_\cap(\Xs;Y_1\ldots Y_{n-1}) \,.
    \end{equation}

    \item \textbf{(TC)}: Target Chain rule~\cite{bertschinger2013shared}\footnote{In this work, we assume that the ordering of the targets should not affect redundancy, i.e.\ $\Ic(X_1,\dots,X_n;Y_1Y_2)=\Ic(X_1,\dots,X_n;Y_2Y_1)$. This implies that the Target Chain rule can be written indifferently as $\Ic(X_1,\dots,X_n;Y_1Y_2) = \Ic(X_1,\dots,X_n;Y_i) + \Ic(X_1,\dots,X_n;Y_{-i}|Y_i)$ for both $i=1,2$.}. \\
    The chain rule property holds for the target variables in the redundancy function:
    \begin{equation}\label{ax:TC}
        I_\cap(\Xs;Y_1Y_2) = I_\cap(\Xs;Y_1) + I_\cap(\Xs;Y_2 | Y_1) \,,
    \end{equation}
    with $I_\cap(\Xs;Y_2 | Y_1) = \sum_{y_1}p(y_1)I_\cap(\Xs;Y_2 | y_1)$.

    \item \textbf{(SE)}: Subset Equality~\cite{ince2017measuring}, also denoted as \textit{deterministic equality} by Kolchinsky~\cite{kolchinsky2022novel}. \\
    Removing a source more informative than another preserves redundancy: 
    \begin{equation} \label{ax:SE}
        \Icaps =  \Icapsm \,,
    \end{equation}
    if there exists $Z\in\{\Xsm\}$ such that $Z=f(X_n)$, i.e.\ $Z\leq_I X_n$.

    \item \textbf{(LB)}: Lower Bound~\cite{griffith2014intersection}. \\
    A variable less informative than all sources is a lower bound for redundancy:
    \begin{equation}\label{ax:LB}
        \Icaps \ge I(Q;Y) \,\,\text{ for }\,\, Q\leq_I X_i \;\;\forall i=1,\ldots,n \,.
    \end{equation}

    \item \textbf{(TE)}: Target Equality~\cite{griffith2014intersection, kolchinsky2022novel}. \\
    Adding the target to the sources does not change redundancy:
    \begin{equation}\label{ax:TE}
        \Icaps = I_\cap(\Xs,Y;Y) \,.
    \end{equation}

    \item \textbf{(M\ts{1})}: Strong Monotonicity~\cite{griffith2014intersection}. \\
    Adding a source to the redundancy function does not increase redundancy, and removing a source more informative than another source or the target preserves redundancy:
    \begin{equation}\label{ax:M1}
        \Icaps \leq \Icapsm \,,
    \end{equation}
    with equality if there exists $Z\in\{\Xsm,Y\}$ such that $Z=f(X_n)$, i.e.\ $Z\leq_I X_n$.

    \item \textbf{(S\ts{1})}: Strong Symmetry~\cite{bertschinger2013shared}. \\
    Redundancy is invariant under any permutation of sources or target:
    \begin{equation}\label{ax:S1}
        \Icaps = I_\cap(\sigma(X_1),\ldots,\sigma(X_n);\sigma(Y))\,,
    \end{equation}
    for any permutation $\sigma$.

    \item \textbf{($\boldsymbol{\ast}$)}: Assumption ($\ast$) of Bertschinger \etal~\cite{bertschinger2014quantifying}. \\
     Redundancy and unique information only depend on the marginal and pairwise distributions between sources and the target:
    \begin{equation}\label{ax:AST}
        \Icaps \text{ is a function of only } p(Y), p(X_i,Y) \;\forall \,i =1,\ldots,n \,.
    \end{equation}

    \item \textbf{(BP)}: Blackwell Property~\cite{rauh2017extractable, kolchinsky2022novel, venkatesh2022partial}. \\
    The unique information of a source vanishes if and only if the source is Blackwell inferior to all others:
    \begin{equation}\label{ax:BP}
        I_\partial(X_i;Y) =0\iff X_i\preceq_Y X_j \;\;\forall X_i\in\{X_1,\ldots,X_{i-1},X_{i+1},\ldots,X_n\}\,.
    \end{equation}

    \item \textbf{(AD)}: Additivity~\cite{rauh2023continuity}. \\
    The redundancy of two independent subsystems is equal to the sum of their redundancies: 
    \begin{equation}\label{ax:AD}
    \begin{aligned}
        I_\cap(X_1,\ldots, X_n,X_1',\ldots, X_n',Z_1,\ldots, Z_n,Z_1',\ldots, Z_n';Y_1,\ldots, Y_n,Y_1',\ldots, Y_n') = \\
        = I_\cap(X_1,\ldots, X_n,Z_1,\ldots, Z_n;Y_1,\ldots, Y_n)+I_\cap(X_1',\ldots, X_n',Z_1',\ldots, Z_n';Y_1',\ldots, Y_n') \,,
    \end{aligned}
    \end{equation}
    where $(X_1,\ldots, X_n,Z_1,\ldots, Z_n;Y_1,\ldots, Y_n)\ind (X_1',\ldots, X_n',Z_1',\ldots, Z_n';Y_1',\ldots, Y_n')$.

    \item \textbf{(CO)}: Continuity~\cite{rauh2023continuity}. \\
    Redundancy is continuous under small changes in the probability distribution:
    \begin{equation}\label{ax:CO}
        ||p(x_1,\ldots,x_n)-p'(x_1,\ldots,x_n)||_p < \delta \implies  |I^p_\cap(\XsY)-I^{p'}_\cap(\XsY)|<\varepsilon \,,
    \end{equation}
    where we indicated with $||\cdot||_p$ the standard $p-$norm.

    \item \textbf{(DI)}: Differentiability~\cite{makkeh2021introducing}. \\
    Redundancy is differentiable with respect to the probability distribution:
    \begin{equation}\label{ax:DI}
        \lim_{||p' - p||_p \to 0} 
        \frac{I^{p'}_\cap(\XsY) - I^p_\cap(\XsY)}
        {||p' - p||_p} = \nabla_p \,I^p_\cap(\XsY)^\top (p' - p) < \infty \,,
    \end{equation}
    where $\nabla_p I^p_\cap(\XsY)$ denotes the gradient of the redundancy with respect to the probability distribution $p(x_1,\ldots,x_n)$.

    \item \textbf{(EI)}: Equivalence-class Invariance~\cite{griffith2014intersection}. \\
    Redundancy is invariant under the substitution of a random variable with an informationally equivalent one. In other words, it is invariant under the equivalence relation $\sim$ (Eq.~\eqref{eq:def_equivalence}).
\end{itemize}

\section{Relations between PID properties}

In this section, we present a comprehensive catalogue of implications and incompatibilities among PID properties, together with their precise mathematical formulations and proofs. With the exception of Prop.~\ref{prop:Mz-SR--GP}, Prop.~\ref{prop:LP--M1}, Theo.~\ref{theo:SR-SE-So--ID}, and the corollaries, all results reported here are minimal, in the sense that additional relations can be derived by combining them. We nevertheless state these findings explicitly, due to their historical significance and practical importance.
    
\subsection{Basic implications}

\begin{lemma} \label{lemma:LP--GP}
    $\LP \implies \GP.$
\end{lemma}

\begin{lemma} \label{lemma:ID--IID}
    $\ID \implies \IID.$
\end{lemma}

\begin{lemma} \label{lemma:So--Sz}
    $\So \implies \Sz$.
\end{lemma}

\begin{lemma} \label{lemma:Mz--Mo}
    $\Mz \implies \SE.$
\end{lemma}

\begin{lemma} \label{lemma:DI--CO}
    $\DI \implies \CO.$
\end{lemma}

\begin{proposition}[\!\!\cite{griffith2014intersection}] \label{prop:M0-SR--LB}
    \Mz, \SR $\implies$ \LB
\end{proposition}
\begin{proof}
    If there exists $Q$ s.t. $Q\leq_I X_i\,\,\,\,\forall i =1,\ldots n$, then by applying the equality condition of \Mz recursively we get
    \begin{equation}
        I_\cap(X_1,\ldots,X_n;Y) \geq I_\cap(X_1,\ldots,X_n,Q;Y) = \ldots = I_\cap(Q;Y) = I(Q;Y) \,,
    \end{equation}
    where the last step follows from \SR.
\end{proof}

\begin{proposition}[\!\!\cite{williams2010nonnegative}] \label{prop:LB-SR--GP}
    \SR, \LB $\implies$ \GP.
\end{proposition}
\begin{proof}
    The proof directly follows from choosing $f_i(X_i)=C\,\,\,\forall \,i=1,\ldots,n$ as constant functions in the definition of \LB, which then implies
    \begin{equation}
        I_\cap(\Xs;Y) \ge I_\cap(C;Y) = I(C;Y) = 0 \,.
    \end{equation}
\end{proof}

\begin{proposition}[\!\!\cite{williams2010nonnegative}] \label{prop:Mz-SR--GP}
    \Mz, \SR $\implies$ \GP.
\end{proposition}
\begin{proof}
This follows from Props.~\ref{prop:M0-SR--LB} and \ref{prop:LB-SR--GP}.
\end{proof}

\begin{proposition}[\!\!\cite{williams2010nonnegative, gutknecht2025babel, matthias2025novel}] \label{prop:LP--M0}
    \Sz, \LP $\implies$ \Mz.
\end{proposition}
\begin{proof}
    It follows directly from the definition of the PID lattice, since:
    \begin{equation}
        I_\cap(\alpha;Y) = \sum_{\beta\preceq\alpha} I_\partial(\beta;Y) \,,
    \end{equation}
    and from \LP
    \begin{equation}
        I_\cap(\alpha;Y) \ge I_\partial(\beta;Y) \,\,\forall \alpha,\beta\,.
    \end{equation}
\end{proof}

\begin{proposition}[\!\!\cite{bertschinger2013shared}] \label{prop:TC-GP--TM}
    \GP, \TC $\implies$ \TM
\end{proposition}
\begin{proof}
    It follows immediately from
    \begin{equation}
    I_\cap(X_1,\ldots,X_n;Y_1Y_2)=I_\cap(X_1,\ldots,X_n;Y_1)+I_\cap(X_1,\ldots,X_n;Y_2|Y_1)\ge I_\cap(X_1,\ldots,X_n;Y_1) \,.
    \end{equation}
\end{proof}

\begin{proposition} \label{prop:LP--TE}
    \SR, \Sz, \LP $\implies$ \TE.
\end{proposition}
\begin{proof}
    Consider the PID with sources $(X_1,\dots,X_n,Y)$ and target $Y$. We want to show that $\Ic(\Xs,Y;Y) = \Ic(\Xs;Y)$.
    First, by the chain rule of mutual information, we have
    \begin{equation}
        I(X_1\dots X_n Y;Y) = I(Y;Y) + I(X_1\dots X_n;Y \mid Y) = I(Y;Y)\,.
    \end{equation}
    By \SR, the left-hand side equals the sum of all atoms of the lattice, while $I(Y;Y) = \Ic(Y;Y)$ equals the sum
    of all atoms $\alpha \preceq \{\{Y\}\}$, i.e.\ of all antichains containing the singleton $\{Y\}$. Hence
    \begin{equation}
        \sum_{\alpha \text{ s.t. }\{Y\}\notin\alpha}
        I_\partial(\alpha;Y) = 0\,.
    \end{equation}
    By \LP, every term in this sum is non-negative, and therefore vanishes
    \begin{equation} \label{eq:LP_vanishing_terms}
        I_\partial(\alpha;Y) = 0 \qquad
        \forall\, \alpha \text{ s.t. } \{Y\}\notin\alpha\,.
    \end{equation}
    Since the only antichains below $\{\{X_1\},\dots,\{X_n\}\}$ are the node itself and the bottom node $\{\{X_1\},\dots,\{X_n\},\{Y\}\}$, we have
    \begin{equation}
        \begin{split}
            \Ic(\Xs;Y) & = I_\partial(\Xs;Y) + I_\partial(\Xs,Y;Y) \\
            & = I_\partial(\Xs,Y;Y) \\
            & = \Ic(\Xs,Y;Y)\,,
        \end{split}
    \end{equation}
    where in the second step we used Eq.~\eqref{eq:LP_vanishing_terms}, since $\{Y\}\notin\{\{X_1\},\dots,\{X_n\}\}$.
\end{proof}

\begin{remark}
    Note that \LP implies something stronger than \TE. Since Eq.~\eqref{eq:LP_vanishing_terms} holds for every antichain not containing $\{Y\}$, for any collection of sources $\alpha$ we have
    \begin{equation}
        \Ic(\alpha;Y) = \Ic(\alpha \cup \{\{Y\}\};Y)\,.
    \end{equation}
    Indeed, the antichains below $\alpha \cup \{\{Y\}\}$ are exactly those below $\alpha$ that contain $\{Y\}$. The difference between the two redundancies is therefore a sum of atoms not containing $\{Y\}$, all of which vanish by Eq.~\eqref{eq:LP_vanishing_terms}. In other words, when the target is also a source, \SR and \LP force the redundancy lattice to collapse onto the atoms containing the singleton $\{Y\}$. \TE is recovered as the special case $\alpha = \{\{X_1\},\dots,\{X_n\}\}$.
\end{remark}

\begin{proposition} \label{prop:LP--M1}
    \SR, \Sz, \LP $\implies$ \Mo.
\end{proposition}
\begin{proof}
    This follows from Props.~\ref{prop:LP--M0} and \ref{prop:LP--TE}.
\end{proof}

\subsection{Main implications and incompatibilities}

\begin{theorem}[\!\!\cite{matthias2025novel}]\label{theo:noIID-LP1} 
    \Sz, \SR, \EI, \LPo, and \IID are incompatible. 
\end{theorem}
\begin{proof}
    See Theo. 2 of Ref.~\cite{matthias2025novel}.
\end{proof}

\begin{theorem}\footnote{This result has been obtained independently in Ref.~\cite{matthias2025novel}, see Remark~\ref{remark:noPID-TC}. We were also made aware of the statement of this theorem via private communication by Joseph Lizier. The proof presented here was developed independently.} \label{theo:noSR-Sz-Mz-TC}
    \SR, \Sz, \Mz, and \TC are incompatible.
\end{theorem}
\begin{proof}
    We prove this result by contradiction, showing a direct counterexample for which these properties cannot coexist. Consider a PID with sources $(X_1, X_2)$, with $X_1,X_2\sim \text{Bern}(1/2)$ and $X_1\perp X_2$, and target $X_1 Z$, with $Z=X_1 \oplus X_2$. This example has already been used to show other incompatibilities~\cite{finn2018pointwise, matthias2025novel}. 
    Using \TC, we can write the redundancy term in two equivalent ways as:
    \begin{align} \label{eq:TC_1decomp}
        \Ic(X_1,X_2;X_1 Z) & = \Ic(X_1,X_2;Z) + \Ic(X_1,X_2;X_1 |Z) \\
        & = \Ic(X_1,X_2;X_1) + \Ic(X_1,X_2;Z|X_1) \label{eq:TC_2decomp} \,.
    \end{align}
    Focusing on Eq.~\eqref{eq:TC_1decomp}, using \Mz we can bound the first term as
    \begin{equation}
        0 \le \Ic(X_1,X_2;Z) \le \min\{I(X_1;Z), I(X_2;Z)\} = 0 \,,
    \end{equation}
    where the lower bound is ensured by \GP (see Prop.~\ref{prop:Mz-SR--GP}). Hence, this shows that $\Ic(X_1,X_2;Z)=0$. Instead, for the second term, we notice that, given $Z$, $X_2$ is a function of the source $X_1$ (and the other way around), so that by the equality condition of \Mz (i.e.\ \SE), we can remove one source without affecting the redundancy. This implies:
    \begin{equation}
        \Ic(X_1,X_2;X_1 |Z) = \Ic(X_1;X_1 |Z) = I(X_1;X_1 |Z) = 1 \,,
    \end{equation}
    where we also used \SR. Hence, from Eq.~\eqref{eq:TC_1decomp}, we can conclude that $\Ic(X_1,X_2;X_1 Z)=1$. 
    
    Consider now Eq.~\eqref{eq:TC_2decomp}. In this case, both terms can be tightly bounded by \Mz as follows:
    \begin{align}
        & 0 \le \Ic(X_1,X_2;X_1) \le \min\{\Ic(X_1;X_1),\Ic(X_2;X_1)\} = 0 \\
        & 0 \le \Ic(X_1,X_2;Z|X_1) \le \min\{\Ic(X_1;Z|X_1),\Ic(X_2;Z|X_1)\} = 0 \,.
    \end{align}
    Thus, both terms vanish and from Eq.~\eqref{eq:TC_2decomp} we have $\Ic(X_1,X_2;X_1 Z)=0$, which contradicts the result above. 
\end{proof}

\begin{remark} \label{remark:noPID-TC}
    Theo.~\ref{theo:noSR-Sz-Mz-TC} has been independently derived in Ref.~\cite{matthias2025novel} in slightly different terms. There, Prop. 3 makes it clear that \Mz can be slightly weakened by ``No Information, No Redundancy'' (NR) and ``Full Information, Full Redundancy'' (FR) conditions, following a similar proof as the one presented above. In fact, \Sz, \SR, and \Mz imply both NR and FR.
\end{remark}

\begin{corollary}[\!\!\cite{matthias2025novel}] \label{cor:noSR-LP-TC}
    \SR, \Sz, \LP, and \TC are incompatible.
\end{corollary}
\begin{proof}
    The result directly follows from Prop.~\ref{prop:LP--M0} and Theo.~\ref{theo:noSR-Sz-Mz-TC}. See also Theo. 1 in Ref.~\cite{matthias2025novel}\footnote{Note that Theo. 1 of Ref.~\cite{matthias2025novel} also requires Re-encoding Invariance (REI), which is analogous to our \EI. We omit it here as REI is only used to swap the target order without affecting redundancy, a property which in our construction we take for granted (see also Prop. 2 of the same reference).}.
\end{proof}

\begin{corollary}[\!\!\cite{finn2018pointwise}] \label{cor:LPz-ID--noTC}
    \SR, \Sz, \EI, \LP, \IID\footnote{Although the original stated \ID in the hypotheses, their counterexample is more general and only requires \IID.}, and \TC are incompatible.
\end{corollary}
\begin{proof}
    This is a special case of Cor.~\ref{cor:noSR-LP-TC}. See also Theo. 6 of Ref.~\cite{finn2018pointwise}.
\end{proof}

\begin{theorem} \label{theo:TC-TE-SR--ID}
    \SR, \Sz, \TC, \TE $\implies$ \ID.
\end{theorem}
\begin{proof}
    \begin{equation}
        \begin{aligned}
            I_\cap(X_1,X_2;X_1X_2) & = I_\cap(X_1,X_2;X_1) + I_\cap(X_1,X_2;X_2|X_1) \\
            & = I_\cap(X_2;X_1) + I_\cap(X_1;X_2|X_1) \\
            & = I(X_2;X_1) + I(X_1;X_2|X_1) \\
            & = I(X_1;X_2) \,,
        \end{aligned}
    \end{equation}
    where in the first, second, and third steps we used \TC, \TE, and \SR, respectively.
\end{proof}

\begin{corollary}[\!\!\cite{bertschinger2013shared}] \label{cor:TC-LP--ID}
    \SR, \Sz, \LP, \TC $\implies$ \ID
\end{corollary}
\begin{proof}\footnote{We thank Aaron Gutknecht for suggesting the following proof.}
    Consider a PID with sources ($X_1, X_2$) and target $X_1X_2$. By \TC, we have
    \begin{equation} \label{eq:TC_2source_PID}
         I_\cap(X_1,X_2;X_1X_2) = I_\cap(X_1,X_2;X_1) + I_\cap(X_1,X_2;X_2|X_1) \,.
    \end{equation}
    Using Prop.~\ref{prop:LP--TE} on both terms on the RHS, we obtain 
    \begin{align}
        I_\cap(X_1,X_2;X_1) & = I(X_2;X_1) \\
        I_\cap(X_1,X_2;X_2|X_1) & = \Ic(X_1;X_2|X_1)=I(X_1;X_2|X_1)=0 \,.
    \end{align}
    Hence, substituting these identities back in Eq.~\eqref{eq:TC_2source_PID} completes the proof:
    \begin{equation}
        I_\cap(X_1,X_2;X_1X_2) = I(X_1;X_2) \,.
    \end{equation}
\end{proof}

\begin{remark}
    Note that, although Cor.~\ref{cor:TC-LP--ID} outlines a valid logical implication, its hypotheses can never be satisfied together (Cor.~\ref{cor:noSR-LP-TC}).
\end{remark}

\begin{theorem} \label{theo:M0-TE--M1}
    \Mz, \TE $\iff$ \Mo 
\end{theorem}
\begin{proof}
    $(\Rightarrow)$ The inequality condition follows from \Mz. For the equality, we only have to check the case where $Y\leq_I W$, as the rest is covered by \Mz: 
    \begin{equation}
        I_\cap(X_1,\ldots,X_n,W;Y) = I_\cap(X_1,\ldots,X_n,Y;Y) = I_\cap(X_1,\ldots,X_n;Y)
    \end{equation}
    where the first step follows from \Mz and $Y\leq_I W$, and the second from \TE. \\
    $(\Leftarrow)$ \Mz follows directly from the definition of \Mo. \TE follows from \Mo since the source $Y$ is a trivial function of the target $Y$.
\end{proof}

\begin{theorem} \label{theo:BP--TE-SE}
    For $n=2$ sources, \BP $\implies$ \SE, \TE.
\end{theorem}
\begin{proof}
    $(\BP \implies \SE)$: \\
    Consider two random variables $X_1,X_2$. If $X_1=f(X_2)$ then for any $Y$ we have $X_1\preceq_Y X_2$, and from \BP it follows $I_\partial(X_1;Y)=0$, i.e.\ $\Icap = I_\cap(X_1;Y)$, which is \SE. 
    \\ $(\BP \implies \TE)$: \\
    Similar to the reasoning above, if $X_2=Y$ then $X_1\preceq_Y X_2$, and thus from \BP we have $I_\partial(X_1;Y)=0$, i.e.\ $\Icap = I_\cap(X_1;Y)$, which is \TE. 
\end{proof}

\begin{proposition}[\!\!\cite{banerjee2015synergy}] \label{prop:ID-TM--ineq}
    $\ID, \TM \implies I_\cap(X_1,X_2;f(X_1,X_2))\leq I_\cap(X_1;X_2)$
\end{proposition}
\begin{proof}
    \begin{equation}
        I_\cap(X_1,X_2;f(X_1,X_2))\leq I_\cap(X_1,X_2;X_1,X_2) = I(X_1;X_2)\,,
    \end{equation}
    where the first inequality comes from \TM and choosing the function $f$ to concatenate the sources.
\end{proof}

\begin{theorem}[\!\!\cite{rauh2017extractable, james2018unique}] \label{theo:noGP-IID-TM-Mech}
    \GP, \IID, \TM and \textbf{mechanistic redundancy} are incompatible.
\end{theorem}
\begin{proof}
    Consider two independent sources $X_1\ind X_2$ and a redundancy measure that satisfies \GP, \IID, and \TM. In this case, we can still use Prop.~\ref{prop:ID-TM--ineq} for independent sources to obtain
    \begin{equation}
        I_\cap(X_1,X_2;f(X_1,X_2))\leq I(X_1;X_2) = 0\,.
    \end{equation}
    Because of \GP, the LHS also needs to vanish for any $f$, and therefore there is no mechanistic redundancy. 
\end{proof}

\subsection{Strong symmetry}

\begin{theorem}[\!\!\cite{bertschinger2013shared}] \label{theo:Mz-So--TM}
    \Mz, \So $\implies$ \TM 
\end{theorem}
\begin{proof}
    \begin{equation}
        I(\Xs;Y_1Y_2) = I(X_1,\ldots,X_{n-1},Y_1,Y_2;X_n)\geq I(X_1,\ldots,X_{n-1},Y_1;X_n) = I(\Xs;Y_1)\,,
    \end{equation}
    where the first and last steps follow from \So and the second from \Mz.
\end{proof}

\begin{theorem} \label{theo:S1-SE--TE}
    \SE, \So $\implies$ \TE.
\end{theorem}
\begin{proof}
    \begin{equation}
    \begin{aligned}
        I_\cap(\Xs,Y;Y) & = I_\cap(\Xsm,Y,Y;X_n) \\
        & = I_\cap(\Xsm,Y;X_n) \\
        & = I_\cap(\Xs;Y) \,,
    \end{aligned}
    \end{equation}
    where in the first and last step we used \So, and in the second \SE.
\end{proof}

\begin{theorem} \label{theo:SR-SE-So--ID}
    \SR, \SE, and \So $\implies$ \ID.
\end{theorem}
\begin{proof}
    \begin{equation}
    \begin{aligned}
        I_\cap(X_1,X_2;X_1X_2) & = I_\cap(X_1,X_1,X_2;X_2) \\
        & = I_\cap(X_1,X_2;X_2) \\
        & = I_\cap(X_1;X_2) \\
        & = I(X_1;X_2) \,,
    \end{aligned}
    \end{equation}
    where, in order, we used \So, \SE, \TE (since it follows from Theo.~\ref{theo:S1-SE--TE}), and \SR.
\end{proof}

\begin{theorem}[\!\!\cite{bertschinger2013shared}] \label{theo:noSo}
     \Mz, \SR, \LPo, and \So are incompatible.
\end{theorem}
\begin{proof}
    See Theo. 1 of Ref.~\cite{bertschinger2013shared}. This can also be derived by combining  Theo.~\ref{theo:noIID-LP1} and Theo.~\ref{theo:SR-SE-So--ID}. 
\end{proof}

\section{PID measures}

In this section, we provide more insights into the existing PID measures. We first present a more complete overview of the PID measures that have been proposed to date. Rather than following a chronological order, we now discuss the measures based on the philosophies behind their introductions and their similarities.
Then, we report the mathematical definitions of the measures in a unified notation with a brief description for each. 

\subsection{A long history of PID measures}
\label{sec:long_history}

Alongside the introduction of PID and the redundancy lattice, Williams and Beer proposed \Imin, the first redundancy function and one of the most widely known to date~\cite{williams2010nonnegative}. \Imin builds on an intuitive notion of redundancy, for which shared information is the minimum information provided by the sources for each outcome of the target. 
Closely related to this philosophy is \IMMI, where the minimum is taken across all realisations~\cite{barrett2010multivariate}.
\IMMI was the first measure specifically designed for continuous (Gaussian) systems, which also reconciled \Imin, \Ired, and \IBROJA by showing that they yield the same decomposition in the particular case of a bivariate Gaussian PID with univariate target. Due to its computational efficiency and simple interpretation, \IMMI has since become one of the most widely used redundancy measures, with applications ranging from neuroscience~\cite{luppi2020synergisticCore} to complex systems in general~\cite{mediano2025toward}. 
Interestingly, the first appearance of \IMMI was in the work by Bertschinger \etal~\cite{bertschinger2013shared} (there denoted by $I_I$).
Although it does not introduce any new PID measure, we mention this study as an influential analysis that shaped subsequent developments. In it, the authors proposed new PID axioms, including \TM and \TC (called left target monotonicity and left chain rule), and discussed the properties of the measures available at the time.
However, a major shortcoming of both \Imin and \IMMI, already noticed by Bertschinger \etal~\cite{bertschinger2013shared}, is that they tend to overestimate redundancy by quantifying the same \textit{amount} of information between the sources, as opposed to the same information~\cite{bertschinger2013shared, harder2013bivariate, griffith2014quantifying}.
With a specific empirical application to environmental time-series in mind, Goodwell and Kumar partially addressed this issue by proposing \IRR as a modified version of \IMMI, which provides zero redundant information for uncorrelated sources~\cite{goodwell2017temporal}.

Motivated by this issue, Harder \etal introduced the identity property \ID, resolving the apparent shortcoming of \Imin and \IMMI on the Two-Bit Copy system. Consequently, they introduced the PID function \Ired, estimating redundancy geometrically by defining redundancy as the projected information onto specific convex closures~\cite{harder2013bivariate}. 
Following these developments, another early measure that satisfied \ID was \IBROJA. Independently introduced by Bertschinger \etal and by Griffith and Koch~\cite{bertschinger2014quantifying, griffith2014quantifying}, \IBROJA was the first redundancy function backed by an explicit operational approach, explicitly acknowledging a specific functional dependency of redundant and unique information via Assumption \AST, and that satisfied many of the PID axioms present in the literature at the time. Within \IBROJA, unique information defines it via a minimisation of the conditional mutual information over a constrained set of probability distributions, leveraging concepts of decision theory~\cite{bertschinger2014quantifying}. However, despite its interpretational clarity, this measure has been argued to overestimate redundancy by artificially inflating the correlation between the sources~\cite{ince2017measuring, james2018uniquekey, james2018unique}, as well as for relying on a decision-theoretic setup that was later argued to be too restrictive~\cite{ince2017measuring}.
On a similar line of work, intending to define a maximum-entropy based redundancy function that satisfies Assumption \AST, James \etal investigated \IMES, a redundancy measure closely related to \IBROJA, underlining the conceptual differences between performing maximum-entropy projections and directly minimising mutual information~\cite{james2018uniquekey}. Although \IMES fails to satisfy the equality condition of \Mz, i.e.\ \SE, thus providing highly counterintuitive results, it highlights how source–target marginals can still implicitly encode higher-order interactions. In this sense, it provides an interesting discussion for understanding how synergistic information may persist even under maximum-entropy projections.
Exploiting maximum entropy models and studies on cybernetics and reconstructability analysis~\cite{zwick2004overview, krippendorff2009ross}, another PID measure that satisfies \ID is \IDEP by James \etal~\cite{james2018unique}. \IDEP is a measure of unique information that assesses the statistical dependencies via the construction of the so-called constraint lattice. In this framework, each level of the lattice introduces an additional constraint between the sources. The unique information of a source is then defined as the smallest change in mutual information between sources and target across all lattice edges involving the addition of the corresponding source-target constraint. Originally proposed only for discrete systems, this measure has then been defined also for Gaussian processes~\cite{kay2018exact}.
Although \Ired, \IBROJA, \IMES, and \IDEP all satisfy \ID, hence following the mechanistic intuition behind the TBC system~\cite{harder2013bivariate}, they are all limited to the bivariate case and cannot be easily extended to antichains with three or more sources~\cite{harder2013bivariate, bertschinger2014blackwell, james2018unique, james2018uniquekey}.

A different line of work has seen the introduction of PID measures based on algebraic considerations. Within these approaches, redundancy is defined via the mutual information between the meet of the sources and the target variable
However, this approach entails choosing a notion of ordering between random variables, which might depend on the empirical application or operational interpretation behind the measure, hence leading to the introduction of several measures.  
Griffith \etal first proposed \Iwedge, a PID measure based on the Gács-Körner common information~\cite{wolf2004zero}, where they employed the more informative partial relation $\le_I$  ~\cite{griffith2014quantifying}. Although \Iwedge was the first measure that satisfied the highly desired \TM, it was criticised for providing an overly strict bound on redundancy~\cite{griffith2015quantifying, kolchinsky2022novel}, vanishing for all distributions with full support, and being insensitive to statistical correlations not captured by a common random variable~\cite{james2018unique}.
Building on this approach, Griffith \etal then subsequently proposed \Ialpha~\cite{griffith2015quantifying}, relaxing the constraints underlying \Iwedge by employing an ordering based on Markov chains and conditional independence~\cite{griffith2015quantifying, kolchinsky2022novel}, at the cost of sacrificing \TM. 
Along this line, Kolchinsky more recently proposed \Iprec, defining redundancy using the same algebraic foundation but employing the Blackwell order~\cite{kolchinsky2022novel}. Through Blackwell's theorem~\cite{blackwell1953equivalent}, this definition admits a clear operational interpretation in decision-theoretic terms, improves the previous algebraic formulations by using the ``more informative [Blackwell] ordering relation'', and provides an explicit intuition in set-theoretic terms~\cite{kolchinsky2022novel}. Moreover, in this work, the author advocates against the inclusion-exclusion principle, introducing the union information as the complementary notion of redundancy, arguing that it is required to define synergy.
These algebraic formulations provide an elegant mathematical groundwork to define the notion of shared information, bridging concepts of information theory and set theory. Although this approach ensures desirable properties such as \IID, it also leads to the rejection of \LP, complicating the interpretation of negative PID atoms. 

Proceeding along the Blackwell ordering perspective, Venkatesh and Schamberg proposed \Idelta, a PID formulation designed for multivariate Gaussian systems~\cite{venkatesh2022partial}. Arguing that the result obtained by Barrett~\cite{barrett2015exploration} with \IMMI in bivariate Gaussian distributions with a univariate target is a trivial case of Blackwell ordering, they introduced a measure of how distant each source is from being Blackwell superior to the other in the general multivariate case. Naming such a distance \textit{deficiency}, they then employed it to define redundancy in a manner analogous to \IMMI. Although \Idelta reduces to \IMMI in the case of a univariate target, it is interesting that it satisfies \ID, unlike \IMMI. Despite this approach being more general than the one obtained with \IMMI, it is only defined for a bivariate PID. 
Finally, continuing along the Blackwell ordering paradigm, Mages \etal proposed \IRDR, a redundancy measure grounded on the set-overlap of the reachable decision regions of the sources~\cite{mages2023decomposing}. Such an overlap can be obtained by transforming the problem into a decision-making task for each realisation of the target. \IRDR yields a pointwise decomposition that nevertheless satisfies \Mz \GP and also \LP for any number of sources, hence violating \ID in the process. Moreover, in the specific scenarios where one of the sources is pointwise Blackwell superior to the other, or if the target is binary, \IRDR reduces to \Imin and \IBROJA, respectively.

A conceptually distinct line of work has focused on the definition of pointwise PID measures. 
Starting from a fundamentally different premise, Ince proposed the first fully pointwise PID decomposition \ICCS, in which redundant information is defined at the local level via the pointwise coinformation~\cite{ince2017measuring}. Specifically, redundancy is identified with the \textit{common change in surprisal} induced by the sources about the target: when the pointwise information contributions from all sources share the same sign, the pointwise coinformation captures the overlap between the individual pointwise information values, providing a local, set-theoretic notion of redundancy. In this process, \ICCS departs from several global PID properties, such as \Mz and \GP, while instead introducing and satisfying \IID and \SE. 
Shortly after, Finn and Lizier proposed another pointwise PID definition, \IPM, based on a double decomposition over specificity and ambiguity lattices~\cite{finn2018probability, finn2018pointwise}. \IPM focuses on distinguishing \textit{informative} information, which reduces uncertainty about the target, from \textit{misinformative} information, which instead increases it. This approach renounces global \Mz and \GP in favour of properties such as \TC, with \Mz, \GP, and \LP still holding locally and separately for the specificity and ambiguity lattice. 
Following a similar philosophy, Makkeh \etal introduced the measure \ISX, first in discrete systems~\cite{makkeh2021introducing}, and subsequently in continuous scenarios~\cite{schick2021partial, ehrlich2024partial}. \ISX measures redundancy by assessing the portion of probability mass associated with a particular target realisation that can be jointly excluded by all sources together. Similarly to \IPM, \ISX is based on separate informative and misinformative lattices, both singularly satisfying \Mz, \GP and \LP, though it loses these properties upon recombination into the global redundancy measure.
Although negative PID atoms are often avoided -- making \LP one of the most sought-after properties -- pointwise approaches not only argue that negative values are not problematic, but they actively embrace them. In fact, the presence of negative atoms would indicate the presence of mechanistic redundancy within \ICCS, ~\cite{ince2017measuring}, or scenarios in which the sources are providing misinformative information about the target for \IPM and \ISX~\cite{finn2018pointwise, makkeh2021introducing}.

Additional definitions of redundant information have then exploited unexplored ideas. Inspired by the hierarchical decomposition of Amari~\cite{amari2001information}, Niu and Quinn proposed \IIG, a redundancy function based on concepts of information geometry. Within \IIG, unique information and synergy are obtained by projecting the full probability distribution onto specific submanifolds characterised by vanishing interaction terms~\cite{niu2019measure}. \IIG was originally defined for a bivariate discrete PID with full-support distributions, and was later generalised to the continuous Gaussian case by Kay~\cite{kay2024partial}. 
Sigtermans proposed the redundancy measure \ICT based on the framework of Causal Tensors -- multivariate stochastic functions that map sources into the target via distinct transmission paths~\cite{sigtermans2020towards, sigtermans2020partial}. Unlike previous redundancy measures, \ICT has been advanced to specifically account for the indirect mediated associations, explicitly representing how information is transmitted through cascades of stochastic channels. 
A different approach led to the introduction of \Ido by Lyu \etal~\cite{lyu2024explicit}. This framework adopts an explicitly interventional perspective inspired by causal inference, where unique information is defined employing Pearl's \textit{do-calculus}~\cite{pearl1995causal, pearl2009causal, pearl2012calculus}. Specifically, unique information is expressed as a conditional mutual information evaluated under a do-operation.
Although \Ido is among the few measures that satisfy both \CO and \AD, it does not satisfy \Mz since it lacks \SE, leading to non-vanishing unique information even in the case when the sources are identical.

Finally, in this work, we also considered \IRAV, an unpublished PID measure that quantifies redundancy by maximising the coinformation between sources, target, and an arbitrary function of the sources, taken over all possible such functions~\cite{james2018dit}.

\subsection{Mathematical definitions}
\label{sec:measures_mathy}
Here we list the formal PID definitions considered in this work, ordered chronologically by their appearance in the literature.

\subsubsection*{Williams and Beer's Minimum Information $(\Imin)$}
This was the first PID measure proposed \cite{williams2010nonnegative}. It defines redundancy as the minimum information any individual source provides about the target, calculated pointwise for each outcome of the target. It has been criticised for not distinguishing between ``whether different random variables carry the same information or just the same amount of information'', thus sometimes overestimating redundancy~\cite{bertschinger2013shared, harder2013bivariate, griffith2015quantifying}.
    \begin{equation}
        \Imin(\Xs;Y) = \sum_y p(y) \min_{i=1,\ldots,n} I(X_i;Y=y) \,.
    \end{equation}

\subsubsection*{Harder \etalx's Projected Information $(\Ired)$}
$\Ired$ defines redundancy in terms of information geometry \cite{harder2013bivariate}. Redundant information is the part of the information that can be represented by projecting one source’s conditional distributions onto the convex closure of the other’s. It is specifically defined for the bivariate case, and it avoids some of the overestimation issues of $\Imin$. It was the first redundancy measure that satisfied \ID. 
    \begin{align}
    \pi_A(p) &= \underset{r \in A}{\arg\min}\,D_{\mathrm{KL}}(p \| r), \\[3pt]
    C_{\mathrm{cl}}(A) & = \left\{ \lambda p + (1-\lambda) q \;\middle|\; p, q \in A,\; \lambda \in [0,1] \right\} \\[3pt]
    \langle X_i \rangle_Y & = \{p(Y|x_i) \text{ s.t. } x_i\in\mathcal{X}_i\} \\[3pt]
    p_{(x_i\searrow \,X_j)}(Y) &= \pi_{C_{\mathrm{cl}}(\langle X_j \rangle_Y)} \big( p(Y|x_i) \big), \\[3pt]
    I_Y^\pi(X_1 \!\!\searrow\! X_2) &= \sum_{x_1,y} p(x_1,y) \log \frac{p_{(x_1 \searrow\,X_2)}(y)}{p(y)}, \\[3pt]
    \Ired(X_1, X_2; Y) &= \min \big\{ I_Y^\pi(X_1 \!\!\searrow\! X_2),\; I_Y^\pi(X_2 \!\!\searrow\! X_1) \big\} \,,
    \end{align}
    where $\pi_{C_{\mathrm{cl}}(\langle X_1 \rangle_Y)}(q)$ denotes the information projection of a distribution $q$ onto the convex closure of the set of conditional distributions of $Y$ given $X_1$.

\subsubsection*{Griffith \etalx's Common Information Redundancy ($\Iwedge$)}
The ``intersection information'' \Iwedge defines redundancy as the maximum mutual information obtainable from a variable $Q$ that can be obtained from any source via a deterministic function~\cite{griffith2014intersection}. It was the first proposed measure that satisfied \TC, serving as ``a lower bound on any reasonable redundancy measure''~\cite{griffith2014intersection}:
    \begin{equation}
        I_\cap^\wedge(\Xs;Y) = \max_{Q} I(Q;Y) \quad\text{such that}\; H(Q|X_i)=0\; \forall i \in 1 , \ldots , n \,.
    \end{equation}

\subsubsection*{Griffith \etalx's Relaxed Common Information Redundancy ($\Ialpha$)}
$\Ialpha$ defines redundancy as the maximal information obtainable from a variable $Q$ that does not increase the mutual information between each source and the target \cite{griffith2015quantifying}. This formalism loosens the constraints given by $\Iwedge$ and provides a tighter bound on a plausible measure of redundancy, as $\Iwedge\leq \Ialpha$.
\begin{equation}
    I_\cap^\alpha(\Xs;Y) = \max_{Q} I(Q;Y) \quad\text{such that}\; I(Q,X_i;Y)=I(X_i;Y) \; \forall i \in 1 , \ldots , n \,.
\end{equation}

\subsubsection*{BROJA/Griffith and Koch's Constrained Optimisation ($\IBROJA$)}
$\IBROJA$ defines redundancy via constrained optimisation over joint distributions~\cite{bertschinger2014quantifying, griffith2014quantifying}. The idea is to find all distributions that preserve each source-target marginal, then define redundancy as the maximum coinformation in these distributions. It was independently proposed by Bertschinger \textit{et al.}~\cite{bertschinger2014quantifying} and Griffith and Koch~\cite{griffith2014quantifying}, and admits an operational interpretation from decision theory. It is widely used in empirical applications due to computability and intuitiveness, though it is only defined for $n=2$ sources.
    \begin{align}
        I_\cup^{\text{BROJA}}(X_1,X_2;Y)& = I(X_1,X_2;Y) - \min_{Q\in\Delta_P} I_Q(X_1,X_2;Y) \\
        I_\cap^{\text{BROJA}}(X_1,X_2;Y) &= \max_{Q\in\Delta_P} I_Q(X_1;X_2;Y) \,,
    \end{align}
    where $\Delta_P=\{Q\in\Delta \;|\; P(X_i;Y)=Q(X_i;Y) ,\;i=1,2\}$, and $I_Q$ refers to the mutual information evaluated on the distribution $Q$.

\subsubsection*{Barrett's Minimum Mutual Information ($\IMMI$)}
Minimal Mutual Information (\IMMI) defines redundancy as the minimum mutual information between any individual source and the target~\cite{barrett2015exploration}. Therefore, it captures the smallest guaranteed contribution from any source, analogous to $\Imin$, but applied directly at the global level of mutual information, rather than pointwise. As such, it serves as the upper bound for any redundancy measure. Its simplicity makes it easy to implement and particularly popular in empirical applications.
\begin{equation}
    I_\cap^{\text{MMI}}(\XsY) = \min_{i=1,\ldots,n}I(X_i;Y) \,.
\end{equation}

\subsubsection*{Goodwell and Kumar's Rescaled Redundancy ($\IRR$)}
Rescaled Redundancy ($\IRR$) was introduced as a rescaled version of $\IMMI$ so that it would provide zero redundancy for independent sources \cite{goodwell2017temporal}. It achieves this by interpolating between the $\IMMI$ and a positive coinformation, weighted by the dependence between sources. $\IRR$ was designed for the study of dynamical processes where it is desired that independent sources should not carry redundancy~\cite{goodwell2017temporal}. 
    \begin{align}
        R_{\min}(X_1,X_2,Y) & = \max(0, I(X_1;X_2;Y))  \\
        R_{\text{MMI}}(X_1,X_2,Y) & = \min(I(X_1;Y),I(X_2;Y))  \\
        I_s(X_1,X_2) & = \frac{I(X_1;X_2)}{\min(H(X_1),H(X_2))} \\
        I_\cap^{\text{RR}}(X_1,X_2;Y) & = R_{\min}(X_1,X_2,Y) + I_s(X_1,X_2)\bigl(R_{\text{MMI}}(X_1,X_2,Y)-R_{\min}(X_1,X_2,Y)\bigr) \,.
    \end{align}

\subsubsection*{Ince's Common Change in Surprisal ($\ICCS$)}
Common Change in Surprisal ($\ICCS$) defines redundancy based on coherent changes in surprisal \cite{ince2017measuring}. A set of sources is redundant if they all increase or decrease the surprisal of the target in the same direction for a given outcome on a constrained maximum entropy distribution. Hence, it captures redundancy as the pointwise coinformation between all sources. The overall redundancy is then computed via the expectation value of the pointwise change in surprisal.
    \begin{align}
    \hat{P}(\Xs,Y)&=\underset{Q \in \Delta_P}{\arg\max}\;H_Q(\Xs,Y)
    \\[6pt]
    \Delta_P=\{Q \in \Delta \;|\; & Q(X_i,Y) = P(X_i,Y),\, i=1,\ldots,n, \text{ and }\,
    Q(\Xs) = P(\Xs)\} \\[6pt]
    \Delta_s h^{\mathrm{com}}(x_1,\ldots,x_n,y)
    &= \begin{cases} 
    c_{\hat{p}}(x_1,\ldots,x_n,y), & \text{if }  \operatorname{sgn}\Delta_s h(x_1) = \cdots = \operatorname{sgn}\Delta_s h(x_n) \\[2pt] &= \operatorname{sgn}\Delta_s h(x_1,\ldots,x_n) = \operatorname{sgn}c(x_1,\ldots,x_n,y), \\[6pt]
    0, & \text{otherwise}.
    \end{cases} \\[6pt]
    I_\cap^{\mathrm{CCS}}(\Xs;Y)
    &= \sum_{x_1,\ldots,x_n,y} \hat{p}(x_1,\ldots,x_n,y)\,
    \Delta_s h^{\mathrm{com}}(x_1,\ldots,x_n,y) \,,
    \end{align}
    where $c_q$ indicates the pointwise coinformation evaluated on the distribution $q$.

\subsubsection*{James \etalx's Unique Information via Dependency Constraints ($\IDEP$)}
Unique information via Dependency Constraints ($\IDEP$) defines the unique information of each source using maximum entropy projections \cite{james2018unique}. Specifically, $\IDEP$ requires the construction of a constraint lattice where each edge represents the addition of a constraint between a source and the target. Unique information is then computed as the minimal information gained when the source-target constraint is added. Since this decomposition stems from unique information, it provides a complete definition of the lattice only for $n=2$ sources.
\begin{equation}
    I_\partial(X_i;Y) = \min\Bigl(I_{X_iY:X_jY}(X_i;Y|X_j), I_{X_iX_j:X_iY:X_jY}(X_i;Y|X_j) \Bigr) \,,
\end{equation}
where $I_{A:B:C}(\cdot)$ indicates that the mutual information is calculated on the maximum entropy distribution with constraints $p(A)=q(A), \,p(B)=q(B), \,p(C)=q(C)$. A continuous version of $\IDEP$ was subsequently introduced for Gaussian systems by Kay and Ince \cite{kay2018exact}.

\subsubsection*{James \etalx's Maximum Entropy Star ($\IMES$)}
Maximum Entropy Star ($\IMES$) defines redundancy as the coinformation in the maximum entropy distribution that preserves each source’s marginal with the target \cite{james2018uniquekey}. Its definition follows $\IBROJA$'s approach, where, instead of minimising the mutual information, entropy is maximised under the constraints given by Assumption \AST~\cite{bertschinger2014quantifying}. 
\begin{align}
    Q & = \underset{Q'\in\Delta_P}{\arg\max}\,\, H_{Q'}(X_1,X_2,Y) \\
    I_\cup^{\text{MES}}(X_1,X_2;Y) & = I(X_1,X_2;Y) - I_Q(X_1,X_2;Y) \\
    \IMES(X_1,X_2;Y) & = I_Q(X_1;X_2;Y) \,,
\end{align}
where $\Delta_P=\{Q\in\Delta \;|\; P(X_i;Y)=Q(X_i;Y) ,\;i=1,2\}$.

\subsubsection*{Finn and Lizier's Specificity and Ambiguity ($\IPM$)}
Plus-Minus (\IPM) redundancy operates at the level of individual realisations of the sources and target. It separates the positive and negative contributions to local surprisal, defining a specificity and ambiguity lattice based on informative and misinformative exclusions~\cite{finn2018probability}. Specifically, redundant specificity and ambiguity are associated with the source event which induces the smallest total or misinformative exclusions, respectively~\cite{finn2018pointwise}.
Finally, the overall redundancy can then be computed via an expectation value. 
\begin{align}
    i_\cap^{\text{PM}\,+}(x_1,\ldots,x_n;y) & = \min_{i=1,\ldots,n} h(x_i) \\
    i_\cap^{\text{PM}\,-}(x_1,\ldots,x_n;y) & = \min_{i=1,\ldots,n} h(x_i|y) \\
    i^{\text{PM}}_{\cap}(x_1,\ldots,x_n;y)&= i^{\text{PM}\,+}_{\cap}(x_1,\ldots,x_n;y)- i^{\text{PM}\,-}_{\cap}(x_1,\ldots,x_n;y), \\
I^{\text{PM}}_\cap(\XsY) & = \sum_{x_1,\ldots,x_n,y} p(x_1,\ldots,x_n,y)\, i_\cap^\text{PM}(x_1,\ldots,x_n;y) \,.
\end{align}

\subsubsection*{Niu and Quinn's Information-Geometric Redundancy ($\IIG$)}
Information Geometry ($\IIG$) uses an information-geometric approach inspired by Amari's Total Correlation hierarchical decomposition~\cite{amari2001information}, providing a directed decomposition of mutual information\cite{niu2019measure}. In this scenario, synergy is defined by the distance between the joint distribution and the geodesic that passes through the points of conditional independence of the original distribution. Unique information is then defined as the distance between such a projection and the points of conditional independence. The decomposition is only defined for $n=2$ sources.
    \begin{align}
    {{\bar P}_1}\; &: = \underset{Q \in {S_1}}{\arg\min}\,\, D\left( {{P_{{X_1}{X_2}Y}}||Q} \right) = \frac{{{P_{{X_1}}}_Y{P_{{X_1}{X_2}}}}}{{{P_{{X_1}}}}} \\ 
    {{\bar P}_2}\; &: = \underset{Q \in {S_2}}{\arg\min}\,\, D\left( {{P_{{X_1}{X_2}Y}}||Q} \right) = \frac{{{P_{{X_2}}}_Y{P_{{X_1}{X_2}}}}}{{{P_{{X_2}}}}}. \\
    t^* & = \underset{t\in\mathbb{R}}{\arg\min}\,\, D(P_{X_1X_2Y}||\mathcal{N}\bar{P}_1^t\bar{P}_2^{(1-t)}) \\
    P^* & = \mathcal{N}\bar{P}_1^{t^*}\bar{P}_2^{(1-t^*)}\\
    I_\partial(X_1X_2;Y) & = D(P_{X_1X_2Y}||P^*) \\
    I_\cap^\text{IG}(X_1,X_2;Y)\, & \begin{aligned}[t] = & \, D(\bar{P}_1||P_{X_1X_2}P_Y) - D(P^*||\bar{P}_2) \\
                        = & \, D(\bar{P}_2||P_{X_1X_2}P_Y) - D(P^*||\bar{P}_1)
                        \end{aligned} \,, \label{eq:defIIG}
    \end{align}
where $S_i$ is a submanifold with specific coupling parameters set to zero, and $\mathcal{N}$ is a normalisation constant. Parametrising the full manifold as $M = \{ {P_{{X_1}{X_2}Y}} \text{ s.t. } {P_{{X_1}{X_2}Y}} = \exp ( {\theta _1}{F_1}({x_1}) + {\theta _2}{F_2}({x_2}) + {\theta _3}{F_3}(y) + {\theta _{12}}{F_{12}}({x_1},\;{x_2}) + {\theta _{13}}{F_{13}}({x_1},\;y) + {\theta _{23}}{F_{23}}({x_2},\;y) + {\theta _{123}}{F_{123}}({x_1},\;{x_2},\;y) - \psi (\theta ) ) \}$, we define $S_i=\{P\in M\text{ s.t. } \theta_{j3}=\theta_{123}=0\},\;i,j=1,2$. 

A continuous version of $\IIG$ was subsequently introduced for (bivariate) Gaussian systems by Kay \cite{kay2024partial}.

\subsubsection*{Sigtermans' Redundancy via Causal Tensors ($\ICT$)}

Causal Tensor (\ICT) defines redundancy using the framework of Causal Tensors, i.e.\ multilinear stochastic maps that transform source data into destination data, enabling one to express an indirect association in terms of the direct, constituent associations \cite{sigtermans2020partial, sigtermans2020towards}. \ICT compares the information transmitted from a source to the target across all possible transmission paths (cascades), and defines redundancy as the weakest indirect path that includes all sources and ends at the target. For two variables, this corresponds to the minimum of the two directional cascades:
\begin{align}
    I_{P_{X_1-Y}}\argct{X_1}{X_2}{Y} & = \sum_{ijk} p^{ijk} \log \frac{\sum_{l} A^l_iB_l^k}{p^k} = \sum_{x_1 x_2 y} p{(x_1, x_2, y)} \log \frac{\sum_{x_2'} p(x_2'|x_1)p(y|x_2')}{p(y)}  \\
    I_\cap^\text{CT}(X_1,X_2;Y) & = \min \Bigl(I_{P_{X_1-Y}}\argct{X_1}{X_2}{Y},I_{P_{X_2-Y}}\argct{X_2}{X_1}{Y} \Bigr) \,,
\end{align}
where the subscript $P_{X_i-Y}$ indicates that the path information is calculated on the probability distribution induced by the cascade $X_i \to X_j \to Y$.

\subsubsection*{Makkeh \etalx's Shared Exclusion ($\ISX$)}

Shared Exclusion ($\ISX$) is a pointwise PID measure that assigns redundancy to the shared exclusion of probability mass \cite{makkeh2021introducing, schick2021partial, ehrlich2024partial}. In other words, redundancy is defined as the amount of information that can be inferred when the joint realisation of all sources excludes a particular realisation of the target. \ISX decomposes the information into an informative and a misinformative part, which can then be combined to obtain the overall redundancy.
\begin{align} \label{eq:SX-eq1}
i^{\text{SX}\,+}_{\cap}(x_1,\ldots,x_n;y)&= \log_2 \frac{1}{P(x_1 \cup x_2 \cup \ldots \cup x_n)}, \\ \label{eq:SX-eq2}
i^{\text{SX}\,-}_{\cap}(x_1,\ldots,x_n;y) &= \log_2 \frac{P(y)}{P\!\left(y \cap (x_1 \cup x_2 \cup \ldots \cup x_n)\right)} \\
i^{\text{SX}}_{\cap}(x_1,\ldots,x_n;y)&= i^{\text{SX}\,+}_{\cap}(x_1,\ldots,x_n;y)- i^{\text{SX}\,-}_{\cap}(x_1,\ldots,x_n;y), \\
\ISX(\XsY) & = \sum_{x_1,\ldots,x_n,y} p(x_1,\ldots,x_n,y)\, i_\cap^\text{SX}(x_1,\ldots,x_n;y) \,. 
\end{align}
Originally defined only for discrete distributions, this approach has then been generalised to continuous settings \cite{schick2021partial, ehrlich2024partial}.

\subsubsection*{Kolchinsky's Blackwell Redundancy ($\Iprec$)}

$\Iprec$ defines redundancy as the maximum information obtainable from a variable $Q$ that is Blackwell inferior to all the sources \cite{kolchinsky2022novel}. As such, it is a more liberal definition than $I_{\wedge,\alpha}$ since $\Iwedge\leq \Ialpha \leq \Iprec\leq \IMMI$. \Iprec is based on an operational interpretation based on utility maximisation in decision theory. Importantly, in the proposed decomposition, the author rejects the inclusion-exclusion principle and defines the union information needed to complete the lattice decomposition~\cite{kolchinsky2022novel}.  
\begin{align}
    I_\cup^\prec(\Xs;Y) & = \min_{Q} I(Q;Y) \quad\text{such that}\; Q\succeq_Y X_i \; \forall i \in 1 , \ldots , n \\
    I_\cap^\prec(\Xs;Y) & = \max_{Q} I(Q;Y) \quad\text{such that}\; Q\preceq_Y X_i \; \forall i \in 1 , \ldots , n \,.
\end{align}

\subsubsection*{Venkatesh and Schamberg's Deficiency Redundancy ($\Idelta$)}

$\Idelta$ employs statistical deficiencies to generalise $\IMMI$'s definition and provide a Blackwellian PID for multivariate Gaussian systems \cite{venkatesh2022partial}. This is achieved by removing from the marginal mutual information the extent to which one source is distant from being Blackwell inferior to the other source, i.e.\ the \textit{deficiency} $\delta$. This approach is only defined for $n=2$ source variables.  
\begin{align}
    (P_{A|B}\circ P_{B|C})(a|c) & = \int db \, P(A=a|B=b)P(B=b|C=c) \\  
    \delta(Y;X_i\backslash X_j) & = \inf_{P(X_i'|X_j)} \mathbb{E}_{P_Y}\left[D(P_{X_i|Y}||P_{X_i'|X_j}\circ P_{X_j|Y})\right] \\
    I_\cap^\delta(X_1,X_2;Y) & = \min \Bigl(I(X_1;Y)-\delta(Y;X_1\backslash X_2), I(X_2;Y)-\delta(Y;X_2\backslash X_1)\Bigr)\,.
\end{align}

\subsubsection*{Mages and Rohner's Decision Reachability Redundancy ($\IRDR$)}

Reachable Decision Region ($\IRDR$) is a redundancy definition based on a direct decision-making viewpoint \cite{mages2023decomposing}. For every possible outcome of the target $y$, it defines a one-vs-rest variable $Y^y=y$, and for every source, it considers its “reachable decision region”, i.e.\ the entire set of performance trade-offs you can achieve using that source alone. 
Hence, since these regions contain all true/false-positive pairs that any decision rule based on that source can reach, they capture the pointwise uncertainty of each source about $Y=y$. Geometrically, these structures correspond to the zonotopes~\cite{eppstein1995zonohedra}. 
Using these regions, \IRDR then expresses redundancy as their overlap, i.e.\ the part of the decision region that is reachable using any of the sources on their own. 
\begin{align}
Y^y & = \begin{cases} y & \text{if } Y = y \\ 1-y & \text{otherwise} \end{cases} \\ \label{eq:RDR_eq1}
i_\cap^\text{RDR}\!\left(X_i; Y^y \right) &= \sum_{x_i} P\!\left(X_i=x_i \mid Y^y = y\right)\, \log \frac{ P\!\left(X_i=x_i \mid Y^y = y\right) }
     { \displaystyle
       \begin{array}{l}
         P(Y^y=y)\,P\!\left(X_i=x_i \mid Y^y = y\right)+ \\[0mm]
         \qquad\qquad+ \bigl(P(Y^y\ne y)\bigr)\, P\!\left(X_i=x_i \mid Y^y \neq y\right)
       \end{array} } \\[4pt]
i_\cap^\text{RDR}\!\left(X_i \cup \beta ; Y^y \right) &= i_\cap^\text{RDR}\!\left(X_i; Y^y\right) + i_\cap^\text{RDR}\!\left(\beta; Y^y\right) - i_\cap^\text{RDR}\!\left( \left\{ X_i \,\sqcup_{Y^y}\, X_j \;\middle|\; X_j \in \beta \right\}; Y^y \right) \\[1em] 
\IRDR(\alpha;Y) &= \sum_y p(y) \,i_\cap^\text{RDR}\left(\alpha; Y^y\right) \,,
\end{align}
where $\alpha,\beta$ are collection of subsets of sources, and $X_i \,\sqcup_{Y^y}\, X_j$ indicates the union under Blackwell order w.r.t.\ the target $Y^y$.

\subsubsection*{Lyu \etalx's do-calculus Redundancy ($\Ido$)}
$\Ido$ (\cite{lyu2024explicit}) is grounded on the interventional approach of the do-calculus proposed by Pearl~\cite{pearl1995causal, pearl2009causal, pearl2009causality}. Within this framework, unique information is defined as the conditional mutual information between an intervened source variable and the target. Since the decomposition is based on a definition of unique information, this approach fully determines all PID atoms only for $n=2$ sources.
\begin{align} \label{eq:DO-def}
    p(X_i'=x_i,Y=y|X_j=x_j) & = p(X_i=x_i|Y=y)p(Y=y|X_j=x_j) \\
    I_\partial(X_i;Y) & = I(X_i';Y|X_j) \\
    \Ido(X_1,X_2;Y) & = I(X_1';X_2)= I(X_2';X_1) \,. \label{eq:DO-def2}
\end{align}

\subsubsection*{Redundancy via Auxiliary Variable ($\IRAV$)}
Redundancy via Auxiliary Variable ($\IRAV$) defines shared information as the maximal coinformation achievable by the sources, the target, and any function of the sources. It captures redundancy as information that can be extracted collectively from sources, regardless of whether it is present individually. This measure is not published, but it is part of the \texttt{dit} package~\cite{james2018dit}. Although it can be easily defined for any number of source variables, we underline that \IRAV satisfies many of the PID properties only in the case $n=2$, which is hence the case we focused on.
\begin{equation}
    I^\text{RAV}_\cap(\Xs;Y) = \max_f \,I(X_1;\ldots;X_n;Y;f(\Xs)) \,.
\end{equation}

\section{Properties of PID measures}
In this section, we provide the justifications of the PID properties for each measure, as well as the proofs needed for the construction of Table~\ref{tab:PID_properties}.

\subsection{List of properties}

\subsubsection*{Williams and Beer's Minimum Information $(\Imin)$}
    \SR, \Sz, \Mz, \GP, \LPz follow directly from the definition~\cite{williams2010nonnegative}; \\
    \ID and \IID are violated, as shown by direct checks on the TBC (e.g., Table~1 of~\cite{harder2013bivariate}); \\
    \TM is violated: a counterexample is provided in the table in Sec.~5 of Ref.~\cite{bertschinger2013shared}; \\
    \TC fails as a consequence of \GP together with the violation of \TM (see Prop.~\ref{prop:TC-GP--TM}); \\
    \So is violated: see Table~1 of Ref.~\cite{griffith2014intersection}; \\
    \Mo is satisfied: see Ref.~\cite{griffith2015quantifying} (although it is not reported in Table~1 of Ref.~\cite{griffith2014intersection}); \\
    \LPo is satisfied: see Refs.~\cite{williams2010nonnegative, griffith2014intersection}; \\
    \BP is violated, whereas \TE and \AST are satisfied: see Ref.~\cite{kolchinsky2022novel}; \\
    \AD is violated, while \CO holds: see Ref.~\cite{rauh2023continuity}; \\
    \DI fails globally but holds almost everywhere, as \Imin is a minimum of differentiable functions; \\  
    \LB follows from \Mz and \SR (see Prop.~\ref{prop:M0-SR--LB}); \\
    \EI is satisfied: see Ref.~\cite{griffith2014intersection}; \\
    \SE follows from \Mz.
    
\subsubsection*{Harder \etalx's Projected Information $(\Ired)$}
    \SR, \Sz, \Mz, \GP, \LPz, \ID, and \IID come from its definition~\cite{harder2013bivariate}; \\
    \TM is violated: see Refs.~\cite{bertschinger2013shared, banerjee2015synergy}; \\
    \TC fails as a consequence of \GP together with the violation of \TM (see Prop.~\ref{prop:TC-GP--TM}); \\
    \So is violated: see Table 1 of Ref.~\cite{griffith2014intersection}; \\
    \Mo is satisfied: see Table 1 of Ref.~\cite{banerjee2015synergy} (although it is not reported in Table 1 of Ref.~\cite{griffith2014intersection}); \\
    \LPo is violated, as the $\Ired$ decomposition is only defined for $n=2$ sources; \\
    \BP, \TE, and \AST are satisfied: see Ref.~\cite{kolchinsky2022novel}; \\
    \AD and \CO are violated: see Ref.~\cite{rauh2023continuity}; \\
    \DI fails globally but holds almost everywhere, as \Ired is a minimum of differentiable functions; \\  
    \LB follows from \Mz and \SR (see Prop.~\ref{prop:M0-SR--LB}); \\
    \EI is satisfied: see Ref.~\cite{griffith2014intersection}; \\
    \SE follows from \Mz. \\
    Note that \CO would hold in the case of full-support distributions.
    
\subsubsection*{Griffith \etalx's Common Information Redundancy ($\Iwedge$)}
    \SR, \Sz, \Mz, and \GP follow from the definition, whereas \LPz and \ID are violated~\cite{griffith2014intersection}; \\
    \IID is satisfied: see Table 1 of Ref.~\cite{kolchinsky2022novel}; \\
    \TM follows from its definition~\cite{griffith2014intersection}; \\
    \TC is violated, since
    $\Iwedge(\XsY_1Y_2)=\max_Q I(Q;Y_1Y_2)=\max_Q (I(Q;Y_1)+I(Q;Y_2|Y_1))\le\max_Q I(Q;Y_1)+\max_Q I(Q;Y_2|Y_1) = \Iwedge(\XsY_1)+\Iwedge(\XsY_2|Y_1)$; \\
    \So is violated by definition~\cite{griffith2014intersection}; \\
    \Mo does not hold: see Table~1 of Refs.~\cite{griffith2014intersection, banerjee2015synergy}; \\
    \LPo fails as a consequence of the violation of \LPz; \\
    \BP and \TE are violated, whereas \AST is satisfied: see Ref.~\cite{kolchinsky2022novel}; \\
    \AD is satisfied, while \CO is violated: see Ref.~\cite{rauh2023continuity}; \\
    \DI is violated as \CO does not hold; \\
    \LB follows from \Mz and \SR (see Prop.~\ref{prop:M0-SR--LB}); \\
    \EI comes from the definition~\cite{griffith2014intersection}; \\
    \SE follows from \Mz.

\subsubsection*{Griffith \etalx's Relaxed Common Information Redundancy ($\Ialpha$)}
    \SR, \Sz, \Mz, and \GP follow from the definition, whereas \LPz is violated~\cite{griffith2015quantifying}; \\
    \ID is violated: see Table 1 of Ref.~\cite{banerjee2015synergy}; \\
    \IID is satisfied: see Table 1 of Ref.~\cite{kolchinsky2022novel}; \\
    \TM is violated: see Refs.~\cite{banerjee2015synergy, kolchinsky2022novel} (despite being claimed without proof in Ref.~\cite{griffith2015quantifying}); \\ 
    \TC fails as a consequence of \GP together with the violation of \TM (see Prop.~\ref{prop:TC-GP--TM}); \\ 
    \So is violated due to the clear asymmetry between sources and target in the maximisation; \\
    \Mo is satisfied: see Table 1 of Ref.~\cite{banerjee2015synergy}; \\
    \LPo is violated as a consequence of the violation of \LPz; \\
    \BP and \AST are violated, whereas \TE is satisfied: see Ref.~\cite{kolchinsky2022novel}; \\
    \AD is satisfied, while \CO is violated: see Ref.~\cite{rauh2023continuity}; \\
    \DI is violated as \CO does not hold; \\
    \LB follows from \Mz and \SR (see Prop.~\ref{prop:M0-SR--LB}); \\
    \EI comes from the invariance of mutual information under equivalence classes~\cite{li2011connection}; \\
    \SE follows from \Mz.
    
\subsubsection*{BROJA/Griffith and Koch's Constrained Optimisation ($\IBROJA$)}
    \SR, \Sz, \Mz, \GP, \LPz, \ID, and \IID follow directly from the definition~\cite{bertschinger2014quantifying}; \\
    \TM is violated: see Ref.~\cite{rauh2014reconsidering}; \\
    \TC fails as a consequence of \GP together with the violation of \TM (see Prop.~\ref{prop:TC-GP--TM})~\cite{griffith2015quantifying, banerjee2015synergy}; \\
    \So is violated due to the clear asymmetry between sources and target induced by the constraints defining $\Delta_P$; \\  
    \Mo is satisfied: see Table 1 of Ref.~\cite{banerjee2015synergy}; \\
    \LPo is violated, as the $\IBROJA$ decomposition is only defined for $n=2$ sources; \\
    \BP, \TE, and \AST are satisfied: see Ref.~\cite{kolchinsky2022novel}; \\
    \AD and \CO are satisfied: see Ref.~\cite{rauh2023continuity}; \\
    \DI does not hold: see Refs.~\cite{rauh2021properties, rauh2023continuity}. \\
    \LB follows from \Mz and \SR (see Prop.~\ref{prop:M0-SR--LB}); \\
    \EI comes from the invariance of mutual information under equivalence classes~\cite{li2011connection}; \\
    \SE follows from \Mz.
    
\subsubsection*{Barrett's Minimum Mutual Information ($\IMMI$)}
    \SR, \Sz, \Mz, \GP, and \LPz follow directly from the definition~\cite{barrett2015exploration}; \\
    \ID and \IID are violated, as shown by direct checks on the TBC; \\
    \TM follows directly from the definition, since
    $I(X_1,\ldots,X_n;Y_1,Y_2)\ge I(X_1,\ldots,X_n;Y_i), \,i=1,2$; \\
    \TC is violated, since
    $\IMMI(X_1,X_2;Y_1Y_2) = \min \Bigl(I(X_1; Y_1) + I(X_1; Y_2 \mid Y_1),\; I(X_2; Y_1) + I(X_2; Y_2 \mid Y_1)\Bigr)\leq\min \Bigl(I(X_1; Y_1) + I(X_2; Y_1)\Bigr) + \min \Bigl( I(X_1; Y_2 \mid Y_1) + I(X_2; Y_2 \mid Y_1)\Bigr) = \IMMI(X_1,X_2;Y_1) + \IMMI(X_1,X_2;Y_2|Y_1)$; \\
    \So is violated, as shown by direct evaluation on the TBC with $X_1\ind X_2$, for which
    $\IMMI(X_1,X_2;Y)=\min\{H(X_1),H(X_2)\}\neq \IMMI(X_1,Y;X_2)=0$; \\
    \Mo follows from \Mz and \TE (see Theo.~\ref{theo:M0-TE--M1}); \\
    \LPo is satisfied: see Prop.~\ref{prop:MMI-LPo}; \\
    \BP is violated, whereas \TE and \AST are satisfied: see Ref.~\cite{kolchinsky2022novel}
    (note that \BP holds in the Gaussian case of two sources and a univariate target~\cite{venkatesh2022partial}); \\
    \AD is violated, while \CO is satisfied: see Ref.~\cite{rauh2023continuity}; \\
    \DI fails globally but holds almost everywhere, as \IMMI is a minimum of differentiable functions; \\  
    \LB follows from \Mz and \SR (see Prop.~\ref{prop:M0-SR--LB}); \\
    \EI comes from the invariance of mutual information under equivalence classes~\cite{li2011connection}; \\
    \SE follows from \Mz.
    
\subsubsection*{Goodwell and Kumar's Rescaled Redundancy ($\IRR$)}
    \SR, \Sz, and \Mz follow directly from the definition~\cite{goodwell2017temporal}; \\
    \GP comes from Prop.~\ref{prop:RR-GP}; \\
    \LPz follows from Prop.~\ref{prop:RR-LPz}; \\
    \ID is violated, as shown by direct evaluation on the TBC; \\
    \IID follows directly from the definition~\cite{goodwell2017temporal}; \\
    \TM is violated: for instance, taking $X_1\ind X_2$ and assuming $I(X_1;X_2;Y_1Y_2)\ge0$ yields
    $\IRR(X_1,X_2;Y_1Y_2)=I(X_1;X_2;Y_1Y_2)\ngeq I(X_1;X_2;Y_1) = \IRR(X_1;X_2;Y_1)$; \\
    \TC is violated by the same reasoning as for $\IMMI$, together with \GP and the violation of \TM (see Prop.~\ref{prop:TC-GP--TM}); \\
    \So is violated, as shown by direct evaluation on a system with $X_1=Y$, for which
    $R_{\min}=0$ and $\IMMI(X_1,X_2;Y)=\IMMI(X_1,Y;X_2)$, but $I_s(X_1,X_2)\ne I_s(X_1,Y)$; \\
    \Mo follows from \Mz and \TE (see Theo.~\ref{theo:M0-TE--M1}); \\
    \LPo is violated, as the $\IRR$ decomposition is only defined for $n=2$ sources; \\
    \BP is violated: see Prop.~\ref{prop:RR-noBP}; \\
    \TE is satisfied: see Prop.~\ref{prop:RR-TE}; \\
    \AST is violated, since both $R_{\min}$ and $I_s(X_1,X_2)$ depend explicitly on $I(X_1;X_2)$, and thus on $p(x_1,x_2)$; \\
    \AD is violated, since $\IRR\propto\IMMI$ for $I(X_1;X_2;Y)\ge0$ and $\IMMI$ is not additive; \\
    \CO comes from the continuity of $\IMMI$ and the coinformation; \\
    \DI fails globally but holds almost everywhere, as \IRR contains minima and maxima of differentiable functions; \\  
    \LB follows from \Mz and \SR (see Prop.~\ref{prop:M0-SR--LB}); \\
    \EI comes from the invariance of mutual information under equivalence classes~\cite{li2011connection}; \\
    \SE follows from \Mz.
    
\subsubsection*{Ince's Common Change in Surprisal ($\ICCS$)}
    \SR and \Sz are satisfied, whereas \Mz and \GP are violated by definition~\cite{ince2017measuring}; \\
    \LPz fails as a consequence of the violation of \GP; \\
    \ID is violated and \IID is satisfied, as given by the definition~\cite{ince2017measuring}; \\
    \TM, \TC, and \So are violated as observed from empirical tests; \\
    \Mo does not hold due to the violation of \Mz; \\
    \LPo is violated because \GP does not hold; \\
    \BP, \TE, and \AST are violated: see Ref.~\cite{kolchinsky2022novel}; \\
    \AD and \CO are violated: see Ref.~\cite{rauh2023continuity}; \\
    \DI does not hold as \CO fails; \\
    \LB fails as a consequence of the violation of \GP (see Prop.~\ref{prop:LB-SR--GP}); \\
    \EI is satisfied due to invariance of local mutual information under equivalence classes~\cite{li2011connection}; \\
    \SE is satisfied by definition~\cite{ince2017measuring}.

\subsubsection*{James \etalx's Unique Information via Dependency Constraints ($\IDEP$)}
    \SR, \Sz, \Mz, \GP, \LPz, and \ID are satisfied by definition~\cite{james2018unique}; \\
    \IID follows from \ID; \\
    \TM is violated by definition~\cite{james2018unique}; \\
    \TC fails due to the validity of \GP and the violation of \TM (see Prop.~\ref{prop:TC-GP--TM}); \\
    \So is violated as observed from empirical tests; \\
    \Mo is satisfied, as it follows from \Mz and \TE (see Theo.~\ref{theo:M0-TE--M1}); \\
    \LPo is violated, as the $\IDEP$ decomposition is only defined for $n=2$ sources; \\
    \BP is violated: see Ref.~\cite{kolchinsky2022novel}; \\
    \TE is satisfied: see Prop.~\ref{prop:DEP-TE} (despite being reported otherwise in Ref.~\cite{kolchinsky2022novel}); \\
    \AST is violated: see Ref.~\cite{kolchinsky2022novel}; \\
    \AD is violated while \CO is satisfied: see Ref.~\cite{rauh2014reconsidering}; \\
    \DI fails globally but holds almost everywhere, as \IDEP is a minimum of differentiable functions; \\  
    \LB follows from \Mz and \SR (see Prop.~\ref{prop:M0-SR--LB}); \\
    \EI is satisfied due to invariance of mutual information under equivalence classes~\cite{li2011connection}; \\
    \SE follows from \Mz.

\subsubsection*{James \etalx's Maximum Entropy Star ($\IMES$)}
    \SR and \Sz follows directly from the definition of $\IMESs$ as a mutual information (see Prop.~\ref{prop:MES-redef}); \\
    \Mz fails due to the violation of \SE (see Prop.~\ref{prop:MES-noSELB}); \\
    \GP is satisfied by definition of $\IMES$ being a mutual information (see Prop.~\ref{prop:MES-redef}); \\
    \LPz is violated, as synergy can be negative as seen from empirical tests; \\
    \IID is follows from \ID; \\
    \ID is satisfied, since $\IMES(X_1,X_2;Y)=I_Q(X_1;X_2)$, and for the TBC the maximum entropy distribution $Q$ equals $P$ (see Prop.~\ref{prop:MES-redef} and Lemma~\ref{lemma:MES_TBC}); \\
    \TM is violated as observed from empirical tests; \\
    \TC fails due to the validity of \GP and the violation of \TM (see Prop.~\ref{prop:TC-GP--TM}); \\
    \So does not hold due to the clear asymmetry between sources and target given by the constraints; \\
    \Mo is violated as \Mz fails (see Theo.~\ref{theo:M0-TE--M1}); \\
    \LPo fails as $\IMES$ is only defined for $n=2$ sources; \\
    \BP is violated: see Prop.~\ref{prop:MES-noBP}; \\
    \TE is satisfied, since $\IMES(X_1,Y;Y)=I_Q(X_1;Y;Y)=I_Q(X_1;Y)=\IMES(X_1;Y)$; \\
    \AST holds by definition~\cite{james2018uniquekey}; \\
    \AD is satisfied: see Prop.~\ref{prop:MES-AD}; \\
    \CO holds: see Prop.~\ref{prop:MES-CO}; \\
    \DI is violated similarly to \IBROJA: see Ref.~\cite{rauh2021properties}; \\
    \LB is violated: see Prop.~\ref{prop:MES-noSELB}; \\
    \EI is satisfied due to invariance under equivalence classes~\cite{li2011connection}; \\
    \SE is violated: see Prop.~\ref{prop:MES-noSELB}.

\subsubsection*{Finn and Lizier's Specificity and Ambiguity ($\IPM$)}
    \SR and \Sz hold, whereas \Mz and \GP are violated by definition~\cite{finn2018pointwise}; \\
    \LPz does not hold as a consequence of the violation of \GP; \\
    \ID and \IID do not hold, as confirmed by direct checks on the TBC~\cite{finn2018pointwise}; \\
    \TM is violated and \TC is satisfied by definition~\cite{finn2018pointwise}; \\
    \So is violated due to the clear asymmetry between sources and target; \\
    \Mo does not hold due to the violation of \Mz; \\
    \LPo does not hold because \LPz fails; \\
    \BP, \TE, and \AST do not hold: see Ref.~\cite{kolchinsky2022novel}; \\
    \AD is violated, since the $\min$ function is superadditive; \\
    \CO holds, as $\IPM$ is a composition of continuous functions; \\
    \DI fails globally but holds almost everywhere, as \IPM contains a minimum of differentiable functions; \\  
    \LB does not hold, since $i^\pm(q;y)\le i^\pm(x_i;y)$ does not imply $I(Q;Y)\le I^\pm(X_i;Y)$ (see also Prop.~\ref{prop:LB-SR--GP}); \\
    \EI holds due to invariance of local mutual information under equivalence classes~\cite{li2011connection}; \\
    \SE follows from the informative and misinformative \Mz (since $i_\cap^\pm(\Xs;Y)=i^\pm(\Xsm;Y)$ when  $X_i=f(X_n)$ for some $i=1,\ldots,n-1$, and hence $i_\cap(\Xs;Y)=i_\cap^+(\Xs;Y)-i_\cap^-(\Xs;Y)=i_\cap(\Xsm,Y)$). \\
    We remark that \Mz, \GP, and \LP hold separately at the pointwise level for the informative and misinformative lattices.

\subsubsection*{Niu and Quinn's Information-Geometric Redundancy ($\IIG$)}
    \SR holds by construction;  \\
    \Sz follows from Prop.~\ref{prop:IG-Sz}; \\
    \Mz follows from Prop.~\ref{prop:IG-Mz}; \\
    \GP comes from its original definition~\cite{niu2019measure}; \\
    \LPz is satisfied since \GP and both unique information and synergy are KL divergences~\cite{niu2019measure}; \\
    \ID and \IID are violated, as confirmed by direct application of the independent TBC;  \\
    \TM holds from empirical tests; \\
    \TC is violated as \Mz is satisfied (see Theo.~\ref{theo:noSR-Sz-Mz-TC}); \\
    \So is violated since there is a clear asymmetry between sources and targets in the geometric projections; \\
    \Mo is satisfied since \Mz and \TE hold (see Theo.~\ref{theo:M0-TE--M1}); \\
    \LPo is violated since the \IIG decomposition is only valid for $n=2$ sources; \\
    \BP is violated as observed from empirical tests; \\
    \TE follows from Prop.~\ref{prop:IG-TE}; \\
    \AST fails as proved in Prop.~\ref{prop:IG-noAST}; \\
    \AD is violated and \CO is satisfied: see Ref.~\cite{rauh2023continuity};  \\
    \DI holds as \IIG is a composition of smooth functions~\cite{rauh2023continuity}; \\
    \LB since \Mz and \SR (see Prop.~\ref{prop:M0-SR--LB}); \\
    \EI comes from the invariance of probability distributions under equivalence classes~\cite{li2011connection}; \\
    \SE follows from \Mz.

\subsubsection*{Sigtermans' Redundancy via Causal Tensors ($\ICT$)}
    \SR, \Sz, \Mz, and \GP are satisfied by definition~\cite{sigtermans2020partial}; \\
    \LPz is violated as synergy can be negative~\cite{sigtermans2020partial}; \\
    \ID follows from its definition~\cite{sigtermans2020partial}; \\
    \IID is satisfied as a consequence of \ID; \\
    \TM is satisfied by definition~\cite{sigtermans2020partial}; \\
    \TC is violated: similarly to \IMMI, the chain rule holds for mutual information, but taking the minimum breaks it; \\
    \So does not hold: see Theorem 3 of Ref.~\cite{sigtermans2020partial}; \\
    \Mo is satisfied, following from \Mz and \TE (see Theo.~\ref{theo:M0-TE--M1}); \\
    \LPo is violated by definition~\cite{sigtermans2020partial}; \\
    \BP does not hold: see Prop.~\ref{prop:CT-BP}; \\
    \TE is satisfied, since paths including $Y$ as sources cannot be considered as they are directed, and hence $\ICT(\Xs,Y;Y)=\ICT(\Xs;Y)$; \\
    \AST is violated, since altering the distribution between $X_1$ and $X_2$ evidently changes $\ICT(X_1,X_2;Y)$ (see Prop.~\ref{prop:CT-noMechRed}); \\
    \AD does not hold as observed from empirical tests; \\
    \CO is satisfied, as \ICT is the minimum of mutual information terms; \\
    \DI fails globally but holds almost everywhere, as \ICT is a minimum of differentiable functions; \\  
    \LB follows from \Mz and \SR (see Prop.~\ref{prop:M0-SR--LB}); \\
    \EI is satisfied, since causal tensors can be rewritten in terms of conditional probabilities, which are invariant under equivalence classes~\cite{li2011connection}; \\
    \SE follows from \Mz. 

\subsubsection*{Makkeh \etalx's Shared Exclusion ($\ISX$)}
    \SR and \Sz are satisfied by definition~\cite{makkeh2021introducing}; \\
    \Mz, \GP, and \LPz do not hold due to the pointwise nature of the measure~\cite{makkeh2021introducing}; \\
    \ID and \IID are violated, as it can be verified on the independent TBC; \\
    \TM is not satisfied as observed from empirical tests; \\
    \TC follows from its original definition~\cite{makkeh2021introducing}; \\
    \So does not hold due to the clear asymmetry between sources and target; \\
    \Mo is violated, because \Mz does not hold; \\
    \LPo does not hold, as \LPz fails; \\
    \BP is violated as \TE fails (see Theo.~\ref{theo:BP--TE-SE}); \\
    \TE does not hold as $i_\cap^{SX-}(\Xs,Y;Y)=0$ but $i_\cap^{SX+}(\Xs,Y;Y)\ne0$; \\
    \AST is not satisfied due to the pointwise unique behaviour of the measure~\cite{makkeh2021introducing}; \\
    \AD is violated as observed from empirical tests; \\
    \CO and \DI are satisfied as \ISX is a composition of smooth functions~\cite{makkeh2021introducing}; \\
    \LB does not hold, since \GP fails (see Prop.~\ref{prop:LB-SR--GP}); \\
    \EI is satisfied due to invariance of probability distributions under equivalence classes~\cite{li2011connection}; \\
    \SE follows from the informative and misinformative \Mz (since $i_\cap^\pm(\Xs;Y)=i^\pm(\Xsm;Y)$ when  $X_i=f(X_n)$ for some $i=1,\ldots,n-1$, and hence $i_\cap(\Xs;Y)=i_\cap^+(\Xs;Y)-i_\cap^-(\Xs;Y)=i_\cap(\Xsm,Y)$). \\
    We remark that \Mz, \GP, and \LP are satisfied separately at the pointwise level for the informative and misinformative lattices.
    
\subsubsection*{Kolchinsky's Blackwell Redundancy ($\Iprec$)}
    \SR, \Sz, and \Mz hold by definition (see Theo.~1 in Ref.~\cite{kolchinsky2022novel}); \\
    \GP is satisfied, as it is a supremum of mutual information~\cite{kolchinsky2022novel}; \\
    \LPz is violated as observed from empirical tests; \\
    \ID does not hold and \IID is satisfied by definition~\cite{kolchinsky2022novel}; \\
    \TM is not satisfied: see Prop.~\ref{prop:Prec-noTM}; \\
    \TC fails as a consequence of \GP and the violation of \TM (see Prop.~\ref{prop:TC-GP--TM}); \\
    \So is violated due to Blackwell order asymmetry between sources and targets; \\
    \Mo holds following from \TE and \Mz (see Theo.~\ref{theo:M0-TE--M1}); \\
    \LPo does not hold, since \LPz is violated; \\
    \BP, \TE, and \AST are satisfied: see Ref.~\cite{kolchinsky2022novel}; \\
    \AD is satisfied while \CO does not hold (though it is locally continuous under specific assumptions~\cite{kolchinsky2022novel}): see Ref.~\cite{rauh2023continuity}; \\
    \DI is violated as \CO does not hold; \\
    \LB follows from \Mz and \SR (see Prop.~\ref{prop:M0-SR--LB}); \\
    \EI holds due to invariance under equivalence classes~\cite{li2011connection}; \\
    \SE is satisfied, as \Iprec obeys the more general ``order equality'' and follows from \Mz~\cite{kolchinsky2022novel}. \\
    
\subsubsection*{Venkatesh and Schamberg's Deficiency Redundancy ($\Idelta$)}
    \SR is satisfied by construction; \\
    \Sz holds by definition~\cite{venkatesh2022partial}; \\
    \Mz follows from Prop.~\ref{prop:Delta-Mz}; \\
    \GP holds since $0\leq \delta(Y;X_1\backslash X_2)\leq I(X_1;Y)$; \\
    \LPz is satisfied according to empirical tests~\cite{venkatesh2022partial}; \\
    \ID follows from Prop.~\ref{prop:Delta-ID}; \\
    \IID holds as a consequence of \ID; \\
    \TM is violated as observed from empirical tests; \\
    \TC fails due to \GP combined with the violation of \TM (see Prop.~\ref{prop:TC-GP--TM}); \\
    \So is violated: for univariate target $Y$, \Idelta reduces to \IMMI which does not satisfy \So; \\
    \Mo follows from \Mz and \TE (see Theo.~\ref{theo:M0-TE--M1}); \\
    \LPo does not hold, as the decomposition is defined only for $n=2$ sources; \\
    \BP is satisfied by definition~\cite{venkatesh2022partial}; \\
    \TE follows from Prop.~\ref{prop:Delta-TE}; \\
    \AST is satisfied, since $\delta$ depends only on the marginal distributions $P(X_1;Y)$ and $P(X_2;Y)$ (see also Appendix C in Ref.~\cite{venkatesh2022partial}); \\
    \AD is violated, since in the univariate case \Idelta reduces to \IMMI, which does not satisfy \AD; \\
    \CO is satisfied, as \Idelta is a composition of KL divergences and expectations; \\
    \DI fails globally but holds almost everywhere, as \Idelta is a minimum of differentiable functions; \\  
    \LB follows from \Mz and \SR (see Prop.~\ref{prop:M0-SR--LB}); \\
    \EI is satisfied due to invariance under equivalence classes~\cite{li2011connection}; \\
    \SE follows from \Mz.

\subsubsection*{Mages and Rohner's Decision Reachability Redundancy ($\IRDR$)}
    \SR, \Sz, \Mz, \GP, and \LPz are satisfied, whereas \ID does not hold by definition~\cite{mages2023decomposing}; \\
    \IID is violated based on direct checks on the independent TBC; \\
    \TM is not satisfied, since $\Imin$ fails it and \IRDR equals $\Imin$ under pointwise Blackwell ordering; \\
    \TC does not hold as a consequence of \GP combined with the violation of \TM (see Prop.~\ref{prop:TC-GP--TM}); \\
    \So is violated, as neither $\Imin$ nor $\IBROJA$ satisfy it, and \IRDR is equal to them in specific scenarios; \\
    \Mo follows from \Mz and \TE (see Theo.~\ref{theo:M0-TE--M1}); \\
    \LPo and \BP are satisfied by definition~\cite{mages2023decomposing}; \\
    \TE follows from Prop.~\ref{prop:LP--TE}. For $n=2$, this is also a consequence of \BP (see Theo.~\ref{theo:BP--TE-SE}), i.e.\ $\IRDR(X_1,Y;Y)=I(X_1;Y)$, since $I_\partial(X_1;Y)=0$; \\
    \AST is satisfied, as $\IRDR$ depends only on $p(X_1,Y)$ and $p(X_2,Y)$ (see Eq.~\eqref{eq:RDR_eq1}); \\
    \AD does not hold, since pointwise information does not factorise for independent subsystems (see denominator of Eq.~\eqref{eq:RDR_eq1}); \\
    \CO is satisfied, as $\IRDR$ is a composition of continuous functions; \\
    \DI fails globally but holds almost everywhere, as \IRDR is a combination of differentiable functions joined by Blackwell union; \\  
    \LB follows from \Mz and \SR (see Prop.~\ref{prop:M0-SR--LB}); \\
    \EI holds due to the invariance under equivalence classes~\cite{li2011connection}; \\
    \SE follows from \Mz. \\ 

\subsubsection*{Lyu \etalx's do-calculus Redundancy ($\Ido$)}
    \SR holds by construction (see Remark~\ref{rem:DO-SR}); \\
    \Sz is satisfied by definition~\cite{lyu2024explicit}; \\
    \Mz does not hold, as \SE fails: see Prop.~\ref{prop:DO-noSE}; \\
    \GP is satisfied by definition~\cite{lyu2024explicit}; \\
    \LPz is violated, as it only holds when $H(Y|X_1,X_2)=0$ (Remark 1 in Ref.~\cite{lyu2024explicit}); \\
    \ID holds: see Prop.~\ref{prop:DO-ID}; \\
    \IID is satisfied, since \ID holds; \\
    \TM is violated as observed from empirical tests; \\
    \TC fails as a consequence of \GP and violation of \TM (see Prop.~\ref{prop:TC-GP--TM}); \\
    \So is violated due to the clear asymmetry between sources and target; \\
    \Mo does not hold, since \Mz fails; \\
    \LPo does not hold, as \Ido decomposition is defined only for $n=2$ sources; \\
    \BP is violated, since $X_1=X_2$ gives non-zero $I_\partial(X_1;Y)$, or equivalently as \SE fails (see Theo.~\ref{theo:BP--TE-SE}); \\
    \TE is satisfied, since $\Ido(X_1,Y;Y)=\Ido(X_1';Y)=\Ido(X_1;Y)$ (using $H(X_1'|Y)=H(X_1|Y)$ from Lemma 1 in Ref.~\cite{lyu2024explicit}); \\
    \AST is satisfied, as $\Idos=I(X';Y)$ depends only on pairwise marginals $p(X_1,Y)$ and $p(X_2,Y)$; \\
    \AD, \CO, and \DI are satisfied by definition~\cite{lyu2024explicit}; \\
    \LB does not hold: see Prop.~\ref{prop:DO-noLB}; \\
    \EI holds due to the invariance under equivalence classes~\cite{li2011connection}; \\
    \SE is violated, since $\Ido(X_1,f(X_1);Y)\ne \Ido(X_1;Y)$ (see Prop.~\ref{prop:DO-noSE}).
    
\subsubsection*{Redundancy via Auxiliary Variable ($\IRAV$)}
    \SR is satisfied: see Prop.~\ref{prop:RAV-SR}; \\
    \Sz follows directly from the definition of coinformation; \\
    \Mz follows from Prop.~\ref{prop:RAV-Mz}; \\
    \GP holds, since $\IRAV(\Xs;Y)= \max_f I(X_1;\ldots;X_n;Y;f(\Xs))\ge I(X_1;\ldots;X_n;Y;C)=0$ by choosing $f(\Xs)=C$ to be a constant; \\
    \LPz is satisfied: see Prop.~\ref{prop:RAV-LPz}; \\
    \ID holds: see Prop.~\ref{prop:RAV-ID}; \\
    \IID is satisfied as a consequence of \ID; \\
    \TM is violated as observed from empirical checks; \\
    \TC fails as a consequence of \GP and violation of \TM (see Prop.~\ref{prop:TC-GP--TM}); \\
    \So is violated due to the clear asymmetry between sources and targets in the argument of $f$; \\
    \Mo follows from \Mz and \TE (see Theo.~\ref{theo:M0-TE--M1}); \\
    \LPo does not hold as observed from empirical tests; \\
    \BP is violated as observed from empirical checks; \\
    \TE holds (see Prop.~\ref{prop:RAV-TE}); \\
    \AST is violated, since $\IRAV(\Xs;Y)=\max_f I(X_1;\ldots;X_n;Y;f(\Xs))$ also depends on $p(\Xs)$; \\ 
    \AD does not hold, due to the presence of the maximum and to the function $f$, which might correlate independent sources; \\
    \CO is violated, as the $\arg \max_f$ can change abruptly for small perturbations; \\
    \DI does not hold as \CO fails; \\
    \LB follows from Prop.~\ref{prop:M0-SR--LB}; \\
    \EI is satisfied due to invariance of entropies under equivalence classes~\cite{li2011connection}; \\
    \SE follows from \Mz. \\
    We remark that \Mz, \Mo, \TE, and \LB do not hold for $n>2$ (empirical counterexamples exist).

\subsection{Proofs}

\begin{proposition} \label{prop:Union-LP}
    Consider a system $\mathcal{S}$ with sources $\Xs$ and target $Y$. If for such a system the union information satisfies $I_\cup(\Xs;Y)>I(\Xs;Y)$, then any PID performed on $S$ violates \LP.
\end{proposition}
\begin{proof}
    By definition of union information and PID lattice (Eq.~\eqref{eq:pid_lattice_defeq}), we have that:
    \begin{equation}
    \begin{split}
        I_\cup(\Xs;Y) & = \sum_{J\subseteq\{1,\ldots,n\}} (-1)^{|J|-1} \Ic(X_{J_1}, \ldots, X_{J_{|J|}};Y) \\
        & = \sum_{J\subseteq\{1,\ldots,n\}} (-1)^{|J|-1}\sum_{\alpha\preceq \{X_{J_1}, \ldots, X_{J_{|J|}}\}} I_\partial(\alpha;Y) \\
        & = \sum_{\substack{\alpha\preceq \{X_{J_i}\}\\i=1,\ldots,n}} I_\partial(\alpha;Y) \,,
    \end{split}
    \end{equation}
    where the last line follows from the structure of the PID lattice, as the alternating signs reduce all multiplicities of the PID atoms to 1. This can be seen explicitly by writing out the different combinations of the atoms, but can be more easily understood via the set-theoretic relationship between intersection and union information \cite{williams2010nonnegative, williams2011information}.
    Hence we obtain
    \begin{equation}
    \begin{split}
        I(\XsY) & = \sum_{\alpha\preceq{\{X_1\ldots X_n\}}} I_\partial(\alpha;Y) \\
        & =I_\cup(\XsY) + \sum_{\substack{\alpha\succ \{X_{J_i}\}\\i=1,\ldots,n}}I_\partial(\alpha;Y) \,.
    \end{split}
    \end{equation}
    Thus, if $I_\cup(\XsY) > I(\XsY)$, the last term must be negative, hence violating \LP.
\end{proof}

\begin{lemma} \label{lemma:2constr-maxent}
    The maximum entropy distribution of the system $(X_1,X_2,Y)$ with constraints $p(X_i,Y)=q(X_i,Y)\,\,i=1,2$ is
    \begin{equation}
        q(x_1,x_2,y) = p(x_1|y)p(x_2|y)p(y) = \frac{p(x_1,y)p(x_2,y)}{p(y)}
    \end{equation}
\end{lemma}
\begin{proof}
    We prove this explicitly via Lagrange multipliers. 
    Consider the Lagrangian
    \begin{equation}
        \mathcal{L} = \sum q(x_1,x_2,y) \log q(x_1,x_2,y) - \lambda \left(\sum q(x_1,x_2,y)-1\right) - \eta\left(\sum_{x_1} q(x_2,y)-p(x_2,y)\right) - \delta\left(\sum_{x_2} q(x_1,y)-p(x_1,y)\right) \,,
    \end{equation}
    differentiating we get 
    \begin{equation}
    \begin{aligned}
        \frac{\partial\mathcal{L}}{\partial q_i} & = \log q_i + 1-\lambda-\eta_{x_2 y}-\delta_{x_1 y} = 0 \\
        q_i & \propto \exp(\eta_{x_2 y}) \exp(\delta_{x_1 y} ) \equiv A(x_2,y) B(x_1, y)\,.
    \end{aligned}
    \end{equation}    
    Thus, the maximising joint distribution must have the form
    \begin{equation}
        q(x_1,x_2,y) = C\,A(x_2,y)B(x_1,y),
    \end{equation}
    for some normalisation constant $C$. We now impose the constraints to determine $A(x_2,y)$ and $B(x_1,y)$. 
    Summing over $x_1$ gives
    \begin{equation}
    \begin{aligned}
        \sum_{x_1} q(x_1,x_2,y) & = CA(x_2,y)\sum_{x_1} B(x_1,y) \\
        & q(x_2,y) = p(x_2,y),
    \end{aligned}
    \end{equation}
    and hence
    \begin{equation}
        CA(x_2,y) = \frac{p(x_2,y)}{\sum_{x_1} B(x_1,y)}. \label{eq:maxentA}
    \end{equation}
    Similarly, summing over $x_2$ gives
    \begin{equation}
    \begin{aligned}
        \sum_{x_2} q(x_1,x_2,y) & = CB(x_1,y)\sum_{x_2} A(x_2,y) \\
        & = q(x_1,y) = p(x_1,y),
    \end{aligned}
    \end{equation}
    and hence
    \begin{equation}
        CB(x_1,y) = \frac{p(x_1,y)}{\sum_{x_2} A(x_2,y)}. \label{eq:maxentB}
    \end{equation}
    Multiplying \eqref{eq:maxentA} and \eqref{eq:maxentB} we obtain
    \begin{equation}
        C^2 A(x_2,y) B(x_1,y) = \frac{p(x_1,y)p(x_2,y)}{\Big(\sum_{x_1} B(x_1,y)\Big)\Big(\sum_{x_2} A(x_2,y)\Big)}.
    \end{equation}
    But $q(x_1,x_2,y) = CA(x_2,y)B(x_1,y)$, so
    \begin{equation}
        q(x_1,x_2,y) = \frac{p(x_1,y)p(x_2,y)}{Z(y)},
    \end{equation}
    where $Z(y)$ is the normalisation factor
    \begin{equation}
        Z(y) = C\Big(\sum_{x_1} B(x_1,y)\Big)\Big(\sum_{x_2} A(x_2,y)\Big).
    \end{equation}
    Finally, summing over $x_1,x_2$ we must have $\sum_{x_1x_2} q(x_1,x_2,y) = p(y)$, hence $Z(y)=p(y)$ and therefore
    \begin{equation}
        q(x_1,x_2,y) = \frac{p(x_1,y)p(x_2,y)}{p(y)} = p(x_1|y)p(x_2|y)p(y).
    \end{equation}
\end{proof}

\subsubsection{Proofs regarding \IMMI}

\begin{proposition} \label{prop:MMI-LPo}
    $\IMMI$ satisfies \LPo.
\end{proposition}
\begin{proof}
    We propose a proof similar to that of $\Imin$ (Theo.4 in Ref.~\cite{williams2010nonnegative}). Using Theo.3 and Lemma 2 of the same reference,  we have:
    \begin{equation} \label{eq:MMI-LPo}
        \begin{aligned} 
            I_\partial(\alpha;Y) & = \IMMI(\alpha;Y) - \sum_{k=1}^{|\alpha^-|} (-1)^{k-1} \sum_{\substack{\mathcal{B}\subseteq\alpha^- \\ |\mathcal{B}|=k}} \IMMI(\bigwedge \mathcal{B};Y) \\
            & = \IMMI(\alpha;Y) - \sum_{k=1}^{|\alpha^-|} (-1)^{k-1} \sum_{\substack{\mathcal{B}\subseteq\alpha^- \\ |\mathcal{B}|=k}} \min_{B \in \bigwedge \mathcal{B}}I(B;Y) \\
            & = \IMMI(\alpha;Y) - \sum_{k=1}^{|\alpha^-|} (-1)^{k-1} \sum_{\substack{\mathcal{B}\subseteq\alpha^- \\ |\mathcal{B}|=k}} \min_{\beta\in\mathcal{B}}\min_{B \in \beta}I(B;Y) \\
            & = \IMMI(\alpha;Y) - \max_{\beta\in\alpha^-}\min_{B \in \beta}I(B;Y) \\
            & = \min_{A\in\alpha} I(A;Y) - \max_{\beta\in\alpha^-}\min_{B \in \beta}I(B;Y) \\
            & \geq 0 \,.
        \end{aligned}
    \end{equation}
    The last line follows from Theo. 2 of Ref.~\cite{williams2010nonnegative}, i.e.\ $I(B;Y)\leq I(A;Y)$ for $B\subseteq A$, and the ordering of the antichain lattice ${\beta\preceq\alpha \iff\forall A\in \alpha \;\;\exists B\in\beta \;|\; B\subseteq A}$. 
    In fact, considering any $\beta\in\alpha^-$, we have that by definition of $\alpha^-$ ${\beta\preceq \alpha}$, thus defining $A^*\equiv\arg\min_{A\in\alpha} I(A;Y)$, there exists $B^*\in\beta \;|\;B^*\subseteq A^*$. Hence, we obtain:
    \begin{equation}
        I(A^*;Y) = \min_{A\in\alpha} I(A;Y) \ge I(B^*;Y)\ge\min_{B \in \beta}I(B;Y) \,.
    \end{equation}
    Since this holds $\forall \beta\in\alpha^-$, the last line of Eq.\eqref{eq:MMI-LPo} follows.
\end{proof}

\subsubsection{Proofs regarding \IRR}

\begin{proposition} \label{prop:RR-GP}
    $\IRR$ satisfies \GP.
\end{proposition}
\begin{proof}
    Without loss of generality, we assume $I(X_1;Y)\leq I(X_2;Y)$. Then we have
    \begin{equation}
        0\leq R_{\min} = \max (0, I(X_1;X_2;Y)) = \max (0, I(X_1;Y)-I(X_1;Y|X_2))\leq\max (0, I(X_1;Y)) \,.
    \end{equation}
    Hence $R_{\text{MMI}}-R_{\min}\ge0$, and since $I_s(X_1,X_2)\ge0$, it follows
    \begin{equation}
        \IRRs = R_{\min} + I_s(X_1,X_2)\bigl(R_\text{MMI}-R_{\min}\bigr)\ge 0 \,.
    \end{equation}
\end{proof}

\begin{proposition} \label{prop:RR-LPz}
    $\IRR$ satisfies \LPz.
\end{proposition}
\begin{proof}
    Since we have \GP, we just need to check unique information and synergy terms. For the former, we have
    \begin{equation}
    \begin{aligned}
        \Ip^\text{RR}(X_i;Y)& = I(X_i;Y)-\IRRs \\
        & \ge R_\text{MMI} - \IRRs \\
        & = R_\text{MMI} - R_{\min} - I_s(X_1,X_2)\bigl(R_\text{MMI}-R_{\min}\bigr) \\
        & = (R_\text{MMI} - R_{\min})(1-I_s(X_1,X_2)) \\
        & \ge 0 \,.
    \end{aligned}
    \end{equation}
    For synergy, we have:
    \begin{equation}
        \begin{aligned}
            \Ip^\text{RR}(X_1X_2;Y) &= I(X_1,X_2;Y)-I(X_1;Y)-I(X_2;Y) + \IRRs \\
            & = -I(X_1;X_2;Y)+\IRRs \\
            & \ge -I(X_1;X_2;Y)+R_{\min} \\
            & = \begin{cases}
                -I(X_1;X_2;Y) + I(X_1;X_2;Y)  &\text{if }\, I(X_1;X_2;X_2)\ge0 \\
                -I(X_1;X_2;Y) & \text{if }\, I(X_1;X_2;X_2)\le0
            \end{cases} \\
            & \ge 0 \,,
        \end{aligned}
    \end{equation}
    where we used that $ I_s(X_1,X_2)\bigl(R_\text{MMI}-R_{\min}\bigr)\ge 0$ (see Prop.~\ref{prop:RR-GP}).
\end{proof}

\begin{proposition} \label{prop:RR-noBP}
    $\IRR$ does not satisfy \BP.
\end{proposition}
\begin{proof}
    We prove this by violating the implication $(\Longleftarrow)$ from the definition of Eq.~\eqref{ax:BP}. Assuming $I(X_1;X_2;Y)\le0$, we have that $X_1\preceq_Y X_2$ implies $\IRR(X_1,X_2;Y)=I_s(X_1,X_2)\,I(X_1;Y)\leq I(X_1;Y)$ and hence $\Ip^\text{RR}(X_1;Y)$ can be non-zero.

    Alternatively, just note that for $I(X_1;X_2;Y)\le0$ we have $\IRR\propto \IMMI$, and \IMMI does not satisfy \BP.
\end{proof}

\begin{proposition}\label{prop:RR-TE}
    $\IRR$ satisfies \TE.
\end{proposition}
\begin{proof}
    Consider the PID with $X_2=Y$, then $R_\text{min}(X_1,Y,Y)=R_\text{MMI}(X_1,Y,Y) = I(X_1;Y)$. It follows that
    \begin{equation}
        \begin{aligned}
            \IRR(X_1,Y;Y)=R_\text{min}(X_1,Y,Y)=I(X_1;Y) \,.
        \end{aligned}
    \end{equation}
\end{proof}

\begin{proposition}\label{prop:RR-noMechRed}
    $\IRR$ does not allow for mechanistic redundancy.
\end{proposition}
\begin{proof}
    We just need to prove that if $X_1\ind X_2$ then $\IRR(X_1,X_2;Y)=0$ $\forall Y$. This directly follows from the definition, since $I_s(X_1,X_2)\propto I(X_1;X_2)=0$ and $R_{\min}=0$. Thus $\IRRs=0$.  
\end{proof}

\subsubsection{Proofs regarding \IDEP}

\begin{proposition} \label{prop:DEP-TE}
    $\IDEP$ satisfies \TE.
\end{proposition}
\begin{proof}
    Consider the PID $(X_1,X_2;Y)$ with $X_2=Y$. For \TE to hold, we need $\IDEP(X_1,X_2;Y)=\IDEP(X_1;Y)=I(X_1;Y)$, i.e. $\Ip^\text{DEP}(X_1;Y)=0$. Using the definition of unique information, we have:
    \begin{equation}
    \begin{aligned}
        I_\partial^\text{DEP}(X_1;Y) & = \min\Bigl(I_{X_1Y:X_2Y}(X_1;Y|X_2), I_{X_1X_2:X_1Y:X_2Y}(X_1;Y|X_2) \Bigr) \\
        & \geq I_{X_1Y:X_2Y}(X_1;Y|X_2) \\
        & = 0 \,,
    \end{aligned}
    \end{equation}
    where in the last step we used that the $X_1\ind Y|X_2$ on the maximum entropy distribution with constraints $p(x_i,y)=q(x_i,y),\,i=1,2$ given by Lemma \ref{lemma:2constr-maxent}.
    In fact, we have:
    \begin{equation}
    \begin{aligned}
        q(x_1,y|x_2) & = \frac{q(x_1,x_2,y)}{q(x_2)} \\
        & =  \frac{p(x_1,y)\,p(x_2,y)}{p(x_2)p(y)} \\
        & =  \frac{p(x_1,y)\,\delta(x_2,y)}{p(y)} \\
        & =  p(x_1|y)\,\delta(x_2,y) \,,
    \end{aligned}
    \end{equation}
    and thus 
    \begin{equation}
    \begin{aligned}
        I_{X_1Y:X_2Y}(X_1;Y|X_2) & = \sum_{x_1,x_2,y} q(x_1,x_2,y) \log \left(\frac{q(x_1,y|x_2)}{q(x_1|x_2)q(y|x_2)} \right) \\
        & = \sum_{x_1,x_2,y} q(x_1,x_2,y) \log \left(\frac{p(x_1|y)\,}{p(x_1|y)}\delta(x_2,y)\right) \\
        & = 0 \,.
    \end{aligned}
    \end{equation}
    The result then follows from \LPz.
\end{proof}

\subsubsection{Proofs regarding \IMES}

\begin{proposition} \label{prop:MES-redef}
    $\IMES(X_1,X_2;Y)=I_Q(X_1;X_2;Y)=I_Q(X_1;X_2)$
\end{proposition}
\begin{proof}
    From the definition of $I_\cup^{\text{MES}}(X_1,X_2;Y)$ and the inclusion-exclusion principle we have
    \begin{equation}
    \begin{aligned}
        \IMES(X_1,X_2;Y) & = I(X_1;X_2;Y) + \Ip^\text{MES}(X_1X_2;Y)  \\
        & = I(X_1;Y) + I(X_2;Y) - I(X_1,X_2;Y) + I(X_1,X_2;Y) - I_Q(X_1,X_2;Y) \\
        & = I(X_1;Y) + I(X_2;Y) - I_Q(X_1,X_2;Y) \\
        & = I_Q(X_1;Y) + I_Q(X_2;Y) - I_Q(X_1,X_2;Y) \\
        & = I_Q\cxxy \\
        & = I_Q(X_1;X_2) - I_Q(X_1;X_2|Y) \\
        & = I_Q(X_1;X_2) \,,
    \end{aligned}
    \end{equation}
    where in the last step we noticed that $I_Q(X_1;X_2|Y)=0$ due to $Q$ being the constrained maximum entropy distribution (see Lemma \ref{lemma:2constr-maxent}. 
\end{proof}

\begin{proposition}
    $\Ip^\text{MES}(X_i;Y)\ge0$ 
\end{proposition}
\begin{proof}
From Prop.~\ref{prop:MES-redef} we have:
\begin{equation} \label{eq:mes_unq1}
    \begin{aligned}
    \Ip^\text{MES}(X_i;Y) & = I(X_i;Y)- \IMESs \\
    & = I(X_i;Y)- I_Q(X_1;X_2) \,,
    \end{aligned}
\end{equation}
    but we can also notice that
    \begin{equation} \label{eq:mes_unq2}
        \begin{aligned}
            I_Q(X_1,X_2) \leq I_Q(X_1;X_2,Y)=I_Q(X_1;X_2|Y) + I_Q(X_1;Y) = I_Q(X_1;Y) = I(X_1;Y)
        \end{aligned}
    \end{equation}
    where we used $I_Q(X_1;X_2|Y)=0$ and the constraints on $Q$. Therefore, Eq.~\eqref{eq:mes_unq2} ensures that the RHS in Eq.\eqref{eq:mes_unq2} is non-negative.
    Analogous derivation holds for $X_2$.
\end{proof}

\begin{lemma} \label{lemma:MES-zeroIc}
    Consider $Q$ being the maximum entropy distribution between $X_1,X_2,Y$ with the constraints given by $\Delta_P$. Then $I_Q(X_1;Y|X_2)=0$ if and only if $X_1\ind Y$ on $P$.
\end{lemma}
\begin{proof}
    $I_Q(X_1;Y|X_2)=0$ means that $q(x_1,y|x_2)=q(x_1|x_2)q(y|x_2)$. On the other hand, writing the joint conditional probability explicitly, we have:
\begin{equation}
\begin{aligned}
    q(x_1,y|x_2) & = \frac{q(x_1,x_2,y)}{q(x_2)} \\
    & = \frac{q(x_1,y)q(x_2,y)}{q(y)q(x_2)} \\
    & = q(x_1|y)q(y|x_2) \,.
\end{aligned}
\end{equation}
    Hence, for the partitioning to hold, we have the condition
    \begin{equation}
        q(x_1|x_2) = q(x_1|y) \,,
    \end{equation}
    since the RHS does not depend on $x_2$ and the LHS does not depend on $y$, it implies that $X_1\ind (X_2,Y)$ on $Q$. 
    Hence, we can write
    \begin{equation}
        \begin{aligned}
            q(x_1,x_2,y) & = q(x_1)q(x_2,y) \\
            \frac{p(x_1,y)p(x_2,y)}{p(y)} & = p(x_1)p(x_2,y) \\
            p(x_1,y) & = p(x_1)p(y)
        \end{aligned}
    \end{equation}
    which shows that $X_1\ind Y$ on $P$.
\end{proof}

\begin{proposition} \label{prop:MES-noSELB}
    $\IMES$ does not satisfy \LB and \SE.
\end{proposition}
\begin{proof}
    For both cases, we use the result of Lemma \ref{lemma:MES-zeroIc}.
    First, consider the PID with the two sources to be identical $X_1=X_2$ and not orthogonal to an otherwise generic target $Y$. Using Prop.~\ref{prop:MES-redef}, we have that
    \begin{equation}
        \begin{aligned}
            \IMES(X_1,X_2;Y) & = I_Q(X_1;X_2;Y) \\
            & = I_Q(X_1;Y) - I_Q(X_1;Y|X_2) \\
            & = I(X_1;Y) - I_Q(X_1;Y|X_2) \\
            & < I(X_1;Y) \,,
        \end{aligned}
    \end{equation}
    where the strict inequality follows from Lemma \ref{lemma:MES-zeroIc} for any $Y\notind X_1$, since $I_Q(X_1;Y|X_2)>0$. Hence $\IMES$ does not satisfy \SE.
    
    Similarly, if we consider the PID with first source $X_1$ and second source the joint state $(X_1 X_2)$ and a generic target $Y$, then the same reasoning above applies, and we obtain $\IMES(X_1,X_1;Y)< I(X_1;Y)$ for any $Y\notind X_1$. Hence $\IMES$ does not satisfy \SE.
\end{proof}

\begin{lemma} \label{lemma:MES_TBC}
    The maximum entropy distribution of the TBC given pairwise constraints is the TBC itself.
\end{lemma}
\begin{proof}
    The proof directly follows from the closed-form solution for the maximum entropy distribution given pairwise constraints (see Prop.~\ref{prop:MES-redef}):
    \begin{equation}
        q(x_1,x_2,y) = q(x_1,x_2) = \frac{p(x_1,y)p(x_2,y)}{p(y)} = p(x_1,x_2) = p(x_1,x_2,y) \,.
    \end{equation}
\end{proof}

\begin{proposition} \label{prop:MES-CO}
    $\IMES$ satisfies \CO.
\end{proposition}
\begin{proof}
    For ease of readability, here we denote with $X,Y$ the sources and $Z$ the target.
    Consider two probability distributions $P,P'$ with full support $\mathcal{X}\times\mathcal{Y}\times\mathcal{Z}$, and let 
    \begin{equation}
        Q = \underset{\tilde{Q} \in \Delta_P}{\arg\max}\,\, H_{\tilde{Q}}(X,Y,Z) \qquad Q' = \underset{\tilde{Q} \in \Delta_{P'}}{\arg\max}\,\, H_{\tilde{Q}}(X,Y,Z)
    \end{equation}
    If  $||Q-Q'||_1<\delta_Q$, by the Fannes–Audenaert continuity bound for mutual information, we have that
    \begin{equation}
    \begin{aligned}
        ||I^\text{MES}_{\cap,P}(X,Y;Z) - I^\text{MES}_{\cap,P'}(X,Y;Z)||_1 & = ||I_{Q}(X;Y) - I_{Q'}(X;Y)||_1 \\ 
        & \leq h(\delta_Q/2) + (\delta_Q/2) \log \min(|\mathcal{X}|, |\mathcal{Y}|) \\
        & \equiv g(\delta_Q) \,,
    \end{aligned}
    \end{equation}
    where $h(x)=-(x\log x+(1-x)\log (1-x))$ is the binary entropy function.
    The proof is then completed by showing that $||Q-Q'||_1\leq C||P-P'||_1$:
    \begin{equation} \label{eq:MES-QPbound}
    \begin{aligned}
        ||Q-Q'||_1 & = \left|\left| \frac{p(x,z)p(y,z)}{p(z)} - \frac{p'(x,z)p'(y,z)}{p'(z)} \right|\right|_1 \\
        & = \left|\left| \frac{p(x,z)p(y,z)p'(z)-p'(x,z)p'(y,z)p(z)}{p(z)p'(z)} \right|\right|_1 \\
        & \leq \frac1\epsilon \left|\left| p(x,z)p(y,z)p'(z)-p'(x,z)p'(y,z)p(z) \right|\right|_1 \\
        & \leq \frac1\epsilon \sum_{xyz} \left(|p'(z)|\left|p(x,z)p(y,z)- p'(x,z)p'(y,z)\right|+|p'(x,z)p'(y,z)||p'(z)-p(z)|\right) \\
        & \leq \frac1\epsilon \left(|\mathcal{Z}|\sum_{xyz} \left|p(x,z)p(y,z)- p'(x,z)p'(y,z)\right|+|\mathcal{Z}|^2\sum_z|p'(z)-p(z)| \right)\\
        & \leq \frac1\epsilon \left(|\mathcal{Z}|\sum_{xyz} \left(|p(x,z)|\left|p(y,z)-p'(y,z)\right|+|p'(y,z)||p(x,z)-p'(x,z)|\right)+|\mathcal{Z}|^2 ||P-P'||_1 \right)\\
        & \leq \frac1\epsilon \left(|\mathcal{Z}|^2 \sum_{yz}|p(y,z)-p'(y,z)|+|\mathcal{Z}|^2 \sum_{xz}|p(x,z)-p'(x,z)|+|\mathcal{Z}|^2 ||P-P'||_1 \right)\\
        & \leq \frac{3|\mathcal{Z}|^2}{\epsilon}||P-P'||_1\,,
    \end{aligned}
    \end{equation}
    where we introduced $\epsilon\equiv \min_{z\in\mathcal{Z}}(p(z), p'(z))$, used that $|ab-a'b'|\leq |a||b-b'|+|b'||a-a'|$ which follows from the triangle inequality, and used the fact that 
    \begin{equation}
        \begin{aligned}
            ||P-P'||_1 & = \sum_{xyz} |p(x,y,z)-p'(x,y,z)| \\
            & \ge \sum_{xz} \left|\sum_y(p(x,y,z)-p'(x,y,z)) \right| \\
            & = \sum_{xz} \left|p(x,z)-p'(x,z) \right| \\
            & = ||p(x,z)-p'(x,y)||_1 \,,
        \end{aligned}
    \end{equation}
    and similarly for the other marginals.
    Finally, to make the constant of Eq.\eqref{eq:MES-QPbound} uniform, we define $m\equiv\min_{z\in\mathcal Z} p(z)$ and consider a small neighbourhood of $P$ given by $\delta_P=m$. This way for $\|P-P'\|_1 < \delta_P$ we have that for every $z$
    $|p'(z)-p(z)| \le \|P-P'\|_1 < m/2$, and hence $p'(z)\ge p(z)-m/2 \ge m/2$. 
    Thus $\epsilon \ge m/2$, and consequently, the constant in the above bound can be taken as
    \begin{equation}
      C_0 \equiv \frac{3|\mathcal Z|^2}{m/2} \;=\; \frac{6|\mathcal Z|^2}{m}\,,
    \end{equation}
    so that for all $P'$ with $\|P-P'\|_1<\delta_P$ we have the uniform inequality
    \begin{equation}
      \|Q-Q'\|_1 \le C_0 \|P-P'\|_1\,.
    \end{equation}
    Finally, taking $\delta\equiv\min(\delta_Q,\delta_P)$, we have that for $||P-P'||_1<\delta$:
    \begin{equation}
        \big|I^\text{MES}_{\cap,P}(X_1,X_2;Y)-I^\text{MES}_{\cap,P'}(X_1,X_2;Y)\big|\le g(\|Q-Q'\|_1)
      \le g(C_0\|P-P'\|_1) < \varepsilon\,.
    \end{equation}
    which is the definition of continuity.
\end{proof}

\begin{proposition} \label{prop:MES-noBP}
    $\IMES$ does not satisfy \BP.
\end{proposition}
\begin{proof}
    We prove this by showing that the requirements for Blackwell sufficiency and having zero unique information are different. 
    Consider the Blackwell ordering $X_1\preceq_YX_2$, then $X_1$ is a garbling of $X_2$:
    \begin{equation} \label{eq:MES-bp_def}
        p(x_1|y) = \sum_{x_2} k(x_1|x_2)p(x_2|y) \,.
    \end{equation}
    On the other hand, having vanishing unique information $\Ip^\text{MES}(X_1;Y)$ means 
    \begin{equation}
        I(X_1;Y) = I_Q(X_1;X_2) \,,
    \end{equation}
    which implies 
    \begin{equation} \label{eq:MES-un0}
        p(x_1|y) = q(x_1|x_2) \,.
    \end{equation}
    The two conditions of Eqs.\eqref{eq:MES-bp_def}-\eqref{eq:MES-un0} are clearly different, hence \BP cannot hold. An easy way to see it is to consider a probability distribution $P$ where $X_1\ind X_2|Y$, hence $P=Q$ and Eq.\eqref{eq:MES-un0} implies that $p(x_1|\cdot)$ should be a constant not depending on the conditioning. Clearly, this is a stronger requirement than Eq.~\ref{eq:MES-bp_def}, hence Blackwell sufficiency cannot imply vanishing unique information.
\end{proof}

\begin{proposition} \label{prop:MES-AD}
    $\IMES$ satisfies \AD.
\end{proposition}
\begin{proof}
    Consider the PID with sources $S_1\equiv(X_1,X_1')$ and $S_2\equiv(X_2,X_2')$, and target $T\equiv(Y,Y')$, where $(X_1,X_2,Y_1)\ind(X_1',X_2',Y')$. It is easy to see that the solution to the maximum entropy projection with constraints $q(S_1,T)=p(S_1,T)$ and $q(S_2,T)=p(S_2,T)$ factorises as
    \begin{equation}
        \begin{aligned}
            q(s_1,s_2,t) & = \frac {p(s_1,t)p(s_2,t)}{p(t)} \\
            & = \frac{p(x_1,x_1',y,y')p(x_2,x_2',y,y')}{p(y,y')} \\
            & = \frac{p(x_1,y)p(x_2,y)}{p(y)}\frac{p(x_1',y')p(x_2',y')}{p(y')} \\
            & = q(x_1,x_2,y) q(x_1',x_2',y') \,.
        \end{aligned}
    \end{equation}
    Thus we have $(X_1,X_2,Y_1)\ind(X_1',X_2',Y')$ on $Q$, from which it follows that 
    \begin{equation}
        \begin{aligned}
            \IMES((X_1,X_1'),(X_2,X_2');YY') & = I_Q(X_1,X_1';X_2,X_2';YY') \\
            & = I_Q(X_1;X_2;Y) +I_Q(X_1';X_2';Y') \\
            & = \IMES(X_1,X_2;Y) +\IMES(X_1',X_2';Y') \,.
        \end{aligned}
    \end{equation}
\end{proof}

\subsubsection{Proofs regarding \IIG}

\begin{proposition} \label{prop:IG-Sz}
    $\IIG$ satisfies \Sz.
\end{proposition}
\begin{proof}
    This follows easily from the definition of unique information:
    \begin{equation}
    \begin{aligned}
        \IIGs & = D(\bar{P}_1||P_{X_1 X_2}P_Y)-D(P^*||\bar{P}_2) \\
        & = D(\bar{P}_2||P_{X_1 X_2}P_Y)-D(P^*||\bar{P}_1) \\
        & = \IIG(X_2,X_1;Y)
    \end{aligned}
    \end{equation}
\end{proof}

\begin{proposition} \label{prop:IG-Mz}
    $\IIG$ satisfies \Mz.
\end{proposition}
\begin{proof}
    The inequality condition follows from \SR and the non-negativity of unique information, as it is a KL divergence:
    \begin{equation}
        \IIGs \leq \IIG(X_i;Y)=I(X_i;Y) \leq \IIG(X_1 X_2;Y)=I(X_1 X_2;Y) \quad i=1,2\,,
    \end{equation}
    so it only remains to check the equality condition, which, in a bivariate case, reduces to 
    \begin{equation}
        X_1 = f(X_2) \implies \IIG(X_1,X_2;Y)=\IIG(X_1;Y)=I(X_1;Y) \,.    
    \end{equation}
    We have that
    \begin{align}
        \bar{P}_1 & = \frac{P(X_1,Y)P(X_2)\delta(f(X_2),X_1)}{P(X_1)} \\
        \bar{P}_2 & = \frac{P(X_2,Y)P(X_2)\delta(f(X_2),X_1)}{P(X_2)} = P(X_2,Y)\delta(f(X_2),X_1) = P(X_1,X_2,Y) \,.
    \end{align}
    Thus the $P^*$ that minimises $D(P||P^*)$ is exactly $\bar{P}_2$. This implies that:
    \begin{align}
        \Ip^\text{IG}(X_1\{X_1 X_2\};Y)& = D(P||P^*) = D(P||P)=0 \\
        \Ip^\text{IG}(X_1;Y) & = D(P^*||\bar{P}_2) = D(P||P) = 0 \\
        \Ip^\text{IG}(\{X_1X_2\};Y) & = D(P^*||\bar{P}_1) = D(P||\bar{P}_1) = I(X_1X_2;Y|X_1) \,.
    \end{align}
    Hence $\IIG(X_1,X_1X_2;Y)= I(X_1;Y)$, concluding the proof.
\end{proof}

\begin{proposition} \label{prop:IG-TE}
    $\IIG$ satisfies \TE.
\end{proposition}
\begin{proof}
    Considering the PID $(X_1,X_2;Y)$ with $X_2=Y$, we have
    \begin{align}
        & P(X_1,X_2,Y) = P(X_1,Y)\delta(X_2,Y) \\[0.3em]
        & \bar{P}_1 = \frac{P(X_1,Y)P(X_1,X_2)\delta(X_2,Y)}{P(X_1)} \\
        & \bar{P}_2 = \frac{P(X_2,Y)P(X_2,X_1)\delta(X_2,Y)}{P(Y)} = P(X_1,Y)\delta(X_2,Y) \,.
    \end{align}
    Thus the $P^*$ that minimises $D(P||P^*)$ is exactly $\bar{P}_2$. Similarly to Prop.~\ref{prop:IG-Mz}, this quickly implies that:
    \begin{equation}
        \Ip^\text{IG}(X_1;Y) = D(P^*||\bar{P}_2) = D(P||P) = 0 \,,
    \end{equation}
    and hence $\IIG(X_1,X_2;Y)=I(X_1;Y)=\IIG(X_1;Y)$.
    
\end{proof}

\begin{proposition} \label{prop:IG-noAST}
    \IIG does not satisfy \AST.
\end{proposition}
\begin{proof}
    Recall the definition of redundancy (Eq.~\eqref{eq:defIIG}):
    \begin{equation}
        \IIGs = D(\bar{P}_1||P_{X_1X_2}P_Y)-D(P^*||\bar{P}_2) \,.
    \end{equation}
    The first term is the mutual information $I(X_1; Y)$, and hence only depends on the marginal distribution $P(X_1, Y)$. On the other hand, writing the second term explicitly, we obtain:
    \begin{equation}
        \begin{split}
            D(P^*||\bar{P}_2) & = \sum_{x_1,x_2,y} \mathcal{N}\bar{p}_1^t\bar{p}_2^{1-t} \log \frac{\mathcal{N}\bar{p}_1^t\bar{p}_2^{1-t}p(x_1)}{p(x_1,x_2)p(x_1,y)} = \\
            & = \sum_{x_1,x_2,y} \mathcal{N}p(x_1,x_2)p(y|x_1)^t p(y|x_2)^{1-t} \log \frac{\mathcal{N}p(x_1,x_2)p(y|x_1)^t p(y|x_2)^{1-t}p(x_1)}{p(x_1,x_2)p(x_1,y)} \\
            & = \sum_{x_1,x_2,y} \mathcal{N}p(x_1,x_2)p(y|x_1)^t p(y|x_2)^{1-t} \log \left(\mathcal{N}p(y|x_1)^{t-1} p(y|x_2)^{1-t}\right) \,,
        \end{split}
    \end{equation}
    with $\mathcal{N} = (p(x_1,x_2)p(y|x_1)^t p(y|x_2)^{1-t})^{-1}$.
    Hence, although there is no explicit term including $p(x_1,x_2)$ inside the logarithm, the redundancy evidently still depends on such a term in the summation and inside the normalisation constant $\mathcal{N}$. Thus \IIG does not satisfy \AST. 
\end{proof}

\subsubsection{Proofs regarding \ICT}

\begin{proposition} \label{prop:CT-noMechRed}
    $\ICT$ does not allow for mechanistic redundancy.
\end{proposition}
\begin{proof}
    We prove that if $X_1\ind X_2$ then $\ICT(X_1,X_2;Y)=0$ $\forall Y$: 
    \begin{equation}
    \begin{aligned}
        \ICT(X_1,X_2;Y) & = \sum_{ijk} p^{ijk} \log \left(\frac{\sum_l A_i^lB^k_l}{p^k} \right) \\
        & = \sum_{ijk}p^{ijk}\log\left(\frac{\sum_l p^lB^k_l}{p^k} \right) \\
        & = \sum_{ijk}p^{ijk}\log\left(\frac{p^k}{p^k} \right) = 0 \,,
    \end{aligned}
    \end{equation}
    where in the first step we assumed without loss of generality that the maximum inside the redundancy is obtained by $I_{P_{X_1-Y}}\argct{X_1}{X_2}{Y}$, and in the second step we made use of the fact that $X_1\ind X_2$ and hence $A^j_i=p^j$.
\end{proof}

\begin{proposition} \label{prop:CT-BP}
    $\ICT$ does not satisfy \BP.
\end{proposition}
\begin{proof}
    We prove the statement by showing that \BP does not necessarily imply vanishing unique information. Consider the Blackwell ordering $X_2\preceq_Y X_1$, this implies that $\ICTs\leq I(X_2;Y)$, in fact to have equality we would need 
    \begin{equation} \label{eq:CT_BP_casc}
        \sum_m A_j^{\dagger m} C^k_m = B^k_j \,,
    \end{equation}
    but the LHS can be written as 
    \begin{equation} \label{eq:CT_BP_requir1}
        \sum_m A_j^{\dagger m} C^k_m = \sum_m p(m|x_2)p(y|m) = \sum_m \frac{p(x_2|m)p(m,y)}{p(x_2)} 
    \end{equation}
    where we used the following identity on the Blackwell order:
    \begin{equation} \label{eq:BP_rewriting}
    \begin{aligned}
        p(x_2|y) & = \sum_x k(x_2|x_1) p(x_1|y) \\
        p(x_2,y) & = \sum_x k(x_2|x_1) p(x_1,y) \\
        p(y|x_2) & = \sum_x \frac{k(x_2|x_1) p(x_1,y)}{p(x_2)} \,.
    \end{aligned}
    \end{equation}
    On the other hand, the RHS is 
    \begin{equation} \label{eq:CT_BP_requir2}
        B^k_j = p(y|x_2) = \sum_m \frac{k(x_2|m)p(m,y)}{p(x_2)} \,.
    \end{equation}
    However, Eqs.\eqref{eq:CT_BP_requir1}-\eqref{eq:CT_BP_requir2} need not be the same, as the stochastic channel $k$ can differ from the conditional probability distribution $p$.
    Since redundancy is smaller than the marginal mutual information, unique information cannot be zero.
\end{proof}

Although Prop.~\ref{prop:CT-BP} denies \BP for $\ICT$, one side of the implication still holds. 
\begin{remark} \label{rem:CT-BP-undirected}
    Vanishing unique information in $\ICT$ implies \BP.
\end{remark}
\begin{proof}
    We prove this by showing that if $\ICTs=I(X_2;Y)$, and hence $\Ip^\text{CT}(X_2;Y)=0$, then $X_1\preceq_Y X_2$. The general proof for an arbitrary number of sources is a direct extension of the following.
    First, we notice that the Blackwell ordering imposes that the cascade $X_2\rightarrow X_1\rightarrow Y$ is equal to $X_2\rightarrow Y$:
    \begin{equation} \label{eq:CT_BP_casc}
        \sum_m A_j^{\dagger m} C^k_m = \sum_m p(m|x_2)p(y|m) = \sum_m \frac{p(x_2|m)p(m,y)}{p(x_2)} = p(y|x_2) = B^k_j \,,
    \end{equation}
    where in the last step we used Eq.\eqref{eq:BP_rewriting} and took the stochastic channel $k$ to be the conditional probability $p$.
    With this at hand, we can write:
    \begin{equation} \label{eq:CT_BP_red}
    \begin{aligned}
        \ICTs & = \min\left(I_{P_{X_1-Y}}\argct{X_1}{X_2}{Y}, I_{P_{X_2-Y}}\argct{X_2}{X_1}{Y}\right) \\
        & = I_{P_{X_2-Y}}\argct{X_2}{X_1}{Y} \\
        & = \sum_{ijk} p^{ijk} \log \left(\frac{\sum_m A_j^{\dagger m} C^k_m}{p^k} \right) \\
        & = \sum_{ijk} p^{ijk} \log \left(\frac{B_j^k}{p^k} \right) \\
        & = I(X_1;Y) \,,
    \end{aligned}
    \end{equation}
    where in the first step we used that the cascade $X_2\rightarrow Y$ is directed ($\sum_lA_i^lB_l^k\ne C^k_i$), so it is not a valid path for redundancy.  
\end{proof}

\subsubsection{Proofs regarding \Iprec}

\begin{proposition} \label{prop:Prec-noTM}
    $\Iprec$ does not satisfy \TM\footnote{We thank Artemy Kolchinsky for the suggestion of the reported counterexample.}.
\end{proposition}
\begin{proof}
    The proof follows by a direct counterexample. Consider a TBC system with full support, i.e.\ $Y=(X_1,X_2)$ and $p(x_1,x_2)>0$ $\forall x_1, x_2$. In Ref.~\cite{kolchinsky2022novel} it was shown that $\Iprec(X_1,X_2;X_1,X_2)$ is equal to the G\'acs-K\"orner common information, which vanishes for distributions with full support. Hence $\Iprec(X_1,X_2;X_1,X_2)=0$. 
    
    On the other hand, by \BP we have that $\Iprec(X_1,X_2;X_2)=I(X_1;X_2)>0$, where the strict positivity comes from the full-support requirement. Thus $0=\Iprec(X_1,X_2;X_1,X_2)<\Iprec(X_1,X_2;X_2)=I(X_1;X_2)$ and \TM does not hold.
\end{proof}

\subsubsection{Proofs regarding \Idelta}

\begin{proposition} \label{prop:Delta-Mz}
    $\Idelta$ satisfies \Mz.
\end{proposition}
\begin{proof}
    The inequality condition follows from \LPz. For the equality condition, it is enough to note that $X_2=f(X_1)$ is a stronger condition than the Blackwell ordering $X_2\preceq_Y X_1$, and hence from \BP we have $\Ip^\delta(X_2;Y)=0$, and thus $\Ideltas = I(X_2;Y)=\Idelta(X_2;Y)$.

    \textit{Alternative proof for equality condition.}
    For the equality condition, we need to prove that $\Idelta(X_1, f(X_1); Y)=\Idelta(f(X_1); Y)$.
    Consider $X_2=f(X_1)$, then we have
    \begin{equation}
        \begin{aligned}
            \Idelta(X_1, f(X_1); Y) & = \min \Bigl(I(X_1;Y)-\delta(Y;X_1\backslash f(X_1)), I(f(X_1);Y)-\delta(Y;f(X_1)\backslash X_1)\Bigr) \\
            & = \min \Bigl(I(X_1;Y)-\delta(Y;X_1\backslash f(X_1)), I(f(X_1);Y)\Bigr) \\
            & = \min \Bigl(I(X_1;Y)-I(X_1;Y|f(X_1)), I(f(X_1);Y)\Bigr) \\
            & = I(f(X_1);Y) \\
            & = I(X_2;Y) \,,
        \end{aligned}
    \end{equation}
    where we used that:
    \begin{equation}
        \begin{aligned}
            \delta(Y;X_1\backslash f(X_1)) & =\inf_k \mathbb{E}_{P_Y}\left[D(P_{X_1|Y}||K_{X_1|f(X_1)}\circ P_{f(X_1)|Y})\right] \\
            & = \inf_k \mathbb{E}_{P_Y}\left[\int dx_1 \,p(x_1|y)\log\frac{p(x_1|y)}{\int d{f(x_1)'} \,k(x_1|{f(x_1)'})p({f(x_1)'}|y)}\right] \\
            & = \mathbb{E}_{P_Y}\left[\int dx_1 \,p(x_1|y)\log\frac{p(x_1|y)}{\int df(x_1) \,\frac{p(x_1)p(f(x_1)'|x_1)}{p(f(x_1)')}p(f(x_1)'|y)}\right] \\
            & = \mathbb{E}_{P_Y}\left[\int dx_1 \,p(x_1|y)\log\frac{p(x_1|y)\,p(f(x_1))}{p(x_1)p(f(x_1)|y)}\right] \\
            & = \mathbb{E}_{P_Y}\left[\int dx_1 \,p(x_1|y)\log\frac{p(x_1|f(x_1),y)}{p(x_1|f(x_1))}\right] \\
            & = I(X_1;Y|f(X_1)) \,,
        \end{aligned}
    \end{equation}
    where we used that the infimum is attained by the true channel $p(x_1|f(x_1))$, and that $p(x_1,f(x_1))=p(x_1)$.
\end{proof}

\begin{proposition} \label{prop:Delta-TE}
    $\Idelta$ satisfies \TE.
\end{proposition}
\begin{proof}
    The proof follows trivially from the Blackwell ordering relation. Considering any two stochastic variables $(X_1,Y)$, we have that $X_1\preceq_Y Y$:
    \begin{equation}
        \int dy' k(x|y')p(y'|y) = \int dy' k(x|y')\delta(y',y) = k(x|y) = p(x|y) \,.
    \end{equation}
    Hence it follows that $\delta(Y;X_1\backslash Y)=0$ and 
    \begin{equation}
        \Idelta(X_1,Y;Y) = I(X_1;Y) = \Idelta(X_1;Y) \,.
    \end{equation}
\end{proof}

\begin{proposition} \label{prop:Delta-ID}
    $\Idelta$ satisfies \ID.
\end{proposition}
\begin{proof}
    Consider the TBC system where the target is $Y=(X_1,X_2)$. From the definition, we have that the deficiency of $X_1$ w.r.t.\ $X_2$ about $Y$ is
    \begin{equation}
    \begin{aligned}
        \delta(Y;X_1\backslash X_2) & = \inf_k \mathbb{E}_{P_Y} \left[ D(p(X_1|y) ||k(X_1|X_2)\circ p(X_2|y)\right] \\
        & = \inf_k \mathbb{E}_{P_Y} \left[ \int dx_1'\, p(x_1'|y) \log \left(\frac{p(x_1'|y)}{\int k(x_1'|x_2')p(x_2'|y) dx_2'} \right)\right] \\
        & = \inf_k \int dx_1\, dx_2\, p(x_1,x_2) \left[ \int dx_1'\, \delta(x_1'-x_1) \log \left(\frac{\delta(x_1'-x_1)}{\int k(x_1|x_2')\delta(x_2'-x_2) dx_2'} \right)\right] \\
        & = \inf_k \int dx_1\, dx_2\, p(x_1,x_2) \left[ - \log k(x_1|x_2) \right] \\
        & = \inf_k \int dx_2\, p(x_2) \int dx_1\, p(x_1|x_2) \left[ - \log k(x_1|x_2) \right] \\
        & = \int dx_2\, p(x_2) H(X_1|X_2=x_2) \\
        & = H(X_1|X_2) \,,
    \end{aligned}
    \end{equation}
    where the infimum over the stochastic channel $k$ follows from Gibb's inequality~\cite{cover1999elements}. An analogous derivation holds for $\delta(Y;X_2\backslash X_1)$. 
    Hence, we have that
    \begin{equation}
    \begin{aligned}
        \Ideltas & = \min \left(I(X_1;Y)-\delta(Y;X_1\backslash X_2), I(X_2;Y)-\delta(Y;X_2\backslash X_1)\right) \\
        & = \min \left(H(X_1)- H(X_1|X_2), H(X_2)- H(X_2|X_1)\right) \\
        & = \min \left(I(X_1;X_2), I(X_2;X_1)\right) \\
        & = I(X_1;X_2) \,,
    \end{aligned}
    \end{equation}
    which concludes the proof.
\end{proof}

\subsubsection{Proofs regarding \Ido}

Although in the original publication the authors of $\Ido$ claimed that the measure satisfies monotonicity and self-redundancy (Axiom 3 in Ref.~\cite{lyu2024explicit}), we point out that their proof does not effectively show self-redundancy and does not check the subset equality condition (\SE) in the monotonicity definition, i.e.\ that $\Icap=I_\cap(X_1;Y)$ if $X_1=f(X_2)$. In fact, it turns out that $\Ido$ does not satisfy \SE. Below, we discuss what implications this has on \SR and \Mz. 

As for \SR, we underline that a PID measure based on the definition of unique information can be built so that it satisfies the \SR axiom constructively~\cite{james2018unique}, since there is otherwise no way to prove directly that $I_\cap(X_1;Y)=I(X_1;Y)$, as one cannot use the definition of unique information for a single source.
Hence, we make the following remark:
\begin{remark} \label{rem:DO-SR}
    $\Ido$ satisfies \SR constructively, i.e.\ we define $\Ido(X_1;Y)\equiv I(X_1;Y)$.
\end{remark}

We can use this constructive definition to rewrite $\Ido$ in a more convenient form.
\begin{lemma} \label{lemma:DO-Red-CoI}
    The redundancy function can be written as a coinformation: $\Ido(X_1,X_2;Y) = I(X_1';X_2;Y)=I(X_1;X_2';Y)$ 
\end{lemma}
\begin{proof}
    The proof readily follows from the definition of unique information and \SR:
    \begin{equation} \label{eq:DO-CoI}
    \begin{aligned}
        \Ido(X_1,X_2;Y) & = \Ido(X_1;Y) - \Ido(X_1';Y|X_2) \\
        & = I(X_1;Y) - \Ido(X_1';Y|X_2) \\
        & = I(X_1';Y) - \Ido(X_1';Y|X_2) \\
        & =I(X_1;X_2';Y) \,,
    \end{aligned}
    \end{equation}
    where we used that $I(X_1';Y)=I(X_1;Y)$ since $H(X)=H(X')$ and $H(X|Y)=H(X'|Y)$. 
    An analogous procedure would give $\Idos =I(X_1;X_2';Y)$.
\end{proof}
Using this Lemma we can also obtain Lemma 2 of Ref.~\cite{lyu2024explicit} by noting that $X_1'$ and $X_2$ are conditionally independent given $Y$, in fact:
\begin{equation}
    \begin{aligned}
        p(X_1'=x_1,X_2=x_2|Y=y) & = \frac{p(X_2=x_2)}{p(Y=y)} p(X_1'=x_1,Y=y|X_2=x_2) \\
        & = \frac{p(X_2=x_2)}{p(Y=y)} p(X_1=x_1|Y=y) p(Y=y|X_2=x_2) \\
        & = p(X_1=x_1|Y=y) p(X_2=x_2|Y=y) \,.
    \end{aligned}
\end{equation}
Hence we have 
\begin{equation}
    \Idos = I(X_1';X_2;Z) = I(X_1';X_2) - I(X_1';X_2|Y) = I(X_1';X_2) \,.
\end{equation}

On the other hand, the lack of the equality condition in the monotonicity definition has an important consequence, which cannot be so easily alleviated: \Ido assigns non-zero unique information even when the sources are the same. Before proving this result, we provide a useful Lemma.
\begin{lemma} \label{lemma:DO-zeroIc}
    Consider the system $(X_1,X_2,Y)$ and the random variable $X_1'$ as defined in Eq.\eqref{eq:DO-def}. Then $I_Q(X_1';Y|X_2)=0$ if and only if $X_1\ind Y$.
\end{lemma}
\begin{proof}
    $I_Q(X_1';Y|X_2)=0$ means that $p(x_1',y|x_2)=p(x_1'|x_2)p(y|x_2)$. This happens when
\begin{align}
    p(x_1',y|x_2) & = p(x_1|y)p(y|x_2) \\
    p(x_1'|x_2) & = \sum_{\tilde{y}} p(x_1',\tilde{y}|x_2) = \sum_{\tilde{y}} p(x_1|\tilde{y})p(\tilde{y}|x_2) \,.
\end{align}
    Hence, for the partitioning to hold, we have the condition
    \begin{equation}
        p(x_1|y) = \sum_{\tilde{y}} p(x_1|\tilde{y})p(\tilde{y}|x_2) \,,
    \end{equation}
    since the RHS does not depend on $y$, it implies that $X_1\ind Y$. 
\end{proof}

Using this Lemma, we have:
\begin{proposition} \label{prop:DO-noSE}
    $\Ido$ does not satisfy \SE.
\end{proposition}
\begin{proof}
    Consider the PID with identical first and second source $X_1=X_2$ and generic target $Y$. From Lemma\ref{lemma:DO-Red-CoI} we have that
    \begin{equation}
        \begin{aligned}
            \Ido(X_1,X_1;Y) & = I(X_1';Y) - I(X_1';Y|X_2) \\
         & = I(X_1;Y) - I(X_1';Y|X_2) \\
         & < I(X_1;Y) \,,
        \end{aligned}
    \end{equation}
    where the strict inequality follows from Lemma \ref{lemma:DO-zeroIc} when $X_1\notind Y$, since $I(X_1';Y|X_2)>0$. 
    Hence $\Ido$ does not satisfy \SE.
\end{proof}

A similar reasoning can be used for the following property.
\begin{proposition} \label{prop:DO-noLB}
    $\Ido$ does not satisfy \LB.
\end{proposition}
\begin{proof}
     We prove it by counterexample. Consider the PID with first source $X_1$ and second source the joint state $X_1 X_2$ and a generic target $Y$, then using Lemma \ref{lemma:DO-Red-CoI} we have
     \begin{equation}
     \begin{aligned}
         \Ido(X_1,X_1X_2;Y) & = I(X_1';Y) - I(X_1';Y|X_1X_2) \\
         & = I(X_1;Y) - I(X_1';Y|X_1X_2) \\
         & < I(X_1;Y) \,,
        \end{aligned}
    \end{equation}
    where the strict inequality follows from Lemma \ref{lemma:DO-zeroIc} when $X_1\notind Y$. On the other hand, \LB would require $\Ido(X_1,X_1X_2;Y)\geq I(X_1;Y)$ since $X_1=f(X_1 X_2)$. Hence $\Ido$ does not satisfy \LB.
\end{proof}

\begin{proposition} \label{prop:DO-ID}
    \Ido satisfies \ID.
\end{proposition}
\begin{proof}
    Considering the TBC system, and denoting as $y=(\tilde{x}_1,\tilde{x}_2)$ the outcome of the target, we have
    \begin{equation}
    \begin{aligned}
        p(X_1'=x_1,Y=y|X_2=x_2) & = p(X_1=x_1|Y=y)p(Y=y|X_2=x_2) \\
        & = p(X_1=x_1|X_1=\tilde{x}_1, X_2=\tilde{x}_2)p(X_1=\tilde{x}_1, X_2=\tilde{x}_2|X_2=x_2) \\
        & = p(X_1=x_1|X_1=x_1, X_2=x_2)p(X_1=x_1, X_2=x_2|X_2=x_2) \delta(\tilde{x}_1,x_1)\delta(\tilde{x}_2,x_2) \\
        & = p(X_1=x_1, X_2=x_2|X_2=x_2) \,,
    \end{aligned}
    \end{equation}
    which means that $X_1'=X_1$. Hence \ID follows from Eq.\eqref{eq:DO-def2}.
\end{proof}

\subsubsection{Proofs regarding \IRAV}
The following proofs only hold for bivariate PIDs ($n=2$ sources).

\begin{proposition} \label{prop:RAV-SR}
    $\IRAV$ satisfies \SR.
\end{proposition}
\begin{proof}
    \begin{equation}
        \begin{aligned}
            \IRAV(X_1;Y) & = \max_f I(X_1;Y;f(X_1)) \\
            & = \max_f \left( I(f(X_1);Y) - I(f(X_1);Y|X_1) \right) \\
            & = \max_f I(f(X_1);Y) \\
            & = I(X_1;Y) \,.
        \end{aligned}
    \end{equation}
\end{proof}

\begin{proposition} \label{prop:RAV-Mz}
    $\IRAV$ satisfies \Mz for $n=2$ sources.
\end{proposition}
\begin{proof}
    For $n=2$ sources, for the inequality condition we only need to prove $\IRAV(X_i;Y)-\IRAV(X_1,X_2;Y)\ge0$. Considering e.g.\ $i=1$:
    \begin{equation}
        \begin{aligned}
            \IRAV(X_i;Y)-\IRAV(X_1,X_2;Y) & = I(X_1;Y) - \max_f I(X_1;X_2;Y;f(X_1,X_2)) \\
            & = I(X_1;Y) - I(X_1;X_2;Y) + \min_f I(X_1;X_2;Y|f(X_1,X_2)) \\
            & = I(X_1;Y|X_2) + \min_f I(X_1;X_2;Y|f(X_1,X_2)) \ge 0 \,.
        \end{aligned}
    \end{equation}
    For the equality condition, if we have $X_2 = g(X_1)$ then 
    \begin{equation}
        \begin{aligned}
            \IRAV(X_1,X_2;Y) = \max_f I(X_1;g(X_1);Y;f(X_1,X_2)) = \max_f I(X_1;Y;f(X_1)) = \IRAV(X_1;Y) \,,
        \end{aligned}
    \end{equation}
    which concludes the proof.
\end{proof}

\begin{proposition} \label{prop:RAV-ID}
    $\IRAV$ satisfies \ID.
\end{proposition}
\begin{proof}
    Consider the joint PID target $Y=(X_1,X_2)$. It follows that
    \begin{equation}
        \begin{aligned}
            \IRAVs & = \max_f I(X_1;X_2;Y;f(X_1,X_2)) \\
            & = \max_f \left( I(X_1;X_2;Y;f(X_1,X_2)) - I(X_1;X_2;f(X_1,X_2)|X_1X_2)\right) \\
            & = \max_f \left( I(X_1;X_2) - I(X_1,X_2|f(X_1,X_2)) \right) \\
            & = I(X_1;X_2) \,,
        \end{aligned}
    \end{equation}
    where in the last step it is enough to take $f(X_1,X_2)=X_i$. 
    \end{proof}

\begin{proposition}\label{prop:RAV-TE}
    $\IRAV$ satisfies \TE for $n=2$ sources.
\end{proposition}
\begin{proof}
    For $n=2$ we just need to prove that $\IRAV(X_1,Y;Y)=\IRAV(X_1;Y)$.
    \begin{equation}
        \begin{aligned}
            \IRAV(X_1,Y;Y) & = \max_f I(X_1;Y;Y;f(X_1,Y)) \\
            & = \max_f I(X_1;Y;f(X_1,Y)) \\
            & = I(X_1;Y) - \min_f I(X_1;Y;Y|f(X_1,Y)) \\
            & = I(X_1;Y) \,,
        \end{aligned}
    \end{equation}
    where we took $f(X_1,Y)=\{X_1,Y\}$.
\end{proof}

\begin{proposition}\label{prop:RAV-LPz}
    $\IRAV$ satisfies \LPz.
\end{proposition}
\begin{proof}
    Redundancy is non-negative from \GP, and unique information is non-negative from \Mz and \GP. It remains to prove non-negativity for synergy: 
    \begin{equation}
        \begin{aligned}
            I_\partial(X_1X_2;Y) & = \IRAVs - I(X_1;X_2;Y) \\
            & = I(X_1;X_2;Y) -\min_f I(X_1;X_2;Y|f(X_1,X_2)) - I(X_1;X_2;Y) \\
            & = \max_f \left( - I(X_1;X_2;Y|f(X_1,X_2)) \right) \\
            & = \max_f \left( I(X_1;Y|X_2,f(X_1,X_2)) - I(X_1;Y|f(X_1,X_2)) \right)  \\
            & \geq  I(X_1;Y|X_2,X_2) - I(X_1;Y|X_2) \\
            & = 0\,,
        \end{aligned}
    \end{equation}
    where in the second last step we took $f(X_1,X_2)=X_2$.
\end{proof}

\bibliographystyle{ieeetr}
\bibliography{bib}

@article{williams2010nonnegative,
  title={Nonnegative decomposition of multivariate information},
  author={Williams, Paul L and Beer, Randall D},
  journal={arXiv preprint arXiv:1004.2515},
  year={2010}
}

@article{combrisson2025higher,
  title={Higher-order and distributed synergistic functional interactions encode information gain in goal-directed learning},
  author={Combrisson, Etienne and Basanisi, Ruggero and Neri, Matteo and Auzias, Guillaume and Petri, Giovanni and Marinazzo, Daniele and Panzeri, Stefano and Brovelli, Andrea},
  journal={Nature Communications},
  volume={16},
  number={1},
  pages={7179},
  year={2025},
  publisher={Nature Publishing Group UK London}
}

@article{neri2025taxonomy,
  title={A taxonomy of neuroscientific strategies based on interaction orders},
  author={Neri, Matteo and Brovelli, Andrea and Castro, Samy and Fraisopi, Fausto and Gatica, Marilyn and Herzog, Ruben and Mediano, Pedro AM and Mindlin, Ivan and Petri, Giovanni and Bor, Daniel and others},
  journal={European Journal of Neuroscience},
  volume={61},
  number={3},
  pages={e16676},
  year={2025},
  publisher={Wiley Online Library}
}

@article{coronel2025integrated,
  title={An integrated computational approach for diversity-sensitive personalized medicine},
  author={Coronel-Oliveros, Carlos and Gatica, Marilyn and Herzog, Rub{\'e}n and Neri, Matteo},
  journal={Neuroscience},
  year={2025},
  publisher={Elsevier}
}

@article{santoro2025beyond,
  title={Beyond Pairwise Interactions: Charting Higher-Order Models of Brain Function},
  author={Santoro, Andrea and Neri, Matteo and Poetto, Simone and Orsenigo, Davide and Diano, Matteo and Gatica, Marilyn and Petri, GIovanni},
  journal={bioRxiv},
  pages={2025--06},
  year={2025},
  publisher={Cold Spring Harbor Laboratory}
}

@article{battiston2020networks,
  title={Networks beyond pairwise interactions: Structure and dynamics},
  author={Battiston, Federico and Cencetti, Giulia and Iacopini, Iacopo and Latora, Vito and Lucas, Maxime and Patania, Alice and Young, Jean-Gabriel and Petri, Giovanni},
  journal={Physics reports},
  volume={874},
  pages={1--92},
  year={2020},
  publisher={Elsevier}
}

@article{robiglio2025synergistic,
  title={Synergistic signatures of group mechanisms in higher-order systems},
  author={Robiglio, Thomas and Neri, Matteo and Coppes, Davide and Agostinelli, Cosimo and Battiston, Federico and Lucas, Maxime and Petri, Giovanni},
  journal={Physical review letters},
  volume={134},
  number={13},
  pages={137401},
  year={2025},
  publisher={APS}
}

@article{mcgill1954multivariate,
  title={Multivariate information transmission},
  author={McGill, William},
  journal={Transactions of the IRE Professional Group on Information Theory},
  volume={4},
  number={4},
  pages={93--111},
  year={1954},
  publisher={IEEE}
}

@article{kay2018exact,
  title={Exact partial information decompositions for Gaussian systems based on dependency constraints},
  author={Kay, Jim W and Ince, Robin AA},
  journal={Entropy},
  volume={20},
  number={4},
  pages={240},
  year={2018},
  publisher={MDPI}
}

@article{rosas2022disentangling,
  title={Disentangling high-order mechanisms and high-order behaviours in complex systems},
  author={Rosas, Fernando E and Mediano, Pedro AM and Luppi, Andrea I and Varley, Thomas F and Lizier, Joseph T and Stramaglia, Sebastiano and Jensen, Henrik J and Marinazzo, Daniele},
  journal={Nature Physics},
  volume={18},
  number={5},
  pages={476--477},
  year={2022},
  publisher={Nature Publishing Group UK London}
}

@article{liardi2025simple,
  title={Simple physical systems as a reference for multivariate information dynamics},
  author={Liardi, Alberto and Sas, Madalina I and Blackburne, George and Knottenbelt, William J and Mediano, Pedro AM and Jensen, Henrik Jeldtoft},
  journal={arXiv preprint arXiv:2504.10372},
  year={2025}
}

@article{barrett2015exploration,
  title={Exploration of synergistic and redundant information sharing in static and dynamical Gaussian systems},
  author={Barrett, Adam B},
  journal={Physical Review E},
  volume={91},
  number={5},
  pages={052802},
  year={2015},
  publisher={APS}
}

@article{ince2017measuring,
  title={Measuring multivariate redundant information with pointwise common change in surprisal},
  author={Ince, Robin AA},
  journal={Entropy},
  volume={19},
  number={7},
  pages={318},
  year={2017},
  publisher={Multidisciplinary Digital Publishing Institute}
}

@incollection{griffith2014quantifying,
  title={Quantifying synergistic mutual information},
  author={Griffith, Virgil and Koch, Christof},
  booktitle={Guided self-organization: inception},
  pages={159--190},
  year={2014},
  publisher={Springer}
}

@article{harder2013bivariate,
  title={Bivariate measure of redundant information},
  author={Harder, Malte and Salge, Christoph and Polani, Daniel},
  journal={Physical Review E},
  volume={87},
  number={1},
  pages={012130},
  year={2013},
  publisher={APS}
}

@inproceedings{bertschinger2013shared,
  title={Shared information—New insights and problems in decomposing information in complex systems},
  author={Bertschinger, Nils and Rauh, Johannes and Olbrich, Eckehard and Jost, J{\"u}rgen},
  booktitle={Proceedings of the European conference on complex systems 2012},
  pages={251--269},
  year={2013},
  organization={Springer}
}

@article{chan2017gene,
  title={Gene regulatory network inference from single-cell data using multivariate information measures},
  author={Chan, Thalia E and Stumpf, Michael PH and Babtie, Ann C},
  journal={Cell systems},
  volume={5},
  number={3},
  pages={251--267},
  year={2017},
  publisher={Elsevier}
}

@article{chen2018evaluating,
  title={Evaluating methods of inferring gene regulatory networks highlights their lack of performance for single cell gene expression data},
  author={Chen, Shuonan and Mar, Jessica C},
  journal={BMC bioinformatics},
  volume={19},
  number={1},
  pages={1--21},
  year={2018},
  publisher={BioMed Central}
}

@article{wibral2017partial,
  title={Partial information decomposition as a unified approach to the specification of neural goal functions},
  author={Wibral, Michael and Priesemann, Viola and Kay, Jim W and Lizier, Joseph T and Phillips, William A},
  journal={Brain and cognition},
  volume={112},
  pages={25--38},
  year={2017},
  publisher={Elsevier}
}

@article{james2018unique,
  title={Unique information via dependency constraints},
  author={James, Ryan G and Emenheiser, Jeffrey and Crutchfield, James P},
  journal={Journal of Physics A: Mathematical and Theoretical},
  volume={52},
  number={1},
  pages={014002},
  year={2018},
  publisher={IOP Publishing}
}

@article{barrett2010multivariate,
  title={Multivariate Granger causality and generalized variance},
  author={Barrett, Adam B and Barnett, Lionel and Seth, Anil K},
  journal={Physical Review E},
  volume={81},
  number={4},
  pages={041907},
  year={2010},
  publisher={APS}
}

@article{rosas2019quantifying,
  title={Quantifying high-order interdependencies via multivariate extensions of the mutual information},
  author={Rosas, Fernando E and Mediano, Pedro AM and Gastpar, Michael and Jensen, Henrik J},
  journal={Physical Review E},
  volume={100},
  number={3},
  pages={032305},
  year={2019},
  publisher={APS}
}

@article{jaynes1957information,
  title={Information theory and statistical mechanics},
  author={Jaynes, Edwin T},
  journal={Physical review},
  volume={106},
  number={4},
  pages={620},
  year={1957},
  publisher={APS}
}

@article{luppi2020synergisticCore,
  title={A synergistic core for human brain evolution and cognition},
  author={Luppi, Andrea I and Mediano, Pedro AM and Rosas, Fernando E and Holland, Negin and Fryer, Tim D and O’Brien, John T and Rowe, James B and Menon, David K and Bor, Daniel and Stamatakis, Emmanuel A},
  journal={BioRxiv},
  year={2020},
  publisher={Cold Spring Harbor Laboratory}
}

@article{beer2015information,
  title={Information processing and dynamics in minimally cognitive agents},
  author={Beer, Randall D and Williams, Paul L},
  journal={Cognitive science},
  volume={39},
  number={1},
  pages={1--38},
  year={2015},
  publisher={Wiley Online Library}
}

@inproceedings{finn2018quantifying,
  title={Quantifying information modification in cellular automata using pointwise partial information decomposition},
  author={Finn, Conor and Lizier, Joseph T},
  booktitle={Artificial Life Conference Proceedings},
  pages={386--387},
  year={2018},
  organization={MIT Press One Rogers Street, Cambridge, MA 02142-1209, USA journals-info~…}
}

@article{rosas2018information,
  title={An information-theoretic approach to self-organisation: Emergence of complex interdependencies in coupled dynamical systems},
  author={Rosas, Fernando and Mediano, Pedro AM and Ugarte, Mart{\'\i}n and Jensen, Henrik J},
  journal={Entropy},
  volume={20},
  number={10},
  pages={793},
  year={2018},
  publisher={MDPI}
}

@article{quax2017quantifying,
  title={Quantifying synergistic information using intermediate stochastic variables},
  author={Quax, Rick and Har-Shemesh, Omri and Sloot, Peter MA},
  journal={Entropy},
  volume={19},
  number={2},
  pages={85},
  year={2017},
  publisher={MDPI}
}

@article{griffith2014intersection,
  title={Intersection information based on common randomness},
  author={Griffith, Virgil and Chong, Edwin KP and James, Ryan G and Ellison, Christopher J and Crutchfield, James P},
  journal={Entropy},
  volume={16},
  number={4},
  pages={1985--2000},
  year={2014},
  publisher={Molecular Diversity Preservation International (MDPI)}
}

@article{griffith2015quantifying,
  title={Quantifying redundant information in predicting a target random variable},
  author={Griffith, Virgil and Ho, Tracey},
  journal={Entropy},
  volume={17},
  number={7},
  pages={4644--4653},
  year={2015},
  publisher={MDPI}
}

@article{schneidman2003synergy,
  title={Synergy, redundancy, and independence in population codes},
  author={Schneidman, Elad and Bialek, William and Berry, Michael J},
  journal={Journal of Neuroscience},
  volume={23},
  number={37},
  pages={11539--11553},
  year={2003},
  publisher={Soc Neuroscience}
}

@article{tax2017partial,
  title={The partial information decomposition of generative neural network models},
  author={Tax, Tycho MS and Mediano, Pedro AM and Shanahan, Murray},
  journal={Entropy},
  volume={19},
  number={9},
  pages={474},
  year={2017},
  publisher={MDPI}
}

@article{rosas2020operational,
  title={An operational information decomposition via synergistic disclosure},
  author={Rosas, Fernando E and Mediano, Pedro AM and Rassouli, Borzoo and Barrett, Adam B},
  journal={Journal of Physics A: Mathematical and Theoretical},
  volume={53},
  number={48},
  pages={485001},
  year={2020},
  publisher={IOP Publishing}
}

@article{varley2023partial,
  title={Partial entropy decomposition reveals higher-order information structures in human brain activity},
  author={Varley, Thomas F and Pope, Maria and Maria Grazia and Joshua and Sporns, Olaf},
  journal={Proceedings of the National Academy of Sciences},
  volume={120},
  number={30},
  pages={e2300888120},
  year={2023},
  publisher={National Acad Sciences}
}

@article{gatica2021high,
  title={High-order interdependencies in the aging brain},
  author={Gatica, Marilyn and Cofr{\'e}, Rodrigo and Mediano, Pedro AM and Rosas, Fernando E and Orio, Patricio and Diez, Ibai and Swinnen, Stephan P and Cortes, Jesus M},
  journal={Brain connectivity},
  volume={11},
  number={9},
  pages={734--744},
  year={2021},
  publisher={Mary Ann Liebert, Inc., publishers 140 Huguenot Street, 3rd Floor New~…}
}

@article{gatica2022high,
  title={High-order functional redundancy in ageing explained via alterations in the connectome in a whole-brain model},
  author={Gatica, Marilyn and E. Rosas, Fernando and AM Mediano, Pedro and Diez, Ibai and P. Swinnen, Stephan and Orio, Patricio and Cofr{\'e}, Rodrigo and M. Cortes, Jesus},
  journal={PLoS Computational Biology},
  volume={18},
  number={9},
  pages={e1010431},
  year={2022},
  publisher={Public Library of Science San Francisco, CA USA}
}

@article{varley2023multivariate,
  title={Multivariate information theory uncovers synergistic subsystems of the human cerebral cortex},
  author={Varley, Thomas F and Pope, Maria and Faskowitz, Joshua and Sporns, Olaf},
  journal={Communications biology},
  volume={6},
  number={1},
  pages={451},
  year={2023},
  publisher={Nature Publishing Group UK London}
}

@article{luppi2024information,
  title={Information decomposition and the informational architecture of the brain},
  author={Luppi, Andrea I and Rosas, Fernando E and Mediano, Pedro AM and Menon, David K and Stamatakis, Emmanuel A},
  journal={Trends in Cognitive Sciences},
  year={2024},
  publisher={Elsevier}
}

@article{ince2017partial,
  title={The Partial Entropy Decomposition: Decomposing multivariate entropy and mutual information via pointwise common surprisal},
  author={Ince, Robin AA},
  journal={arXiv preprint arXiv:1702.01591},
  year={2017}
}

@article{varley2024generalized,
  title={Generalized decomposition of multivariate information},
  author={Varley, Thomas F},
  journal={Plos one},
  volume={19},
  number={2},
  pages={e0297128},
  year={2024},
  publisher={Public Library of Science San Francisco, CA USA}
}

@article{varley2024scalable,
  title={A scalable synergy-first backbone decomposition of higher-order structures in complex systems},
  author={Varley, Thomas F},
  journal={npj Complexity},
  volume={1},
  number={1},
  pages={9},
  year={2024},
  publisher={Nature Publishing Group UK London}
}

@PREAMBLE{
 "\providecommand{\noopsort}[1]{}" 
 # "\providecommand{\singleletter}[1]{#1}%" 
}

@article{amari2001information,
  title={Information geometry on hierarchy of probability distributions},
  author={Amari, S-I},
  journal={IEEE transactions on information theory},
  volume={47},
  number={5},
  pages={1701--1711},
  year={2001},
  publisher={IEEE}
}

@article{shannon1948mathematical,
  title={A mathematical theory of communication},
  author={Shannon, Claude Elwood},
  journal={The Bell System Technical Journal},
  volume={27},
  number={3},
  pages={379--423},
  year={1948},
  publisher={Nokia Bell Labs}
}

@article{gutknecht2025shannon,
  title={Shannon invariants: {A} scalable approach to information decomposition},
  author={Gutknecht, Aaron J and Rosas, Fernando E and Ehrlich, David A and Makkeh, Abdullah and Mediano, Pedro A M and Wibral, Michael},
  journal={arXiv preprint arXiv:2504.15779},
  year={2025}
}

@article{rosas2016understanding,
  title={Understanding interdependency through complex information sharing},
  author={Rosas, Fernando and Ntranos, Vasilis and Ellison, Christopher J and Pollin, Sofie and Verhelst, Marian},
  journal={Entropy},
  volume={18},
  number={2},
  pages={38},
  year={2016},
  publisher={Multidisciplinary Digital Publishing Institute}
}

@book{cover1999elements,
  title={Elements of Information Theory},
  author={Cover, Thomas M and Thomas, Joy A},
  year={1999},
  publisher={John Wiley \& Sons}
}

@article{liardi2025null,
  title={Null models for comparing information decomposition across complex systems},
  author={Liardi, Alberto and Rosas, Fernando E and Carhart-Harris, Robin L and Blackburne, George and Bor, Daniel and Mediano, Pedro AM},
  journal={PLOS Computational Biology},
  volume={21},
  number={11},
  pages={e1013629},
  year={2025},
  publisher={Public Library of Science San Francisco, CA USA}
}

@article{proca2024synergistic,
  title={Synergistic information supports modality integration and flexible learning in neural networks solving multiple tasks},
  author={Proca, Alexandra M and Rosas, Fernando E and Luppi, Andrea I and Bor, Daniel and Crosby, Matthew and Mediano, Pedro AM},
  journal={PLOS Computational Biology},
  volume={20},
  number={6},
  pages={e1012178},
  year={2024},
  publisher={Public Library of Science San Francisco, CA USA}
}

@article{down2025logarithmic,
  author={Down, Keenan J. A. and Mediano, Pedro A. M.},
  journal={IEEE Transactions on Information Theory}, 
  title={A Logarithmic Decomposition and a Signed Measure Space for Entropy}, 
  year={2026},
  volume={72},
  number={1},
  pages={17-41},
  doi={10.1109/TIT.2025.3635730}
}

@article{down2025algebraic,
  title={Algebraic representations of entropy and fixed-sign information quantities},
  author={Down, Keenan JA and Mediano, Pedro AM},
  journal={Entropy},
  volume={27},
  number={2},
  pages={151},
  year={2025},
  publisher={MDPI}
}

@article{down2026synergistic,
  title={Synergistic and redundant information dynamics are modulated by Alzheimer's disease and cognitive impairment},
  author={Down, Keenan JA and Huntley, Jonathan and Mediano, Pedro AM and Bor, Daniel},
  journal={bioRxiv},
  pages={2026--02},
  year={2026},
  publisher={Cold Spring Harbor Laboratory}
}

@book{pearl2009causality,
  title={Causality},
  author={Pearl, Judea},
  year={2009},
  publisher={Cambridge University Press}
}

@article{pearl2009causal,
  title={Causal inference in statistics: An overview},
  author={Pearl, Judea},
  journal={Statistics Surveys},
  volume={3},
  pages={96--146},
  year={2009}
}

@article{pearl1995causal,
  title={Causal diagrams for empirical research},
  author={Pearl, Judea},
  journal={Biometrika},
  volume={82},
  number={4},
  pages={669--688},
  year={1995},
  publisher={Oxford University Press}
}

@inproceedings{pearl2012calculus,
  title={The do-calculus revisited},
  author={Pearl, Judea},
  booktitle={Proceedings of the Twenty-Eighth Conference on Uncertainty in Artificial Intelligence},
  pages={3--11},
  year={2012}
}

@article{bertschinger2014quantifying,
  title={Quantifying unique information},
  author={Bertschinger, Nils and Rauh, Johannes and Olbrich, Eckehard and Jost, J{\"u}rgen and Ay, Nihat},
  journal={Entropy},
  volume={16},
  number={4},
  pages={2161--2183},
  year={2014},
  publisher={Multidisciplinary Digital Publishing Institute}
}

@article{chechik2001group,
  title={Group redundancy measures reveal redundancy reduction in the auditory pathway},
  author={Chechik, Gal and Globerson, Amir and Anderson, M and Young, E and Nelken, Israel and Tishby, Naftali},
  journal={Advances in neural information processing systems},
  volume={14},
  year={2001}
}

@inproceedings{rauh2014reconsidering,
  title={Reconsidering unique information: {Towards} a multivariate information decomposition},
  author={Rauh, Johannes and Bertschinger, Nils and Olbrich, Eckehard and Jost, J{\"u}rgen},
  booktitle={2014 IEEE International Symposium on Information Theory},
  pages={2232--2236},
  year={2014},
  organization={IEEE}
}

@article{finn2018pointwise,
  title={Pointwise partial information decomposition using the specificity and ambiguity lattices},
  author={Finn, Conor and Lizier, Joseph T},
  journal={Entropy},
  volume={20},
  number={4},
  pages={297},
  year={2018},
  publisher={MDPI}
}

@article{kolchinsky2022novel,
  title={A novel approach to the partial information decomposition},
  author={Kolchinsky, Artemy},
  journal={Entropy},
  volume={24},
  number={3},
  pages={403},
  year={2022},
  publisher={MDPI}
}

@article{james2018dit,
  title={``dit``: {A Python} package for discrete information theory},
  author={James, Ryan G and Ellison, Christopher J and Crutchfield, James P},
  journal={Journal of Open Source Software},
  volume={3},
  number={25},
  pages={738},
  year={2018}
}

@article{schneidman2006weak,
  title={Weak pairwise correlations imply strongly correlated network states in a neural population},
  author={Schneidman, Elad and Berry, Michael J and Segev, Ronen and Bialek, William},
  journal={Nature},
  volume={440},
  number={7087},
  pages={1007--1012},
  year={2006},
  publisher={Nature Publishing Group UK London}
}

@article{rauh2023continuity,
  title={Continuity and additivity properties of information decompositions},
  author={Rauh, Johannes and Banerjee, Pradeep Kr and Olbrich, Eckehard and Mont{\'u}far, Guido and Jost, J{\"u}rgen},
  journal={International Journal of Approximate Reasoning},
  volume={161},
  pages={108979},
  year={2023},
  publisher={Elsevier}
}

@article{lyu2025whole,
  title={The Whole Is Less than the Sum of Parts: Subsystem Inconsistency in Partial Information Decomposition},
  author={Lyu, Aobo and Clark, Andrew and Raviv, Netanel},
  journal={arXiv preprint arXiv:2510.14864},
  year={2025}
}

@article{banerjee2015synergy,
  title={Synergy, redundancy and common information},
  author={Banerjee, Pradeep Kr and Griffith, Virgil},
  journal={arXiv preprint arXiv:1509.03706},
  year={2015}
}

@article{goodwell2017temporal,
  title={Temporal information partitioning: Characterizing synergy, uniqueness, and redundancy in interacting environmental variables},
  author={Goodwell, Allison E and Kumar, Praveen},
  journal={Water Resources Research},
  volume={53},
  number={7},
  pages={5920--5942},
  year={2017},
  publisher={Wiley Online Library}
}

@article{sigtermans2020partial,
  title={A Partial Information Decomposition Based on Causal Tensors},
  author={Sigtermans, David},
  journal={arXiv preprint arXiv:2001.10481},
  year={2020}
}

@article{sigtermans2020towards,
  title={Towards a framework for observational causality from time series: when Shannon meets Turing},
  author={Sigtermans, David},
  journal={Entropy},
  volume={22},
  number={4},
  pages={426},
  year={2020},
  publisher={MDPI}
}

@inproceedings{banerjee2018unique,
  title={Unique informations and deficiencies},
  author={Banerjee, Pradeep Kr and Olbrich, Eckehard and Jost, J{\"u}rgen and Rauh, Johannes},
  booktitle={2018 56th Annual Allerton Conference on Communication, Control, and Computing (Allerton)},
  pages={32--38},
  year={2018},
  organization={IEEE}
}

@article{james2018uniquekey,
  title={Unique information and secret key agreement},
  author={James, Ryan G and Emenheiser, Jeffrey and Crutchfield, James P},
  journal={Entropy},
  volume={21},
  number={1},
  pages={12},
  year={2018},
  publisher={MDPI}
}

@inproceedings{lyu2024explicit,
  title={Explicit formula for partial information decomposition},
  author={Lyu, Aobo and Clark, Andrew and Raviv, Netanel},
  booktitle={2024 IEEE International Symposium on Information Theory (ISIT)},
  pages={2329--2334},
  year={2024},
  organization={IEEE}
}

@inproceedings{niu2019measure,
  title={A measure of synergy, redundancy, and unique information using information geometry},
  author={Niu, Xueyan and Quinn, Christopher J},
  booktitle={2019 IEEE International Symposium on Information Theory (ISIT)},
  pages={3127--3131},
  year={2019},
  organization={IEEE}
}

@article{li2011connection,
  title={On a connection between information and group lattices},
  author={Li, Hua and Chong, Edwin KP},
  journal={Entropy},
  volume={13},
  number={3},
  pages={683--708},
  year={2011},
  publisher={MDPI}
}

@inproceedings{venkatesh2022partial,
  title={Partial information decomposition via deficiency for multivariate gaussians},
  author={Venkatesh, Praveen and Schamberg, Gabriel},
  booktitle={2022 IEEE International Symposium on Information Theory (ISIT)},
  pages={2892--2897},
  year={2022},
  organization={IEEE}
}

@article{mages2023decomposing,
  title={Decomposing and tracing mutual information by quantifying reachable decision regions},
  author={Mages, Tobias and Rohner, Christian},
  journal={Entropy},
  volume={25},
  number={7},
  pages={1014},
  year={2023},
  publisher={MDPI}
}

@article{makkeh2021introducing,
  title={Introducing a differentiable measure of pointwise shared information},
  author={Makkeh, Abdullah and Gutknecht, Aaron J and Wibral, Michael},
  journal={Physical Review E},
  volume={103},
  number={3},
  pages={032149},
  year={2021},
  publisher={APS}
}

@article{rauh2017extractable,
  title={On extractable shared information},
  author={Rauh, Johannes and Banerjee, Pradeep Kr and Olbrich, Eckehard and Jost, J{\"u}rgen and Bertschinger, Nils},
  journal={Entropy},
  volume={19},
  number={7},
  pages={328},
  year={2017},
  publisher={MDPI}
}

@article{kay2024partial,
  title={A partial information decomposition for multivariate gaussian systems based on information geometry},
  author={Kay, Jim W},
  journal={Entropy},
  volume={26},
  number={7},
  pages={542},
  year={2024},
  publisher={MDPI}
}

@article{ehrlich2024partial,
  title={Partial information decomposition for continuous variables based on shared exclusions: Analytical formulation and estimation},
  author={Ehrlich, David A and Schick-Poland, Kyle and Makkeh, Abdullah and Lanfermann, Felix and Wollstadt, Patricia and Wibral, Michael},
  journal={Physical Review E},
  volume={110},
  number={1},
  pages={014115},
  year={2024},
  publisher={APS}
}

@article{schick2021partial,
  title={A partial information decomposition for discrete and continuous variables},
  author={Schick-Poland, Kyle and Makkeh, Abdullah and Gutknecht, Aaron J and Wollstadt, Patricia and Sturm, Anja and Wibral, Michael},
  journal={arXiv preprint arXiv:2106.12393},
  year={2021}
}

@article{finn2018probability,
  title={Probability mass exclusions and the directed components of mutual information},
  author={Finn, Conor and Lizier, Joseph T},
  journal={Entropy},
  volume={20},
  number={11},
  pages={826},
  year={2018},
  publisher={MDPI}
}

@article{nyquist1924certain,
  title={Certain Factors Affecting Telegraph Speed 1},
  author={Nyquist, H},
  journal={Bell System Technical Journal},
  volume={3},
  number={2},
  pages={324--346},
  year={1924},
  publisher={Wiley Online Library}
}

@article{hartley1928transmission,
  title={Transmission of information 1},
  author={Hartley, Ralph VL},
  journal={Bell System technical journal},
  volume={7},
  number={3},
  pages={535--563},
  year={1928},
  publisher={Wiley Online Library}
}

@book{clausius1865ueber,
  title={Ueber verschiedene f{\"u}r die Anwendung bequeme Formen der Hauptgleichungen der mechanischen W{\"a}rmetheorie: vorgetragen in der naturforsch. Gesellschaft den 24. April 1865},
  author={Clausius, Rudolf},
  year={1865},
  publisher={Verlag nicht ermittelbar}
}

@incollection{boltzmann1970weitere,
  title={Weitere studien {\"u}ber das w{\"a}rmegleichgewicht unter gasmolek{\"u}len},
  author={Boltzmann, Ludwig},
  booktitle={Kinetische Theorie II: Irreversible Prozesse Einf{\"u}hrung und Originaltexte},
  pages={115--225},
  year={1970},
  publisher={Springer}
}

@article{boltzmann1872weitere,
  author    = {Boltzmann, Ludwig},
  title     = {Weitere Studien über das Wärmegleichgewicht unter Gasmolekülen},
  journal   = {Sitzungsberichte der Kaiserlichen Akademie der Wissenschaften in Wien},
  series    = {Mathematisch-Naturwissenschaftliche Classe},
  volume    = {66},
  number    = {2},
  pages     = {275--370},
  year      = {1872},
  language  = {German}
}

@article{james2017multivariate,
  title={Multivariate dependence beyond Shannon information},
  author={James, Ryan G and Crutchfield, James P},
  journal={Entropy},
  volume={19},
  number={10},
  pages={531},
  year={2017},
  publisher={MDPI}
}

@inproceedings{bell2003co,
  title={The co-information lattice},
  author={Bell, Anthony J},
  booktitle={Proceedings of the fifth international workshop on independent component analysis and blind signal separation: ICA},
  volume={2003},
  year={2003}
}

@article{yeung1991new,
  title={A new outlook on Shannon's information measures},
  author={Yeung, Raymond W},
  journal={IEEE transactions on information theory},
  volume={37},
  number={3},
  pages={466--474},
  year={1991},
  publisher={IEEE}
}

@inproceedings{bertschinger2014blackwell,
  title={The Blackwell relation defines no lattice},
  author={Bertschinger, Nils and Rauh, Johannes},
  booktitle={2014 IEEE International Symposium on Information Theory},
  pages={2479--2483},
  year={2014},
  organization={IEEE}
}

@article{blackwell1953equivalent,
  title={Equivalent comparisons of experiments},
  author={Blackwell, David},
  journal={The annals of mathematical statistics},
  pages={265--272},
  year={1953},
  publisher={JSTOR}
}

@book{torgersen1991comparison,
  title={Comparison of statistical experiments},
  author={Torgersen, Erik},
  volume={36},
  year={1991},
  publisher={Cambridge University Press}
}

@article{rauh2017coarse,
  title={Coarse-graining and the Blackwell order},
  author={Rauh, Johannes and Banerjee, Pradeep Kr and Olbrich, Eckehard and Jost, J{\"u}rgen and Bertschinger, Nils and Wolpert, David},
  journal={Entropy},
  volume={19},
  number={10},
  pages={527},
  year={2017},
  publisher={MDPI}
}

@inproceedings{wolf2004zero,
  title={Zero-error information and applications in cryptography},
  author={Wolf, Stefan and Wultschleger, J},
  booktitle={Information Theory Workshop},
  pages={1--6},
  year={2004},
  organization={IEEE}
}

@article{mediano2025toward,
  title={Toward a unified taxonomy of information dynamics via Integrated Information Decomposition},
  author={Mediano, Pedro AM and Rosas, Fernando E and Luppi, Andrea I and Carhart-Harris, Robin L and Bor, Daniel and Seth, Anil K and Barrett, Adam B},
  journal={Proceedings of the National Academy of Sciences},
  volume={122},
  number={39},
  pages={e2423297122},
  year={2025},
  publisher={National Academy of Sciences}
}

@article{krippendorff2009ross,
  title={Ross Ashby's information theory: a bit of history, some solutions to problems, and what we face today},
  author={Krippendorff, Klaus},
  journal={International journal of general systems},
  volume={38},
  number={2},
  pages={189--212},
  year={2009},
  publisher={Taylor \& Francis}
}

@article{zwick2004overview,
  title={An overview of reconstructability analysis},
  author={Zwick, Martin},
  journal={Kybernetes},
  volume={33},
  number={5/6},
  pages={877--905},
  year={2004},
  publisher={Emerald Group Publishing Limited}
}

@article{pica2017invariant,
  title={Invariant components of synergy, redundancy, and unique information among three variables},
  author={Pica, Giuseppe and Piasini, Eugenio and Chicharro, Daniel and Panzeri, Stefano},
  journal={Entropy},
  volume={19},
  number={9},
  pages={451},
  year={2017},
  publisher={MDPI}
}

@article{chicharro2017synergy,
  title={Synergy and redundancy in dual decompositions of mutual information gain and information loss},
  author={Chicharro, Daniel and Panzeri, Stefano},
  journal={Entropy},
  volume={19},
  number={2},
  pages={71},
  year={2017},
  publisher={MDPI}
}

@article{ay2021information,
  title={Information decomposition based on cooperative game theory},
  author={Ay, Nihat and Polani, Daniel and Virgo, Nathaniel},
  journal={Kybernetika},
  volume={56},
  number={5},
  pages={979--1014},
  year={2021}
}

@article{rauh2017secret,
  title={Secret sharing and shared information},
  author={Rauh, Johannes},
  journal={Entropy},
  volume={19},
  number={11},
  pages={601},
  year={2017},
  publisher={MDPI}
}

@article{perrone2016hierarchical,
  title={Hierarchical quantification of synergy in channels},
  author={Perrone, Paolo and Ay, Nihat},
  journal={Frontiers in Robotics and AI},
  volume={2},
  pages={35},
  year={2016},
  publisher={Frontiers Media SA}
}

@article{olbrich2015information,
  title={Information decomposition and synergy},
  author={Olbrich, Eckehard and Bertschinger, Nils and Rauh, Johannes},
  journal={Entropy},
  volume={17},
  number={5},
  pages={3501--3517},
  year={2015},
  publisher={MDPI}
}

@article{schneidman2003network,
  title={Network information and connected correlations},
  author={Schneidman, Elad and Still, Susanne and Berry, Michael J and Bialek, William},
  journal={Physical review letters},
  volume={91},
  number={23},
  pages={238701},
  year={2003},
  publisher={APS}
}

@article{chicharro2018identity,
  title={The identity of information: How deterministic dependencies constrain information synergy and redundancy},
  author={Chicharro, Daniel and Pica, Giuseppe and Panzeri, Stefano},
  journal={Entropy},
  volume={20},
  number={3},
  pages={169},
  year={2018},
  publisher={MDPI}
}

@misc{matthias2025novel,
      title={Novel Inconsistency Results for Partial Information Decomposition}, 
      author={Philip Hendrik Matthias and Abdullah Makkeh and Michael Wibral and Aaron J. Gutknecht},
      year={2026},
      journal={arXiv preprint arXiv:2512.16662},
      eprint={2512.16662},
      archivePrefix={arXiv},
      primaryClass={cs.IT},
      url={https://arxiv.org/abs/2512.16662}, 
}

@article{gutknecht2025babel,
  title={From Babel to Boole: the logical organization of information decompositions},
  author={Gutknecht, Aaron J and Makkeh, Abdullah and Wibral, Michael},
  journal={Proceedings of the Royal Society A},
  volume={481},
  number={2310},
  pages={20240174},
  year={2025},
  publisher={The Royal Society}
}

@book{pearl2014probabilistic,
  title={Probabilistic reasoning in intelligent systems: networks of plausible inference},
  author={Pearl, Judea},
  year={2014},
  publisher={Elsevier}
}

@article{jansma2025fast,
  title={Fast M{\"o}bius transform: An algebraic approach to information decomposition},
  author={Jansma, Abel and Mediano, Pedro AM and Rosas, Fernando E},
  journal={Physical Review Research},
  volume={7},
  number={3},
  pages={033049},
  year={2025},
  publisher={APS}
}

@incollection{rota1964foundations,
  title={On the foundations of combinatorial theory: I. Theory of M{\"o}bius functions},
  author={Rota, Gian-Carlo},
  booktitle={Classic Papers in Combinatorics},
  pages={332--360},
  year={1964},
  publisher={Springer}
}

@article{eppstein1995zonohedra,
  title={Zonohedra and zonotopes},
  author={Eppstein, David},
  year={1995}
}

@article{shannon1953lattice,
  title={The lattice theory of information},
  author={Shannon, Claude},
  journal={Transactions of the IRE professional Group on Information Theory},
  volume={1},
  number={1},
  pages={105--107},
  year={1953},
  publisher={IEEE}
}

@article{aumann2016agreeing,
  title={Agreeing to disagree},
  author={Aumann, Robert J},
  journal={Readings in Formal Epistemology},
  pages={859--862},
  year={2016},
  publisher={Springer}
}

@article{gacs1973common,
  title={Common information is far less than mutual information.},
  author={G{\'a}cs, Peter and Korner, Janos and others},
  journal={Problems of Control and Information Theory},
  volume={2},
  pages={149--162},
  year={1973},
  publisher={Elsevier Science Limited: Oxford Fulfillment Center, PO Box 800, Kidlington~…}
}

@inproceedings{de2008z3,
  title={Z3: An efficient SMT solver},
  author={De Moura, Leonardo and Bj{\o}rner, Nikolaj},
  booktitle={International conference on Tools and Algorithms for the Construction and Analysis of Systems},
  pages={337--340},
  year={2008},
  organization={Springer}
}

@article{faes2025predictive,
  title={Predictive information decomposition as a tool to quantify emergent dynamical behaviors in physiological networks},
  author={Faes, Luca and Mijatovic, Gorana and Sparacino, Laura and Porta, Alberto},
  journal={IEEE Transactions on Biomedical Engineering},
  year={2025},
  publisher={IEEE}
}

@phdthesis{williams2011information,
  title={Information dynamics: Its theory and application to embodied cognitive systems},
  author={Williams, Paul L},
  year={2011},
  school={Indiana University}
}

@article{rauh2021properties,
  title={Properties of unique information},
  author={Rauh, Johannes and Sch{\"u}nemann, Maik and Jost, J{\"u}rgen},
  journal={Kybernetika},
  volume={57},
  number={3},
  pages={383--403},
  year={2021},
  publisher={Institute of Information Theory and Automation AS CR}
}

@article{walpole2013multiscale,
  title={Multiscale computational models of complex biological systems},
  author={Walpole, Joseph and Papin, Jason A and Peirce, Shayn M},
  journal={Annual review of biomedical engineering},
  volume={15},
  number={1},
  pages={137--154},
  year={2013},
  publisher={Annual Reviews}
}

@article{dada2011multi,
  title={Multi-scale modelling and simulation in systems biology},
  author={Dada, Joseph O and Mendes, Pedro},
  journal={Integrative Biology},
  volume={3},
  number={2},
  pages={86--96},
  year={2011},
  publisher={Oxford University Press}
}

@article{rajpal2025synergistic,
  title={Synergistic small worlds that drive technological sophistication},
  author={Rajpal, Hardik and Guerrero, Omar},
  journal={PNAS nexus},
  volume={4},
  number={4},
  pages={pgaf102},
  year={2025},
  publisher={Oxford University Press US}
}

@incollection{dilworth1990decomposition,
  title={A decomposition theorem for partially ordered sets},
  author={Dilworth, Robert P},
  booktitle={The Dilworth Theorems: Selected Papers of Robert P. Dilworth},
  pages={7--12},
  year={1990},
  publisher={Springer}
}

@article{reing2021discovering,
  title={Discovering higher-order interactions through Neural Information Decomposition},
  author={Reing, Kyle and Ver Steeg, Greg and Galstyan, Aram},
  journal={Entropy},
  volume={23},
  number={1},
  pages={79},
  year={2021},
  publisher={MDPI}
}

@article{gutknecht2021bits,
  title={Bits and pieces: Understanding information decomposition from part-whole relationships and formal logic},
  author={Gutknecht, Aaron J and Wibral, Michael and Makkeh, Abdullah},
  journal={Proceedings of the Royal Society A},
  volume={477},
  number={2251},
  pages={20210110},
  year={2021},
  publisher={The Royal Society Publishing}
}

\end{document}